\documentclass[11pt,a4paper]{article}

\usepackage{graphicx}
\graphicspath{{figures/}}
\usepackage[english]{babel}
\usepackage{latexsym}
\usepackage{url}
\usepackage{amsmath}
\usepackage{amsthm}
\usepackage{dsfont}
\usepackage{ifthen}
\usepackage{fancybox}
\usepackage{microtype,stmaryrd}

\usepackage[charter]{mathdesign}

\usepackage{color}
\definecolor{blu3}{rgb}{.1,.0,.4}
\usepackage{hyperref}
\hypersetup{colorlinks=true, linkcolor=blu3, urlcolor=blu3, citecolor=blu3, pdfpagemode=UseNone, pdfstartview=} %colored links, no rectangles, no bookmarks, zoom

\usepackage[margin=1.1in]{geometry}

\newtheorem{theorem}{Theorem}
\newtheorem{corollary}[theorem]{Corollary}
\newtheorem{lemma}[theorem]{Lemma}

\newtheorem{claim}{Claim}[theorem]

\newcommand{\RR}{\ensuremath{\mathbb R}}  % real numbers
\newcommand{\VV}{\ensuremath{\mathbb V}}  % Voronoi diagram
\newcommand{\UU}{\ensuremath{\mathbb U}}  % Voronoi diagram
\newcommand{\II}{\ensuremath{\mathbb I}}  % family of intervals
\DeclareMathOperator{\cell}{cell}
\DeclareMathOperator{\res}{res}
\DeclareMathOperator{\parent}{pa}

\def\DEF#1{\textbf{\emph{#1}}}
\def\eps{\varepsilon}

\def\TREEcopy{\textsf{Copy}\relax}
\def\TREEreport{\textsf{Report}\relax}
\def\TREEinsert{\textsf{Insert}\relax}
\def\TREEdelete{\textsf{Delete}\relax}
\def\TREEsucc{\textsf{Succ}\relax}
\def\TREEpred{\textsf{Pred}\relax}
\def\TREEsplit{\textsf{Split}\relax}
\def\TREEjoin{\textsf{Join}\relax}
\def\TREEshift{\textsf{Shift}\relax}
\def\TREEdiff{\textsf{diff-val}\relax}
\def\INTlength{\textsf{diff-len}\relax}

\def\INTcopy{\textsf{IntCopy}\relax}
\def\INTreport{\textsf{IntReport}\relax}
\def\INTinsert{\textsf{IntInsert}\relax}
\def\INTdelete{\textsf{IntDelete}\relax}
\def\INThitby{\textsf{IntHitBy}\relax}
\def\INTcontaining{\textsf{IntContaining}\relax}
\def\INTclip{\textsf{IntClip}\relax}
\def\INTjoin{\textsf{IntJoin}\relax}
\def\INTshift{\textsf{IntShift}\relax}
\def\INTextend{\textsf{IntExtend}\relax}

\begin{document}

\title{The Inverse Voronoi Problem in Graphs}

\author{{\'E}douard Bonnet\thanks{Univ Lyon, CNRS, ENS de Lyon, Universit\'e Claude Bernard Lyon 1, LIP UMR5668, France. Email address: \texttt{edouard.bonnet@ens-lyon.fr}. Supported by the LABEX MILYON (ANR-10- LABX-0070) of Universit\'e de Lyon, within the program "Investissements d'Avenir" (ANR-11-IDEX-0007) operated by the French National Research Agency (ANR).}
\and
	Sergio Cabello\thanks{Faculty of Mathematics and Physics, University of Ljubljana, 
			and IMFM, Slovenia.
            Supported by the Slovenian Research Agency, program P1-0297 and projects J1-8130, J1-8155, 	J1-9109, J1-1693, J1-2452.
			Email address: \texttt{sergio.cabello@fmf.uni-lj.si}.}
\and
	Bojan Mohar\thanks{Department of Mathematics, Simon Fraser University, Burnaby, BC, Canada. Email address: \texttt{mohar@sfu.ca}. On leave from IMFM \& FMF, Department of Mathematics, University of Ljubljana. Supported in part by the NSERC Discovery Grant R611450 (Canada), by the Canada Research Chairs program, and by the Research Projects J1-8130 and J1-2452 of ARRS (Slovenia).}
\and
	Hebert P\'erez-Ros\'es\thanks{Departament d'Enginyeria Inform\`atica i Matem\`atiques, Universitat Rovira i Virgili, Tarragona, Spain.
	Partially supported by Grant MTM2017-86767-R from the Spanish Ministry of Economy, Industry and Competitiveness.}
}

\maketitle

\begin{abstract}
	We introduce the inverse Voronoi diagram problem in graphs: given 
	a graph $G$ with positive edge-lengths and a collection $\mathbb{U}$ 
	of subsets of vertices of $V(G)$, decide whether $\mathbb{U}$ is a
	Voronoi diagram in $G$ with respect to the shortest-path metric.
	We show that the problem is NP-hard, even for planar graphs where
	all the edges have unit length.
	We also study the parameterized complexity of the problem and show
	that the problem is W[1]-hard when parameterized by the number of Voronoi
	cells or by the pathwidth of the graph.
	For trees we show that the problem can be solved in $O(N+n \log^2 n)$ time,
	where $n$ is the number of vertices in the tree 
	and $N=n+\sum_{U\in \mathbb{U}}|U|$ is the size of the description 
	of the input.
	We also provide a lower bound of $\Omega(n \log n)$ time for trees with $n$ vertices.		
	
    \medskip
    \textbf{Keywords:} distances in graphs, Voronoi diagram, inverse Voronoi problem,
		NP-complete, parameterized complexity, trees, 
		applications of binary search trees, dynamic programming in trees, lower bounds.
\end{abstract}

%%%%%%%%%%%%%%%%%%%%%%%%%%%%%%%%%%%%%%%%%%%%%%%%%%%%%%%%%%%%%%%%%%%%%%%%%%%%%%%%%%%%%%%%%%%%%%%%%%%%%%%%%%%%%%%%%%%%%%

\section{Introduction}

Let $(X,d)$ be a metric space, where $d\colon X\times X \rightarrow \RR_{\ge 0}$.
Let $\Sigma$ be a subset of $X$. We refer to each element of $\Sigma$ as a \DEF{site},
to distinguish it from an arbitrary point of $X$.
The \DEF{Voronoi cell} of each site $s\in \Sigma$ is then defined by
\[
	\cell_{(X,d)}(s,\Sigma) ~=~ \{ x\in X\mid \forall s'\in \Sigma: d(s,x)\le d(s',x) \}.
\]
The \DEF{Voronoi diagram} of $\Sigma$ in $(X,d)$ is
\[
	\VV_{(X,d)}(\Sigma) ~=~ \{ \cell_{(X,d)}(s,\Sigma) \mid s\in \Sigma \}.
\]
It is easy to see that, for each set $\Sigma$ of sites,
each element of $X$ belongs to some Voronoi cell $\cell_{(X,d)}(s,\Sigma)$.
Therefore, the sets in $\VV_{(X,d)}(\Sigma)$ cover $X$. On the other hand, the
Voronoi cells do not need to be pairwise disjoint. In particular,
when some point $x\in X$ is closest to two sites, then it
is in both Voronoi cells.

In the \DEF{inverse Voronoi problem},
we are given a metric space $(X,d)$ and a sequence $X_1,\dots, X_k$
of subsets of $X$ that cover $X$.
The task it to decide whether $\{ X_1,\dots, X_k\}$
is a Voronoi diagram in $(X,d)$.
This means that we have to decide whether there exist some sites $s_1,\dots ,s_k$ such that, for each index $i$,
we have $X_i=\cell_{(X,d)}(s_i,\{s_1,\dots ,s_k\})$.

The inverse Voronoi problem is closely related to problems in classification and clustering.
In pattern recognition, a classic paradigm to classify is to use the nearest neighbor rule:
given a learning set $\mathcal{L}$ of objects that are already classified, 
each new object is classified into the same class as its closest object from $\mathcal{L}$.
To reduce the size of the learning set, Hart~\cite{Hart68} introduced 
the concept of \emph{consistent subsets}. 
A subset $\mathcal{L}'$ of the learning set $\mathcal{L}$
is a consistent subset if, for each object $\ell$ from $\mathcal{L}$, 
the object $\ell$ and its closest neighbor in $\mathcal{L}'$ are in the same class.
An equivalent, alternative perspective of this is given by Voronoi diagrams: 
in the Voronoi diagram of a consistent subset $\mathcal{L}'$, each
object $\ell$ of $\mathcal{L}$ belongs to a Voronoi cell 
defined by a site $s\in \mathcal{L}'$ if and only if $\ell$ and $s$ belong to the same class.
Ritter et al.~\cite{RitterWLI75} introduced the problem of finding consistent
subsets of minimum size. Surveying the research in this applied area is beyond the scope of our research.
We refer to Biniaz et al.~\cite{abs-1810-09232} and Gottlieb et al.~\cite{GottliebKN18} 
for some of the latest algorithmic results on this topic.
Considering each class as a Voronoi cell,
the inverse Voronoi problem is asking precisely whether there exists
a consistent subset with one element per class. Such consistent subset has of 
course to be of optimal size.

\paragraph{Graphic version.}
Let $G$ be an undirected graph with $n$ vertices and
abstract, positive edge-lengths $\lambda\colon E(G)\rightarrow \RR_{>0}$.
The length of a path in $G$ is the sum of the edge-lengths along the path.
We define the (shortest-path) \DEF{distance} between two vertices $x$ and $y$ of $G$, denoted
by $d_G(x,y)$, as the minimum length over all paths in $G$ from $x$ to $y$.

Since $(V(G),d_G)$ is a metric space, we can consider the concepts of Voronoi cells and
Voronoi diagrams for this space. We denote them by $\cell_G(s,\Sigma)$ and $\VV_G(\Sigma)$ respectively.
Moreover, when the graph is clear from the context, we remove the subscript and thus just talk about $\cell(s,\Sigma)$ and $\VV(\Sigma)$.

In this paper we consider computational aspects
of the inverse Voronoi problem when the metric space
is the shortest-path metric in a graph.
This means that we are given a collection of candidate Voronoi cells in a graph, 
and we would like to decide whether they form indeed a Voronoi diagram. 
Let us describe the problem more formally.

\begin{quote}
	\textsc{Graphic Inverse Voronoi}\\
	Input: $(G,\UU)$, where $G$ is a graph with positive edge-lengths 
        and $\UU=( U_1,\dots, U_k)$ is a sequence
		of subsets of vertices of $G$ that cover $V(G)$.\\
	Question: Are there sites $s_1,\dots,s_k\in V(G)$ such that $\cell_G(s_i,\{s_1,\dots ,s_k\})=U_i$ for each $i\in \{1,\dots, k\}$?
	When the answer is positive, provide a solution: 
	sites $s_1,\dots,s_k\in V(G)$ that certify the positive answer.
\end{quote}

See Figure~\ref{fig:disjoint} for an example.
As far as the existence of polynomial-time algorithms is concerned,
it is equivalent to consider a graph or a finite metric space.
Indeed, for each finite metric space we can build a graph that
encodes those distances by using a complete graph with edge-lengths, and,
inversely, given a graph, we can compute the matrix of distances between all pairs
of vertices in polynomial time.
However, considering special classes of graphs may be useful to get more efficient algorithms.

\begin{figure}
\centering
	\includegraphics[page=1]{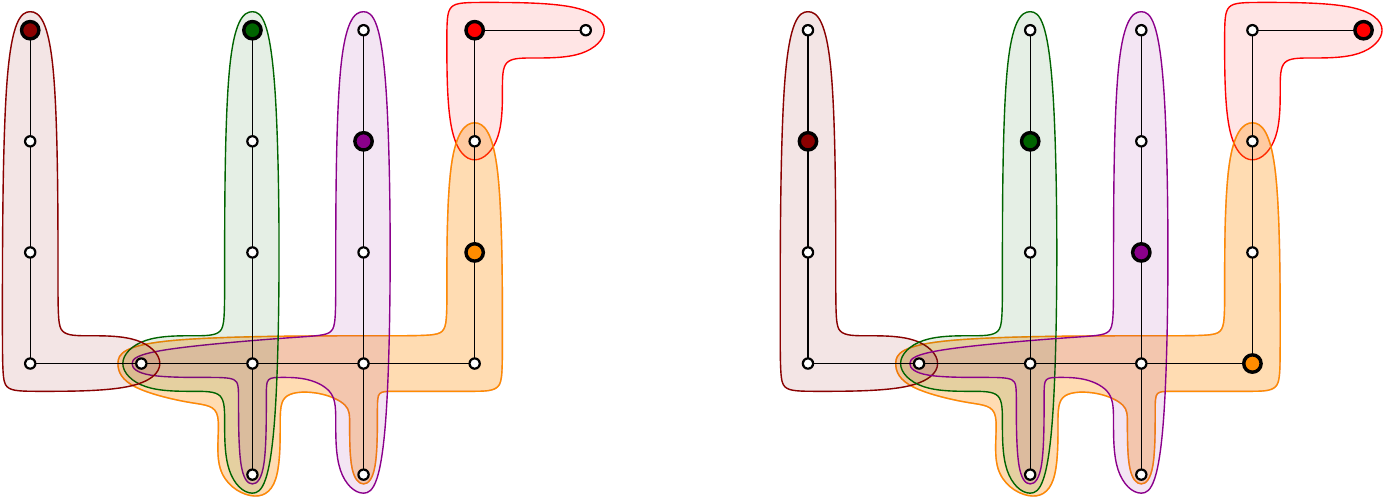}
	\caption{An instance with two solutions.
		The edges have unit length and the larger, filled dots represent the sites.}
	\label{fig:disjoint}
\end{figure}

\paragraph{Our results.}
First we show that the problem \textsc{Graphic Inverse Voronoi} is NP-hard even for planar
graphs where the candidate Voronoi cells are pairwise disjoint and each has at most $3$ vertices. 
The reduction is from a variant of \textsc{Planar $3$-SAT}. 
The bound on the number of vertices per cells is tight: 
when each candidate Voronoi cell has $2$ vertices, the problem can be solved using $2$-SAT.

Many graph decision and optimization problems admit fixed-parameter tractable (FPT) 
algorithms with respect to additional parameters that quantify how complex is the input; 
see for instance \cite{CyganFKLMPPS15}. 
Using the framework of parameterized complexity, we provide stronger lower bounds
when parameterized by the number $k$ of sites and the pathwidth $p(G)$ of $G$. 
More precisely, assuming the Exponential Time Hypothesis (ETH), 
we show that the problem cannot be solved in time $f(k)|V(G)|^{o(k/ \log k)}$
nor in time $f(p(G))|V(G)|^{o(p(G))}$, 
for any computable function $f$.
These hardness results hold for graphs where all the edges have unit length.

Then we consider efficient algorithms for the problem \textsc{Graphic Inverse Voronoi}
when the underlying graph is a \emph{tree}. We refer to this variant of the problem as
\textsc{Graphic Inverse Voronoi in Trees}.

One has to be careful with the size of the description of the input because the size
of the Voronoi diagram may be quadratic in the size of the tree. For example,
in a star with $2n$ leaves and sites in $n$ of the leaves, each Voronoi cell
has size $\Theta(n)$, and thus an explicit description of the Voronoi diagram has size $\Theta(n^2)$.
Motivated by this, we define the \DEF{description size} of
an instance $I=(T,(U_1,\dots,U_k ))$ for the \textsc{Graphic Inverse Voronoi in Trees}
to be $N=N(I)=|V(T)|+\sum_{i}|U_i|$. 
We use $n$ for the number of vertices in the tree $T$, which is potentially smaller
than $N$.

We show that the problem \textsc{Graphic Inverse Voronoi in Trees} 
can be solved in $O(N + n \log^2 n)$ time for arbitrary trees. 
We also show a lower bound of $\Omega(n\log n)$
in the algebraic computation tree model.

One may be tempted to think that the problem is easy for trees.
Our near-linear time algorithm for arbitrary trees is far from trivial. Of course we cannot exclude
the existence of a simpler algorithm running in near-linear time, but we do think that the
problem is more complex than it may seem at first glance. Figure~\ref{fig:disjoint}
may help understanding that the interaction between different Voronoi cells may
be more complex than it seems, even for trees.

To obtain our algorithm for trees, we consider
the following more general problem, where the input also
specifies, for each Voronoi cell, a subset of vertices where the site has to be placed.
\begin{quote}
	\textsc{Generalized Graphic Inverse Voronoi in Trees}\\
	Input: $(T,\UU)$, where $T$ is a tree with positive edge-lengths
		and $\UU=\bigl( (U_1,S_1),\dots, (U_k,S_k)\bigr)$
		is a sequence of pairs of subsets of vertices of $G$ such that
		$U_1,\dots,U_k$ cover $V(G)$. \\
	Question: are there sites $s_1,\dots, s_k\in V(T)$ such that $s_i\in S_i$ and
		$U_i=\cell_T(s_i,\{ s_1,\dots,s_k\})$ for each $i\in \{1,\dots, k\}$?
		When the answer is positive, provide a solution: 
		sites $s_1,\dots,s_k\in V(T)$ that certify the positive answer.
\end{quote}
Following the analogy with \textsc{Graphic Inverse Voronoi in Trees},
we define the \DEF{description size} of an instance $I=(T,( (U_1,S_1),\dots, (U_k,S_k)))$
to be $N(I)=|V(T)|+\sum_i |U_i|+\sum_i |S_i|$.

Clearly, the problem \textsc{Graphic Inverse Voronoi in Trees} can be reduced to the problem
\textsc{Generalized Graphic Inverse Voronoi in Trees} by taking $S_i=U_i$ for all $i\in \{1,\dots, k\}$.
This transformation can be done in linear time (in the size of the instance).
Thus, for the algorithm it suffices to provide an algorithm for
the problem \textsc{Generalized Graphic Inverse Voronoi in Trees}. 
(The lower bound holds for the original problem.)

In our solution we first make a reduction to the same problem in which Voronoi cells are disjoint,
and then we make another transformation to an instance having maximum degree $3$.
Finally, we employ a bottom-up dynamic programming procedure that,
to achieve near-linear time, merges the information from the subproblems 
in time almost proportional to the smallest of the subproblems. For this, we
employ dynamic binary search trees to manipulate sets of intervals.

\paragraph{Related work.}
Voronoi diagrams on graphs were first investigated by Erwig~\cite{Erwig00}, who showed that they can be efficiently computed. Subsequently, graph Voronoi diagrams have been used in a variety of applications. For instance, Okabe and Sugihara~\cite{OkabeS12} describe several applications of graph Voronoi diagrams. More recent applications, many of them for planar graphs, can be found in~\cite{Cabello19,CdV10,GawrychowskiKMS18,GawrychowskiMWW18,KannanMZ18,MarxP15}.
Voronoi diagrams in graphs have also been considered in the context of the so-called Voronoi game~\cite{BandyapadhyayBDS15,GerbnerMPPR14} and in the context of topological data analysis~\cite{DeyFW15}.

On the other hand, the inverse Voronoi problem in the traditional, Euclidean setting has been studied since the mid 1980s, starting with the seminal paper by Ash and Bolker~\cite{Ash85}. We are not aware of any previous work considering the graphic inverse Voronoi problem.

\paragraph{Roadmap.}
In Section~\ref{sec:basic} we provide notation and some basic tools.
In Section~\ref{sec:SAT} we show NP-hardness of the problem 
\textsc{Graphic Inverse Voronoi}, while in 
Sections~\ref{sec:numberofcells} and~\ref{sec:treewidth} we consider the parameterized versions of the problem.
Then we move to the particular case of trees.
More precisely, in Section~\ref{sec:transform} we show how to reduce the 
problem \textsc{Generalized Graphic Inverse Voronoi in Trees} to a special
instance where the candidate Voronoi cells are disjoint and the tree has maximum
degree $3$.
In Section~\ref{sec:algorithmDP} we describe how to solve the problem 
\textsc{Generalized Graphic Inverse Voronoi in Trees}, 
after the transformation, using dynamic programming.
In Section~\ref{sec:lowerbound} we provide a lower bound for the problem 
\textsc{Graphic Inverse Voronoi in Trees}.
We conclude stating a few open problems.

%%%%%%%%%%%%%%%%%%%%%%%%%%%%%%%%%%%%%%%%%%%%%%%%%%%%%%%%%%%%%%%%%%%%%%%%%%%%%%%%%%%%%%%%%%%%
\section{Basics}
\label{sec:basic}

\paragraph{Sets and Graphs.}
For a positive integer $k$ we use the notation $[k]=\{ 1,\dots, k\}$.

We use the standard graph-theoretic definitions and notations that can be found in Diestel's book \cite{Diestel12}.
In particular, we denote by $V(G)$, respectively by $E(G)$, the vertex-set, respectively the edge-set, of a graph $G$.
If $G$ is a graph and $S \subseteq V(G)$, we denote by $G[S]$ the subgraph induced by $S$, and $G - S$ is a short-hand for $G[V(G) \setminus S]$. 
A graph is said \emph{planar} if its vertices can be drawn as distinct points of the real plane, 
its edges can be drawn as simple curves (in the plane) connecting the points that represent its vertices,
and the interior of the curves representing the edges are pairwise non-crossing.

A \emph{vertex-separator} in a graph $G$ is a subset of its vertices $S \subseteq V(G)$ such that $G - S$ is a disconnected non-empty graph.
A vertex-separator $S$ is said \emph{balanced} if all the connected components of $G-S$ have size at most $2|V(G)|/3$.
Up to constant factors, we could define the notion of treewidth by means of repeated balanced vertex-separators.
We choose not to do so, in order to follow the usual definition and to also introduce pathwidth.
We will need the notion of balanced vertex-separators in Theorem~\ref{thm:PGIV-XPtw} and treewidth/pathwidth in Section~\ref{sec:treewidth}.

\paragraph{Treewidth and pathwidth.}
A \emph{tree-decomposition} of a graph $G$ is a tree $T$ whose nodes are labeled by subsets of $V(G)$, called \emph{bags}, such that for each vertex $v \in V(G)$, the bags containing $v$ induce a subtree of $T$, and for each edge $e \in E(G)$, there is at least one bag containing both endpoints of $e$.
The width of a tree-decomposition is the size of its largest bag minus one.
The treewidth of a graph $G$ is the minimum width of a tree-decomposition of $G$.
The point of the ``minus one'' in the definition is that the tree-width of trees is 1, and not 2.
The pathwidth is the same as treewidth except the tree $T$ is now required to be a path, and hence is called a path-decomposition.
In particular pathwidth is always at least as large as treewidth.
We observe that if $G$ has treewidth $w$, then it admits in particular a balanced vertex-separator of size at most $w+1$.
Indeed, any non-leaf bag is a vertex-separator.
One can find one that is also balanced by starting from the root (any node) of the tree-decomposition labeled by say, $S$, and iteratively moving to a component of $G-S$ which is larger than $2/3$ of the whole graph.
When this process is no longer possible, we have our balanced vertex-separator.

In Section~\ref{sec:treewidth}, we will need to bound the pathwidth (and therefore the treewidth) of rather complicated graphs.
Writing down the full description of a tree-decomposition or of a path-decomposition may be a bit tedious.
Kirousis and Papadimitriou \cite{KirousisP85} showed the equality between the interval thickness number, which is equal to pathwidth plus one, and the \emph{node searching number}.
To give an upper bound to the pathwidth, we only need to prove that the number of cleaners required to win the following one-player game is bounded by a suitable function.
We imagine the edges of a graph to be contaminated by a gas.
We shall move around a team of cleaners, placed at the vertices, in order to clean all the edges.
A move consists of removing a cleaner from the graph, adding a cleaner at an unoccupied vertex, or displacing a cleaner from a vertex to any other vertex (not necessarily adjacent).
An edge is cleaned when both its endpoints are occupied by a cleaner.
After each move, all the cleaned edges admitting a free-of-cleaners path from one of its endpoints to the endpoint of a contaminated edge are however recontaminated.
The node searching number is the minimum number of cleaners required to win the game.

\paragraph{ETH and Sparsification Lemma.}
The Exponential-Time Hypothesis (ETH, for short) of Impagliazzo and Paturi~\cite{ImpagliazzoETH} asserts that there is no subexponential-time algorithm solving \textsc{3-SAT}.
More precisely, there is a positive real number $\delta > 0$ such that \textsc{3-SAT} cannot be solved in time $2^{\delta n}$ on $n$-variable instances.
Impagliazzo et al.~\cite{Impagliazzo01} present a subexponential-time Turing-reduction parameterized by a positive real number $\varepsilon > 0$ which, given a \textsc{3-SAT}-instance $\phi$ with $n$ variables and $m$ clauses, produces at most $2^{\varepsilon n}$ \textsc{3-SAT}-instances $\phi_1, \ldots, \phi_t$ such that $\phi \Leftrightarrow \bigvee_{i \in [t]} \phi_i$, each $\phi_i$ having no more than $n$ variables and $C_\varepsilon n$ clauses for some constant $C_\varepsilon$ (depending solely on $\varepsilon$, and \emph{not} on $n$ and $m$).
This important reduction is known as the Sparsification Lemma.
Due to the Sparsification Lemma, there exists a positive real number $\delta' > 0$ such that no algorithm can solve \textsc{3-SAT} in time $2^{\delta'(n+m)}$ on $n$-variable $m$-clause instances, assuming that the ETH holds.

\paragraph{Inverse Voronoi and compatibility.}
Consider an instance $(G,(U_1,\dots, U_k))$ to the \textsc{Graphic Inverse Voronoi}
and a candidate solution $s_1,\dots,s_k\in V(G)$.
We say that $s_i$ and $s_j$ ($i\neq j$) are \DEF{compatible} if we have 
$d(s_i,u)= d(s_j,u)$ for each $u\in U_i\cap U_j$,
$d(s_i,u)< d(s_j,u)$ for each $u\in U_i\setminus U_j$, and
$d(s_j,u)< d(s_i,u)$ for each $u\in U_j\setminus U_i$.
Consider a fixed index $i\in [k]$.
It is straightforward from the definition 
that $\cell_G(s_i,\{s_1,\dots ,s_k\})=U_i$ if and only if  
$s_i$ and $s_j$ are compatible for all $j\neq i$. (Here the relevant assumption is that $U_1\cup \dots\cup U_k$ is $V(G)$.)

In all cases we use $G$ as the ground graph that defines the metric.
Note that in the following claims it is important that $G$ has positive edge-lengths.

We have remarked before that Voronoi cells need not be disjoint.
A vertex belongs to various Voronoi cells if
it is equidistant to different sites.
An alternative is to define cells using strict inequalities. More precisely,
for a set $\Sigma$ of sites, the \DEF{open Voronoi cell} of each site $s\in \Sigma$ is then defined by
\[
	\cell^<(s,\Sigma) ~=~ \{ x\in X\mid \forall s'\in \Sigma\setminus\{s\}: d(s,x)< d(s',x) \}.
\]
In this case, the cells are disjoint but they do not necessarily form a partition of $X$.
The following lemma is straightforward and we omit its proof.

\begin{lemma}
\label{le:strict}
	For each set $\Sigma$ of sites and each site $s\in \Sigma$ we have $s\in \cell^<(s,\Sigma)$ and
$$
   \cell^<(s,\Sigma) = \cell(s,\Sigma) \setminus \left(\bigcup_{s'\neq s} \cell(s',\Sigma)\right) .
$$
\end{lemma}

\begin{lemma}
\label{le:starshaped}
	For each set $\Sigma$ of sites, each site $s\in \Sigma$, and each vertex $v\in \cell(s,\Sigma)$,
	the path in $T$ from $s$ to $v$ is contained in $T[\cell(s,\Sigma)]$, the subgraph of $T$ induced
	by $\cell(s,\Sigma)$. The same statement is true for $\cell^<(s,\Sigma)$.
\end{lemma}

A consequence of this Lemma is that the shortest path from $s$ to $v\in \cell(s,\Sigma)\setminus \cell^<(s,\Sigma)$ has
a part with vertices inside $\cell^<(s,\Sigma)$ followed by a part with vertices of $\cell(s,\Sigma)\setminus \cell^<(s,\Sigma)$.

\begin{lemma}
\label{le:check}
	Given an instance for the problem \textsc{Graphic Inverse Voronoi in Trees} or
	the \textsc{Generalized Graphic Inverse Voronoi in Trees},
	and a candidate solution $s_1,\dots,s_k$, we can check in $O(N)$ time
	whether $s_1,\dots,s_k$ is indeed a solution.
\end{lemma}

\begin{proof}
	Let $T$ be the underlying tree defining the instance.
	We add a new vertex $a$ (called the \emph{apex}) to $T$ 
	and connect it to each candidate site $s_1,\dots,s_k$ with
	edges of the same positive length. See the left drawing in Figure~\ref{fig:check}.
	The resulting graph, denoted by $T_a$, has treewidth $2$, and thus we can
	compute shortest paths from $a$ to all vertices in linear time~\cite{ChaudhuriZ00}.
	Let $d_a[v]$ be the distance in $T_a$ from $a$ to $v$.

	\begin{figure}
	\centering
		\includegraphics[page=7]{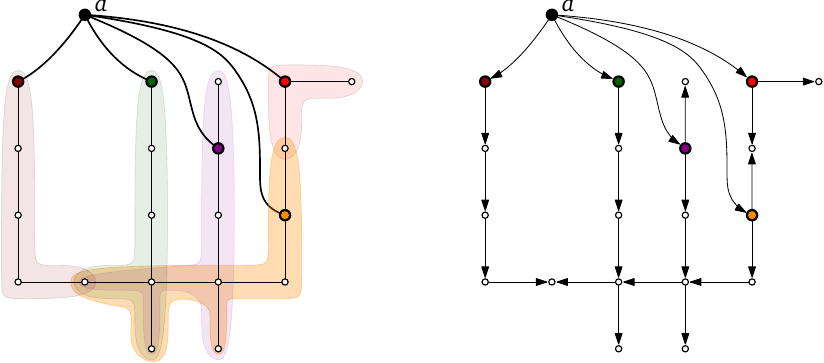}
		\caption{Construction of $T_a$ (left) and the directed acyclic graph $D_a$ (right).}
		\label{fig:check}
	\end{figure}

	Next we build a digraph $D_a$ describing the shortest paths from $a$ to all other vertices.
	The vertex set of $D_a$ is $V(T)\cup \{a\}=V(T_a)$.
	For each arc $u\rightarrow v$, where $uv\in E(T_a)$,
	we add $u\rightarrow v$ to $D_a$ if and only if $d_a[v]=d_a[u]+\lambda(uv)$.
	With this we obtain a directed acyclic graph $D_a$ that contains \emph{all} 
	shortest paths from $a$ to every $v\in V(T)$ and,
	moreover, each directed path in $D_a$ is indeed a shortest path in $T_a$.
	See Figure~\ref{fig:check} right.

	Now we label each vertex $v$ with the indices $i$ of those sites $s_i$, whose Voronoi cells contain $v$, as follows.
	We start setting $L(s_i)=\{i\}$ for each site $s_i$.
	Then we consider the vertices $v\in V(T)$ in topological order with respect to $D_a$.
	For each vertex $v$, we set $L(v)$ to be the union of $L(u)$, where $u$ iterates
	over the vertices of $V(T)$ with arcs in $D$ pointing to $v$.
	It is easy to see by induction that $L(v)=\{ i\in [k]\mid v\in \cell_T(s_i,\{ s_1,\dots,s_k\}) \}$.
	During the process we keep a counter for $\sum_v |L(v)|$, and if at some moment we detect that
	the counter exceeds $N$, we stop and report that $s_1,\dots,s_k$ is not a solution.
	Otherwise, we finish the process when we computed the sets $L(v)$.

	Now we compute the sets $V_i= \{ v\in V(T)\mid s_i\in L(v)\}$ for $i=1,\dots k$.
	This is done iterating over the vertices $v\in V(T)$ and adding $v$ to each site of $L(v)$.
	This takes $O(N+\sum_v |L(v)|)=O(N)$ time.
	Note that $V_i=\cell_T(s_i,\{ s_1,\dots,s_k \})$.
	It remains to check that $U_i=V_i$ for all $i\in [k]$.
	For this we add flags to $V(T)$ that are initially set to false.
	Then, for each $i\in [k]$, we do the following:
	check that $|U_i|=|V_i|$,
	iterate over the vertices of $U_i$ setting the flags to true,
	iterate over the vertices of $V_i$ checking that the flags are true,
	iterate over the vertices of $U_i$ setting the flags back to false.
	The procedure takes $O(N+\sum_v |L(v)|)=O(N)$ time and, if all the checks
	were correct, we have $U_i=V_i=\cell_T(s_i,\{ s_1,\dots,s_k \})$ for all $i\in [k]$.
\end{proof}

%%%%%%%%%%%%%%%%%%%%%%%%%%%%%%%%%%%%%%%%%%%%%%%%%%%%%%%%%%%%%%%%%%%%%%%%%%%%%%%%%%%%%%%%%%%%
\section{Hardness of the Graphic Inverse Voronoi}
\label{sec:SAT}

In this Section we show that the problem \textsc{Graphic Inverse Voronoi} is NP-hard, even for
planar graphs. Stronger lower bounds are derived assuming the Exponential Time Hypothesis (ETH).
We will make a reduction from a variant of the satisfiability (SAT) where each clause has $3$ literals,
all the literals are positive, and we want that each clause is satisfied at exactly one literal.
The problem can be stated combinatorially as follows.

\begin{quote}
	\textsc{Positive 1-in-3-SAT}\\
	Input: $(\mathcal V, \mathcal C)$, where $\mathcal V$ is a ground set
		and $\mathcal C$ is a family of subsets of $\mathcal V$ of size 3.\\
	Question: Is there a subset $T \subseteq \mathcal V$ such that $|C \cap T|=1$ for each $C\in \mathcal C$?
\end{quote}
In this combinatorial setting, $\mathcal V$ represents the variables, $\mathcal C$ represents
the clauses with $3$ positive literals each, and $T$ represents the variables that are set to true.

The incidence graph $I(\mathcal V,\mathcal C)$ of an instance $(\mathcal V, \mathcal C)$
has vertex set $\mathcal V \cup \mathcal C$ and an edge between $v\in \mathcal V$
and $C\in \mathcal C$ precisely when $v\in C$. The graph is bipartite.

As shown by Mulzer and Rote~\cite{Mulzer08},
the problem \textsc{Positive 1-in-3-SAT} is NP-complete 
even when the incidence graph is planar. 

\begin{theorem}
  \label{thm:PGIVhardness}
	The \textsc{Graphic Inverse Voronoi} problem is NP-hard on planar graphs with unit edge-lengths, 
	even when the candidate Voronoi cells are disjoint sets of size at most 3. 
\end{theorem}

\begin{proof}
  We make a reduction from \textsc{Positive 1-in-3-SAT} with planar incidence graphs.
  Let $(\mathcal V=\{x_1,\ldots,x_n\},\mathcal C=\{C_1,\ldots,C_m\})$ 
  be an instance of \textsc{Positive 1-in-3-SAT} with planar incidence graph.
  We produce an equivalent instance $(G,\UU)$ of \textsc{Graphic Inverse Voronoi} as follows.
  See Figure~\ref{fig:reductionSAT1}.

  \begin{figure}
	  \centering
	  \includegraphics[page=1,width=\textwidth]{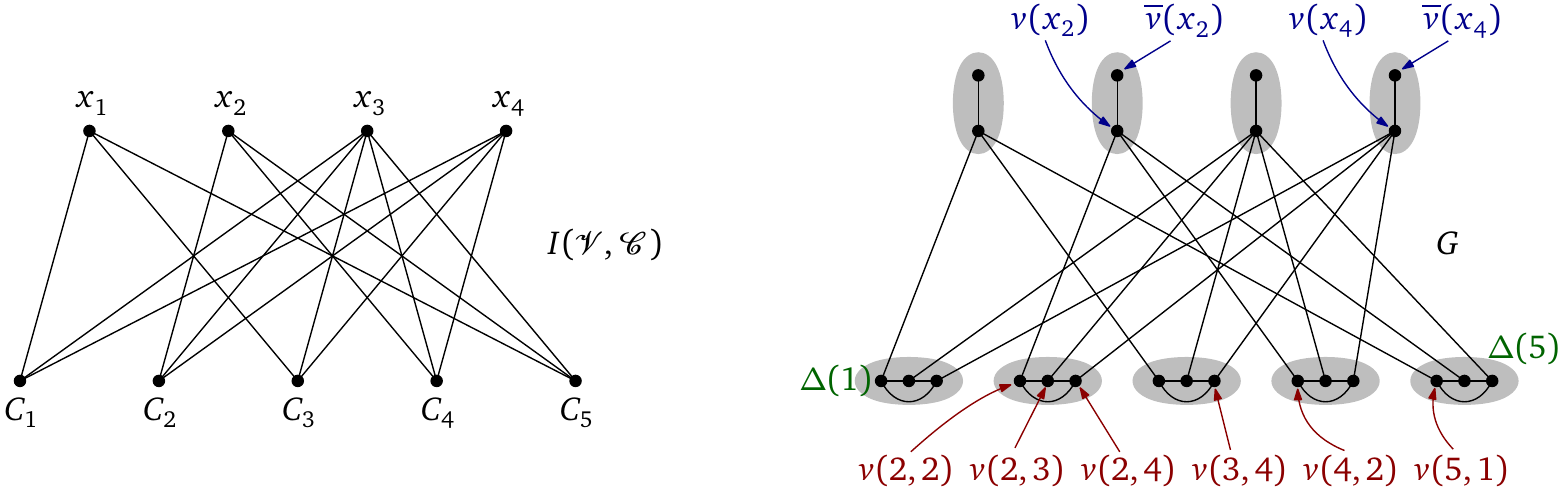}
	  \caption{Left: incidence graph for the \textsc{Positive 1-in-3-SAT} instance with
		$\mathcal V=\{ x_1,x_2,\dots,x_5\}$ and $\mathcal C=\{ C_1=\{x_1,x_3, x_4\}, 
			C_2= \{ x_2,x_3,x_4 \}, \dots, C_5=\{ x_1,x_2,x_3\} \}$. 
			Right: resulting instance for \textsc{Graphic Inverse Voronoi}.
			Each connected shaded region corresponds to one set of $\UU$.}
	  \label{fig:reductionSAT1}
  \end{figure}

  \begin{itemize}
  \item For each element $x_i \in \mathcal V$, 
	we add two vertices $v(x_i)$ and $\overline{v}(x_i)$ to the vertex set of $G$, 
	and we connect them by an edge. 
	We add the set $\{ v(x_i),\overline{v}(x_i) \}$ to the candidate Voronoi cells $\UU$.

  \item For each subset $C_j=\{x_a,x_b,x_c\}$, we add three vertices $v(j,a)$, $v(j,b)$, 
	and $v(j,c)$ to $V(G)$, and we connect the three pairs by edges, forming a triangle.
	We add the set $\Delta(j)=\{v(j,a), v(j,b), v(j,c)\}$ to $\UU$.

  \item Finally, for each $x_a\in \mathcal V$ and each $C_j\in \mathcal C$ with $x_a \in C_j$, 
	we link $v(j,a)$ to $v(x_a)$ by an edge.
  \end{itemize}
  This finishes the construction of $(G,\UU)$.
  We observe that the sets of $\UU$ are indeed pairwise disjoint and of size $2$ or $3$.
  The graph $G$ is planar since it is obtained from the planar incidence graph $I(\mathcal V,\mathcal C)$ 
  by adding pendant vertices and splitting each vertex representing a subset (with three neighbors) 
  into a triangle in which each vertex is linked to one distinct neighbor.  
  
  If there is a solution $T$ to the instance $(\mathcal V,\mathcal C)$, we position the sites in the following way.
  For each $x_i \in \mathcal V$, we place the site of $\{v(x_i),v(\overline{x_i})\}$ in $v(x_i)$ if $x_i \in T$, 
  and in $\overline{v}(x_i)$, otherwise.
  For each $C_j=\{x_a,x_b,x_c\} \in \mathcal C$, 
  we place the site of $\Delta(j)$ in $v(j,z)$, where $x_z$ is the unique element of $C_j \cap T$.
  We denote by $\Sigma$ the obtained set of sites.
  We check that this placement defines the same Voronoi cells as specified by $\UU$. 
  \begin{itemize}
  \item For each $\overline{v}(x_i)\in \Sigma$, 
		we have $\cell_G(\overline{v}(x_i),\Sigma) \supseteq \{v(x_i),\overline{v}(x_i)\}$, 
		since by construction there is no site in $v(x_i)$.
		The only neighbors of $\{v(x_i),\overline{v}(x_i)\}$ 
		are vertices $v(j,i)$ for some values of $j \in [m]$. 
		However, those neighbors do not contain a site of $\Sigma$ by construction.
		On the other, there is always a site of $\Sigma$ at distance at most $1$ of $v(j,i)$, 
		whereas $\overline{v}(x_i)$ is at distance $2$ of $v(j,i)$.
		Hence, $\cell_G(\overline{v}(x_i),\Sigma) = \{v(x_i),\overline{v}(x_i)\}$.
  \item Similarly, for each $v(x_i) \in \Sigma$, we have 
		$\cell_G(v(x_i),\Sigma) \supseteq \{v(x_i),\overline{v}(x_i)\}$, since by construction there 
		is no site in $\overline{v}(x_i)$.
		The only neighbors of $\{v(x_i),\overline{v}(x_i)\}$ are vertices $v(j,i)$ 
		for some values of $j \in [m]$, but since $v(x_i) \in \Sigma$, 
		by construction, $v(j,i)$ also belongs to $\Sigma$.
		Therefore, $\cell_G(v(x_i),\Sigma) = \{v(x_i),\overline{v}(x_i)\}$.
  \item Finally, consider some $v(j,z) \in \Sigma$, where $C_j=\{x_a,x_b,x_c\}$. 
		We have $\cell_G(v(j,z),\Sigma) \supseteq \Delta(j)$ 
		because $v(j,z)$ is the only site in $\Delta(j)$.
		The only other neighbor of $v(j,z)$ is $v(x_z)$, which is in $\Sigma$.
		The only neighbor of $v(j,z')$ with $z' \in \{a,b,c\} \setminus \{z\}$ is $v(x_{z'})$ 
		which is at distance $2$ of $v(j,z)$ and at distance $1$ of the site 
		$\overline{v}(x_{z'}) \in \Sigma$.
		Thus, $\cell_G(v(j,z),\Sigma) = \Delta(j)$.
  \end{itemize}

  If there is no solution to the instance $(\mathcal V,\mathcal C)$, 
  we show that there is no solution to the \textsc{Graphic Inverse Voronoi} instance $(G,\UU)$.
  Fix a position of the sites.
  The set of sites $\Sigma$ has to intersect each $\{v(x_i),\overline{v}(x_i)\}$ exactly once.
  Define the set 
  \[
	T ~=~ \{x_i \in \mathcal V \mid \text{the site chosen for $\{ v(x_i),\overline{v}(x_i)\}$ is $v(x_i)$}\}.
  \]
  As $T$ is not a solution for the \textsc{Positive 1-in-3-SAT} instance, 
  there is a set $C_j=\{x_a,x_b,x_c\} \in \mathcal C$ such that $|C_j \cap T| \neq 1$.
  We now turn our attention to the site chosen for $\Delta(j)$.
  We distinguish two cases: $|C_j \cap T| = 0$ and $|C_j \cap T| \geq 2$.
  If $|C_j \cap T| = 0$, then, for every position of the site, say in $v(j,z)$ (with $z \in \{a,b,c\}$), $\cell_G(v(j,z),\Sigma)$ contains $v(x_z)$, and therefore cannot be equal to $\Delta(j)$.
  Now if $|C_j \cap T| \geq 2$, let $v(x_z)$ and $v(x_{z'})$ be two sites of $\Sigma$ with $z \neq z' \in \{a,b,c\}$.
  Since $\Delta(j)$ contains precisely one site, we have $v(j,z) \notin \Sigma$ or $v(j,z') \notin \Sigma$. 
  If $v(j,z) \notin \Sigma$, then $\cell_G(v(x_z),\Sigma)$ contains $v(j,z)$, and therefore cannot be equal to 
  $\{v(x_z), \overline{v}(x_z)\}$.
  Similarly, if $v(j,z') \notin \Sigma$, then $\cell_G(v(x_{z'}),\Sigma)\neq \{v(x_z), \overline{v}(x_z)\}$.
  In both cases, we reach the conclusion that there cannot be a solution for the instance $(G,\UU)$.
  
  Note that in the argument we did not use that $I(\mathcal V,\mathcal C)$ or $G$ are planar.
\end{proof}

Using additional properties of the reduction from \textsc{(Planar) 3-SAT} 
to \textsc{(Planar) Positive 1-in-3-SAT} given by Mulzer and Rote \cite{Mulzer08}
and the Sparsification Lemma, we derive the following conditional lower bound.

\begin{corollary}
  \label{cor:PGIVhardness}
  Unless the Exponential Time Hypothesis fails, the problem \textsc{Graphic Inverse Voronoi} 
  cannot be solved in time $2^{o(n)}$ in general graphs 
  and in time $2^{o(\sqrt n)}$ in planar graphs, where $n$ is the number of vertices, 
  even when the potential Voronoi cells are disjoint and of size at most 3.
\end{corollary}
\begin{proof}
  Applying the reduction of Mulzer and Rote~\cite{Mulzer08} to a \textsc{3-SAT} instance with $n$
  variables and $m$ clauses gives an instance to \textsc{Positive 1-in-3-SAT} with $O(n+m)$
  variables and $O(m)$ clauses. This is so because in their reduction
  each clause is replaced locally using $O(1)$ new variables and clauses.
  The reduction and the proof used in Theorem~\ref{thm:PGIVhardness} then give an instance 
  with $O(n+m)$ vertices. (The reduction also works for non-planar instances, as mentioned at the end
  of the proof.) Therefore, if we could solve \textsc{Graphic Inverse Voronoi}
  in time $2^{o(|V(G)|)}$, we could solve any \textsc{3-SAT} instance with $n$ variables and $m$ clauses
  in time $2^{o(|V(G)|)}= 2^{o(n+m)}$. 
  However, the Sparsification Lemma~\cite{Impagliazzo01} rules out, under the Exponential Time Hypothesis, 
  a running time $2^{o(n+m)}$ for \textsc{3-SAT}.
  
  The reduction from \textsc{3-SAT} to \textsc{Planar 3-SAT} given by Lichtenstein~\cite{Lichtenstein82} 
  increases quadratically the number of variables and clauses. 
  Together with the reduction of Mulzer and Rote from \textsc{(Planar) 3-SAT} to 
  \textsc{(Planar) Positive 1-in-3-SAT} and our reduction in the proof of Theorem~\ref{thm:PGIVhardness},
  we conclude that each instance of \textsc{3-SAT} with $n$ variables and $m$ clauses becomes
  an instance of \textsc{Graphic Inverse Voronoi} where the graph $G$ is \emph{planar} and has $O((n+m)^2)$ vertices.
  Again, solving the problem in planar graphs in time $2^{o(\sqrt{|V(G)|})}$ time for planar graphs
  would contradict the Sparsification Lemma.
\end{proof}

This upper bound of $3$ for the size of the potential Voronoi cells is sharp.

We show that the problem can be solved in polynomial time when each potential Voronoi cell has at most two points not contained in other potential cells.
For this, one uses a reduction to \textsc{$2$-SAT}.
Inspired by Lemma~\ref{le:strict}, 
we say that each $U\in \UU$ defines the \emph{potential open Voronoi cell}
\[
	U\setminus \left( \bigcup_{U'\in \UU\setminus \{U \}} U'\right) .
\]

\begin{theorem}
  \label{thm:GIVsize2}
  The \textsc{Graphic Inverse Voronoi} problem can be solved in polynomial time 
  when all the potential open Voronoi cells are of size at most $2$. 
\end{theorem}
\begin{proof}
  We present a polynomial reduction to \textsc{2-SAT}. 
  See Figure~\ref{fig:reductionSAT2} for an example.
  Let $(G,\UU=\{U_1,\ldots,U_k\})$ be the \textsc{Graphic Inverse Voronoi} instance.
  We denote by $U'_i$ the open potential Voronoi cell of the potential Voronoi cell $U_i$.
  By assumption, $|U'_i| \leq 2$.
  Because of Lemma~\ref{le:strict}, if the instance has a solution, then $s_i\in U'_i$.
  For each open cell $U'_i$, we introduce a variable $x_i$.
  We interpret putting the site on one fixed but arbitrary vertex of $U'_i$ to setting $x_i$ to true, 
  and putting the site on the other vertex (if it exists) to setting $x_i$ to false.
  Now, $\VV_G(\Sigma) = \UU$ if and only if for each pair of sites $s_i, s_j \in \Sigma$ 
  with $s_i \in U'_i$ and $s_j \in U'_j$:
  \begin{itemize}
  \item
    every vertex of $U_i \setminus U_j$ is strictly closer to $s_i$ than to $s_j$, and
  \item
    every vertex of $U_j \setminus U_i$ is strictly closer to $s_j$ than to $s_i$, and
  \item
    every vertex of $U_i \cap U_j$ is equidistant to $s_i$ and $s_j$.    
  \end{itemize}
  Therefore, one just needs to check that each pair of sites of $\Sigma$ is \emph{compatible}, 
  that is, satisfies those three conditions.

  \begin{figure}
	  \centering
	  \includegraphics[page=2]{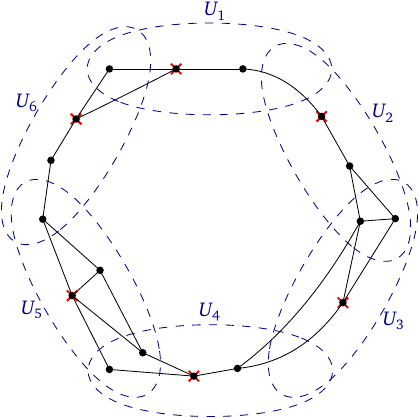}
	  \caption{An instance satisfying the assumption of Theorem~\ref{thm:GIVsize2}.
	    The vertices of each set $U\in \UU$ are enclosed by dashed curve.
		The crosses indicate the position when the variables are true.
	        Some of the compatibility clauses are $x_5\vee \overline{x_4}$, $\overline{x_5}\vee x_4$, $x_5 \vee x_4$, $\overline{x_2} \vee \overline{x_3}$, etc., as well as the 1-clause $x_3$.}
	  \label{fig:reductionSAT2}
  \end{figure}
  
  We define the following set of \textsc{2-SAT} constraints.
  For each open cell $U'_i$ of size $1$, we add the clause $x_i$, which forces to set $x_i$ to true.
  For each pair $s_i \in U'_i, s_{j} \in U'_j$ which is \emph{not} compatible 
  we add the clause $\ell_i \lor \ell_j$ where $\ell_i$ (resp. $\ell_j$) is the opposite literal 
  to the one chosen by placing a site in $s_i$ (resp. $s_j$).

  It is easy to check that the produced \textsc{2-SAT} formula is satisfiable if and only if there is a pairwise compatible set of sites.
  This is in turn equivalent to the existence of a solution for the \textsc{Graphic Inverse Voronoi} instance.
\end{proof}

%%%%%%%%%%%%%%%%%%%%%%%%%%%%%%%%%%%%%%%%%%%%%%%%%%%%%%%%%%%%%%%%%%%%%%%%%%%%%%%%%%%%%%%%%%%%%%%%%%%%%%%%%%%%%%%%%%%%%%%%%%%%%%%%%%%%%%%%%%%%%%%%%%%%%%%%%%%%%%%%%%%%%%%%%%%%%%%%%%%%%%%%%%%%%%%%%%%%%%%%%%%%%%%%%%%%%%%%%%%%%%%%%%%%%%%%%%%%%%%%%%%%%%%%%%%%%%%%%%%%%%%%%%%%%%%%%%%%
\section{Hardness parameterized by the number of Voronoi cells}
\label{sec:numberofcells}
In the previous section we showed that the problem \textsc{Graphic Inverse Voronoi} is NP-hard.
Stronger lower bounds are derived under the assumption of the Exponential Time Hypothesis (ETH).
We will prove the following result.

\begin{theorem}
\label{thm:W1hardness}
	The \textsc{Graphic Inverse Voronoi} problem is \textup{W[1]}-hard parameterized by the
	number of candidate Voronoi cells.
	Furthermore, for $n$-vertex graphs and $k$ subsets to be candidate Voronoi cells, 
	for any computable function $f$,
	there is no algorithm to solve the \textsc{Graphic Inverse Voronoi} problem in
	$f(k)n^{o(k/\log k)}$ time, unless the Exponential Time Hypothesis fails.
	The claim holds even for graphs with unit edge-lengths.
\end{theorem}

Note that it is trivial to solve the problem in $n^{O(k)}$ time: just try
each $n^k$ tuples of $k$ vertices as candidate sites and check each of them.
The remaining of this section is devoted to prove Theorem~\ref{thm:W1hardness}.
We will make a reduction from the following problem.

\begin{quote}
	\textsc{Multicolored Subgraph Isomorphism}\\
	Input: $(H,P)$, where $H$ is a graph whose vertex set $V(H)$ 
		is partitioned into $\ell$ pairwise disjoint sets $V_1 \uplus \cdots \uplus V_\ell$,
		and a pattern graph $P$ with vertex set $V(P)=[\ell]$.\\
	Question: Can we select vertices $v_i\in V_i$ for every $i\in [\ell]$ such that
		we have $v_iv_j\in E(H)$ for each $ij\in E(P)$?
\end{quote}

When the answer is positive, we say that $P$ is isomorphic to a multicolored subgraph of $H$.
It follows from the work of Marx~\cite{Marx10} that,
assuming the Exponential Time Hypothesis, 
the \textsc{Multicolored Subgraph Isomorphism} 
cannot be solved in time $f(\ell)n^{o(\ell/\log \ell)}$ for any computable function $f$, 
even when the pattern $P$ has $\Theta(\ell)$ edges.
This lower bound is made explicit for example in~\cite[Corollary 5.5]{MarxP15},
where $P$ is assumed to be $3$-regular. 

Consider an instance $(H,P)$ to the \textsc{Multicolored Subgraph Isomorphism} problem,
where $V_1,\dots,V_k$ are the partite classes of $V(H)$ and $P$ has $\Theta(\ell)$ edges.
We assume for simplicity that each vertex of $P$ has degree at least $2$.
For each $i,j\in [\ell]$, let $E_H(V_i,V_j)$ denote the edges of $H$ with one
endpoint in $V_i$ and the other endpoint in $V_j$.
We shall assume that $E_H(V_i,V_j)$ is empty whenever $ij\notin E(P)$ because
those edges can be removed without affecting the instance.
We build an instance $(G,\UU)=(G(H,P),\UU(H,P))$ for the
\textsc{Graphic Inverse Voronoi} problem as follows.

Figures~\ref{fig:reduction1} and~\ref{fig:reduction2}
may be helpful to follow the construction.

\begin{figure}
	\centering
	\includegraphics[page=1]{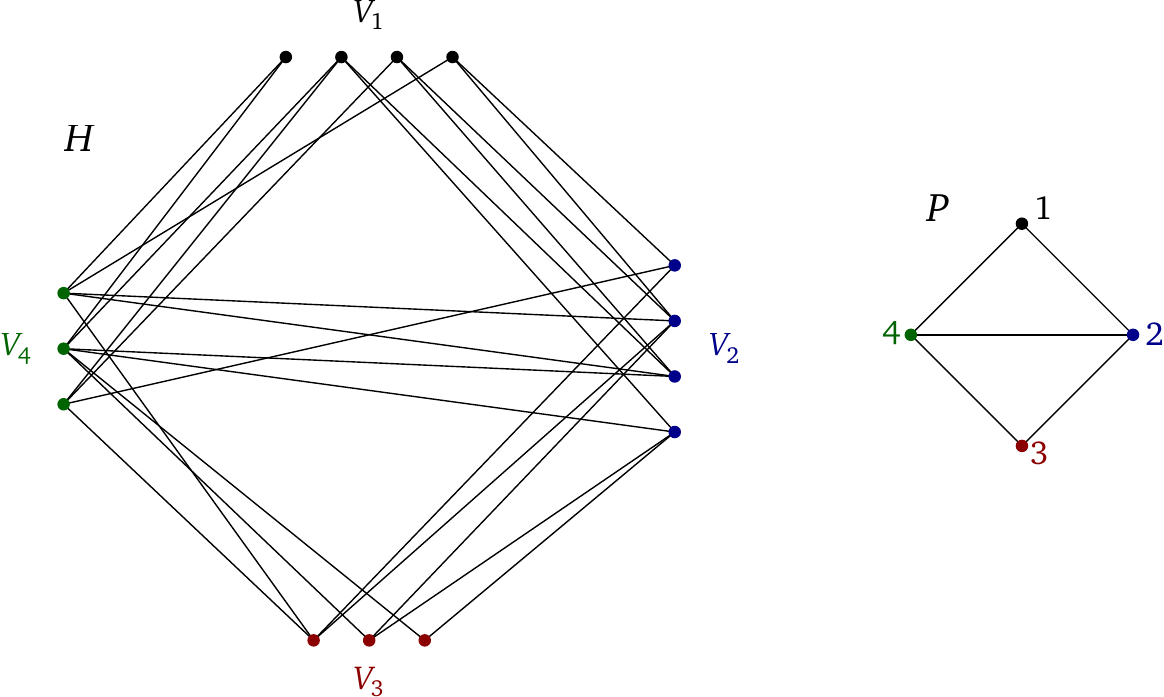}
	\caption{A graph $H$ (left) and a pattern $P$ (right) for the \textsc{Multicolored Subgraph Isomorphism}.}
	\label{fig:reduction1}
\end{figure}

\begin{itemize}
	\item We start with $V(G)=V(H)$ and $E(G)=E(H)$. 
	\item We subdivide each edge $e$ of $G$ with a new vertex, which we call $w(e)$.
        \item For each $i\in [\ell]$, we add all edges between all the vertices in $V_i$. 
	\item For each $ij\in E(P)$, let $W_{ij}$ be the vertices $w(e)$ used to subdivide $E_H(V_i,V_j)$.
		We add all edges between all the vertices in $W_{i,j}$.
	\item For each $ij\in E(P)$, we add $U_{ij} = W_{ij}\cup V_i \cup V_j$ to $\UU$.
\end{itemize}
All the edges have unit length.
This completes the construction of $G=G(H,P)$ and $\UU=\UU(H,P)$.
Note that $\UU$ has $|E(P)|=\Theta(\ell)$ candidate Voronoi regions,
while $G$ has $|V(H)|+ |E(H)|=\Theta(|V(H)|^2)$ vertices and
\[
	\sum_{i\in [\ell]} \binom {|V_i|}{2}+ \sum_{ij\in E(P)} \binom{|E_H(V_i,V_j)|}{2} + 2|E(H)|
\]
edges. 

The next two lemmas show that the pair $(G,\UU)$ is a correct reduction
from \textsc{Multicolored Subgraph Isomorphism} to \textsc{Graphic Inverse Voronoi}.
The intuition of the reduction is that selecting the site of
each Voronoi cell corresponds to selecting an edge of $E_H(V_i,V_j)$ for each $ij\in E(P)$.
Moreover, the selection of the edges we make needs to have compatible endpoints in 
each partite set $V_i$, as otherwise we do not get the correct Voronoi cells.

\begin{figure}
	\centering
	\includegraphics[page=2,width=\textwidth]{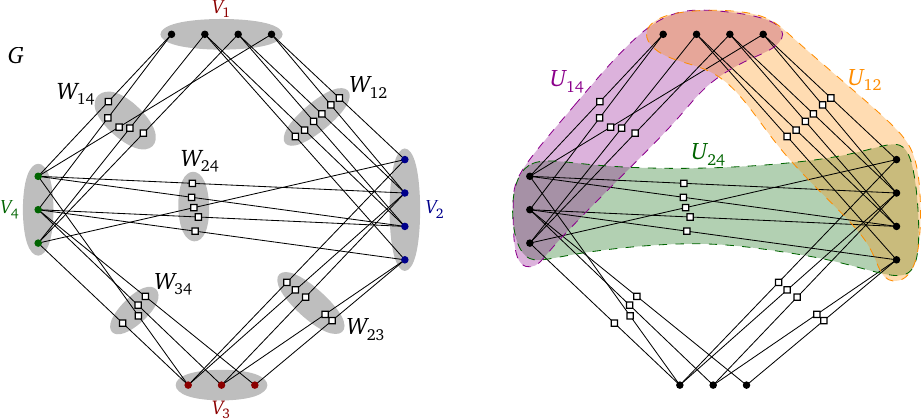}
	\caption{Example showing the construction for $(H,P)$ of the Figure~\ref{fig:reduction1}.
		Left: the graph $G$. The vertices inside a connected shaded region form a clique whose edges
		are not shown in the drawing.
		Right: three (of the five) candidate Voronoi cells of $\UU$ are indicated by shaded regions of different colors.
		The cliques induced by $V_i$ and $W_{ij}$ are not shown in this figure.}
	\label{fig:reduction2}
\end{figure}

\begin{lemma}
\label{le:forward}
	If $P$ is isomorphic to a multicolored subgraph of $H$, then
	$G$ has a set $\Sigma$ of sites such that $\VV(\Sigma)=\UU$.
\end{lemma}
\begin{proof}
	Assume that $P$ is isomorphic to a multicolored subgraph of $H$.
	This means that there are vertices $v_i\in V_i$, for all $i\in [\ell]$,
	such that $v_iv_j\in E(H)$ for every $ij\in E(P)$.
	This means that, for every $ij\in E(P)$, 
	the vertex $w(v_iv_j)$ obtained when subdividing $v_iv_j$ belongs to $G$.
	We define $s_{ij}=w(v_iv_j)$ for every $ij\in E(P)$ and 
	$\Sigma=\{ w(v_iv_j)\mid ij\in E(P)\}$.
	
	We claim that $\Sigma$ is a set of sites in $G$ such that 
	$\cell(s_{ij},\Sigma)= U_{ij}$, for every $ij\in E(P)$.
	This claim implies the lemma.

	For each $s_{ij}\in \Sigma$ and for each vertex $u$ of $G$ we have the following distances
	\begin{equation*}
		d_G\left(s_{ij},u\right) ~=~
		\begin{cases}
			0 & \text{if $u=s_{ij}$},\\
			1 & \text{if $u\in W_{ij}\setminus \{ s_{ij}\}$},\\
			1 & \text{if $u=v_i$ or $v=u_j$},\\
			2 & \text{if $u\in (V_i\cup V_j)\setminus \{ v_i,v_j\}$},\\
			\ge 2 & \text{if $u\in W_{i'j'}$ for some $i'j'\neq ij$},\\
			\ge 3 & \text{if $u\in V_t$ for some $t\neq i,j$}.
		\end{cases}
	\end{equation*}
	Therefore, each vertex in $V_1\cup\dots \cup V_\ell=V(H)$ is at distance at most $2$ from some vertex of $\Sigma$
	and each vertex in $\bigcup_{ij\in E(P)} W_{ij}$ is at distance at most $1$ from some vertex of $\Sigma$.
	
	Now we note that, for each $ij\in E(P)$,
	each vertex of $W_{ij}$ is strictly closer to $s_{ij}$ than to any other site.
	Furthermore, for each $i\in [\ell]$, each vertex of $V_i$ has the same distance 
	to each site $s_{ij'}$ with $ij'\in E(P)$,
	and a larger distance to each $s_{i'j'}$ with $i'j'\in E(P-i)$.
	Therefore $\cell(s_{ij},\Sigma)= W_{ij}\cup V_i \cup V_j = U_{ij}$.
	The result follows.
\end{proof}

\begin{lemma}
\label{le:backward}
	If $G$ has a set $\Sigma$ of sites such that 
	$\VV(\Sigma)=\UU$, then
	$P$ is isomorphic to a multicolored subgraph of $H$.
\end{lemma}
\begin{proof}
	Let $\Sigma$ be a set of sites in $G$ such
	that $\VV(\Sigma)=\UU$. For each $ij\in E(P)$,
	let $s_{ij}$ be the site of $\Sigma$ with $\cell(s_{ij},\Sigma)=U_{ij}$.
	
	Because of Lemma~\ref{le:strict},
	each $s_{ij}\in \Sigma$ belongs to 
	\[
		\cell^<(s_{ij},\Sigma) ~=~ \cell(s_{ij},\Sigma)\setminus 
			\left(\bigcup_{s\in \Sigma\setminus\{ s_{ij}\}} \cell(s,\Sigma)\right)
		~=~ U_{ij}\setminus \left(\bigcup_{i'j'\in E(P)\setminus\{ ij \}} U_{i'j'} \right) ~=~
		W_{ij}.
	\]
	In the last equality we have used that each vertex of $P$ has degree at least $2$,
	which means that each $V_i$ is contained in at least $2$ sets $U_{ij}$ of $\UU$.
	We conclude that, for each $ij \in E(P)$,
	the site $s_{ij}$ must be in $W_{ij}$.
	
	Since each site $s_{ij}$ is in $W_{ij}$, for each $ij \in E(P)$,
	the construction of $G$ implies that
	there are unique vertices $v(i,ij)\in V_i$ and $v(j,ij)\in V_j$
	such that $s_{ij}$ is the vertex obtained when subdividing the edge connecting
	$v(i,ij)$ and $v(j,ij)$. In particular, $v(i,ij)v(j,ij)$ is an edge of $E(H)$.

	Fix the index $i\in [\ell]$ and consider two edges $ij, ij'\in E(P)$ incident to $i$.
	We must have $v(i,ij)=v(i,ij')$, as otherwise we would have
	$d_G(s_{ij},v(i,ij))=1 < 2 = d_G(s_{ij'},v(i,ij'))$,
	which would imply that $V_i\not\subset \cell(s_{ij'},\Sigma)=U_{ij'}$ and would contradict
	the definition of $U_{ij'}=V_i\cup V_{j'}\cup W_{ij'}$.
	Therefore, each of the (three) edges $ij$ of $E(P)$ defines the same vertex $v(i,ij)\in V_i$.
	Henceforth we denote this vertex by $v_i$.
	
	We have found $\ell$ vertices $v_1,\dots,v_\ell$ with the property 
	that $v_i\in V_i$, for each $i\in [\ell]$, and such that
	the edge $v_iv_j= v(i,ij)v(j,ij)$ in $E(H)$, for each $ij\in E(P)$. 
	This means that 
	$P$ is isomorphic to the multicolored subgraph of $H$ defined by $\{v_1,\dots,v_\ell\}$.
\end{proof}

\begin{proof}[Proof of Theorem~\ref{thm:W1hardness}]
	As shown in Lemmas~\ref{le:forward} and~\ref{le:backward},
	$H$ has a multicolored subgraph isomorphic to $P$ if and only if $\UU$ is a valid
	Voronoi diagram of $G$.
	Thus, the answers to \textsc{Multicolored Subgraph Isomorphism}$(H,P)$ and
	\textsc{Graphic Inverse Voronoi}$(G,\UU)$ are the same.
	
	Recall that $\UU$ has $|E(P)|=\Theta(\ell)$ candidate Voronoi regions.
	If we could solve each instance of the
	\textsc{Graphic Inverse Voronoi} problem with $n$ vertices and $k$
	sites in time $f(k)n^{o(k/\log k)}$, for some computable function $f$,
	then we could solve the instance $(G,\UU)$ in
	\[
		f(|\UU|)\cdot |V(G)|^{o(|\UU|/\log \UU|)} ~\le~ 
		f(\Theta(\ell))(\Theta(|V(H)|^2))^{o(\Theta(\ell)/\log (\Theta(\ell)))} \le
		g(\ell)|V(H)|^{o(\ell/\log \ell)}
	\]
	time, for some computable function $g$. 
	However, this also means that we could solve the 
	\textsc{Multicolored Subgraph Isomorphism} in $H$ with pattern $P$
	in $g(\ell)|V(H)|^{o(\ell/\log \ell)}$ time,
	and this contradicts the Exponential Time Hypothesis.
\end{proof}

%%%%%%%%%%%%%%%%%%%%%%%%%%%%%%%%%%%%%%%%%%%%%%%%%%%%%%%%%%%%%%%%%%%%%%%%%%%%%%%%%%%%%%%%%%%%
\section{Hardness parameterized by the pathwidth and the treewidth}
\label{sec:treewidth}

In this section we show that the \textsc{Graphic Inverse Voronoi} problem is unlikely to 
be fixed parameter tractable with respect to the pathwidth of the graph. Since
the pathwidth is always smaller than the treewidth, this implies the same result
for the treewidth.
More precisely, in this section we will prove the following.

\begin{theorem}
\label{thm:pathwidth}
  The \textsc{Graphic Inverse Voronoi} problem is \textup{W[1]}-hard 
  parameterized by the pathwidth of the input graph.
  Furthermore, for $n$-vertex graphs with pathwidth $p$, 
  there is no algorithm to solve the \textsc{Graphic Inverse Voronoi} problem
  in time $f(p)n^{o(p)}$ for any computable function $f$,
  unless the Exponential Time Hypothesis fails.
  The claims hold even for graphs with unit edge-lengths and disjoint candidate Voronoi cells.
\end{theorem}

In order to show that, we will reduce the following \textup{W[1]}-hard problem.
\begin{quote}
	\textsc{Multicolored Independent Set}\\
	Input: A graph $H=(V,E)$ whose vertex set $V$ is partitioned into $\ell$ pairwise
		disjoint sets $V_1 \uplus \cdots \uplus V_\ell$.\\
	Question: Is there an independent set $X$ of size $\ell$ in $H$ 
		such that $|X \cap V_i|=1$, $\forall i \in [\ell]$?
\end{quote}

An independent set of $H$ is said \emph{multicolored} if it satisfies $|X \cap V_i|=1$, $\forall i \in [\ell]$. 
The \textsc{Multicolored Independent Set} problem is \textup{W[1]}-hard with respect to $\ell$ 
and cannot be solved in time $f(\ell)n^{o(\ell)}$ for any computable function $f$, assuming 
the Exponential Time Hypothesis~\cite[Corollary 14.23]{CyganFKLMPPS15}.
The lower bounds still hold if all the partite sets $V_i$ have the same cardinality
and there are no edges connecting any two vertices within a set $V_i$.

Let $H=(V_1 \uplus \cdots \uplus V_\ell, E)$ be an instance 
of \textsc{Multicolored Independent Set} such that $|V_1| = |V_2| = \cdots = |V_\ell|=t$.
Let $m$ be the number of edges in $H$.
We build an equivalent \textsc{Graphic Inverse Voronoi} instance $(G,\UU)$ 
where the pathwidth of $G$ is $\Theta(\ell)$.
This instance will have unit edge-length edges and the sets in $\UU$ will be pairwise disjoint.

Our global strategy for the reduction is to propagate a vertex choice 
in each $V_i$ with a \emph{path-like} structure with $\ell$ \emph{rows} 
and $m=|E|$ \emph{columns}.
In each column, we introduce a single distinct edge of $E$ 
so that the pathwidth of the built graph stays in $\Theta(\ell)$. 
Figure~\ref{fig:treewidth1} shows
the whole reduction in the graph $H$ of Figure~\ref{fig:graphH}. 
(Seeing the details requires zooming in.)
In Figure~\ref{fig:treewidth2} we show a part of the construction in detail showing also the
notation we employ.
The detailed construction is as follows.  

\begin{figure}
	\centering
	\includegraphics[page=1,scale=.75]{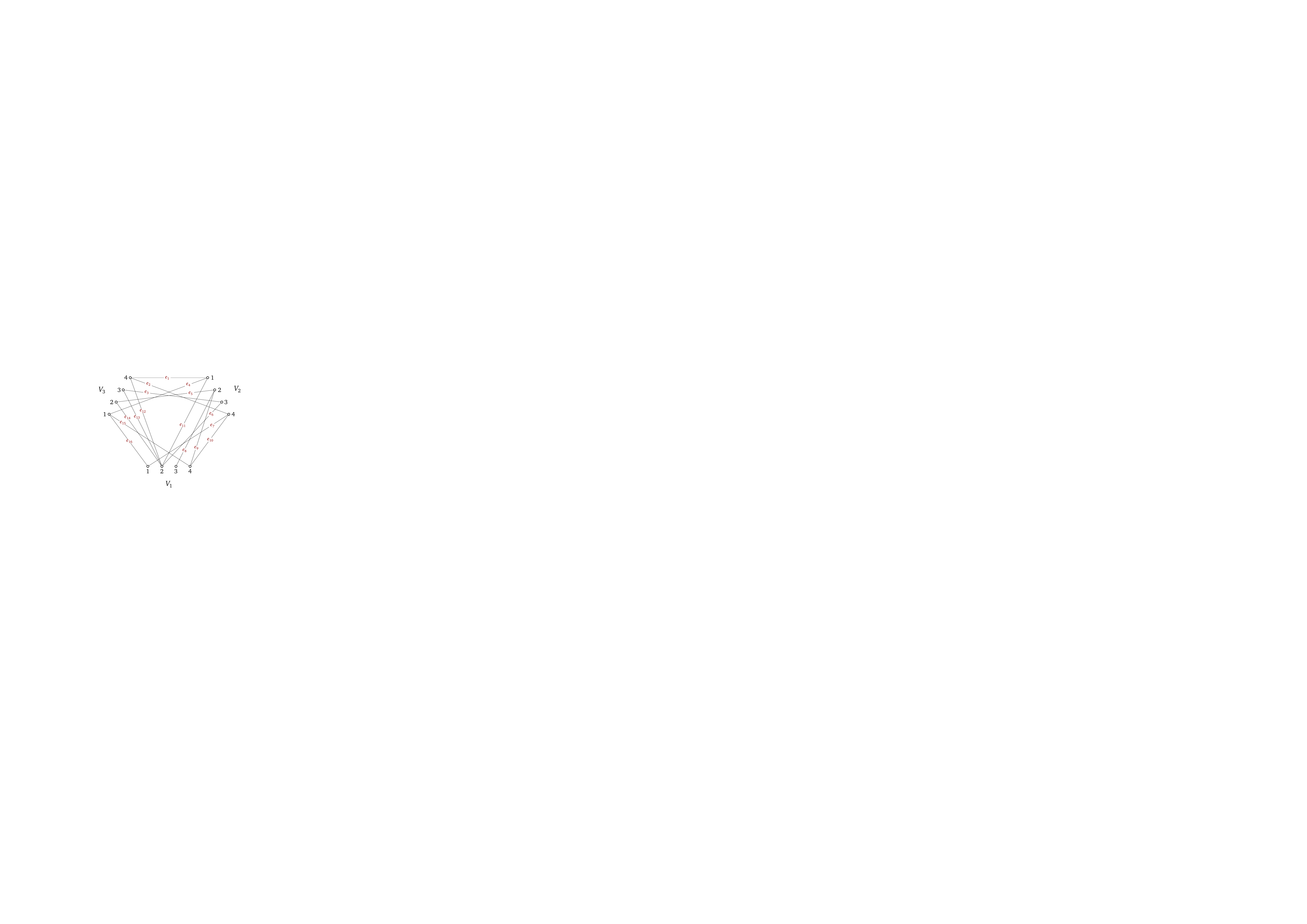}
	\caption{A graph $H$ whose vertex set is partitioned into $\ell=3$ partite sets, 
		each with $t=4$ vertices.}
	\label{fig:graphH}
\end{figure}

\begin{figure}
	\centering
	\includegraphics[page=2,angle=90,height=.9\textheight]{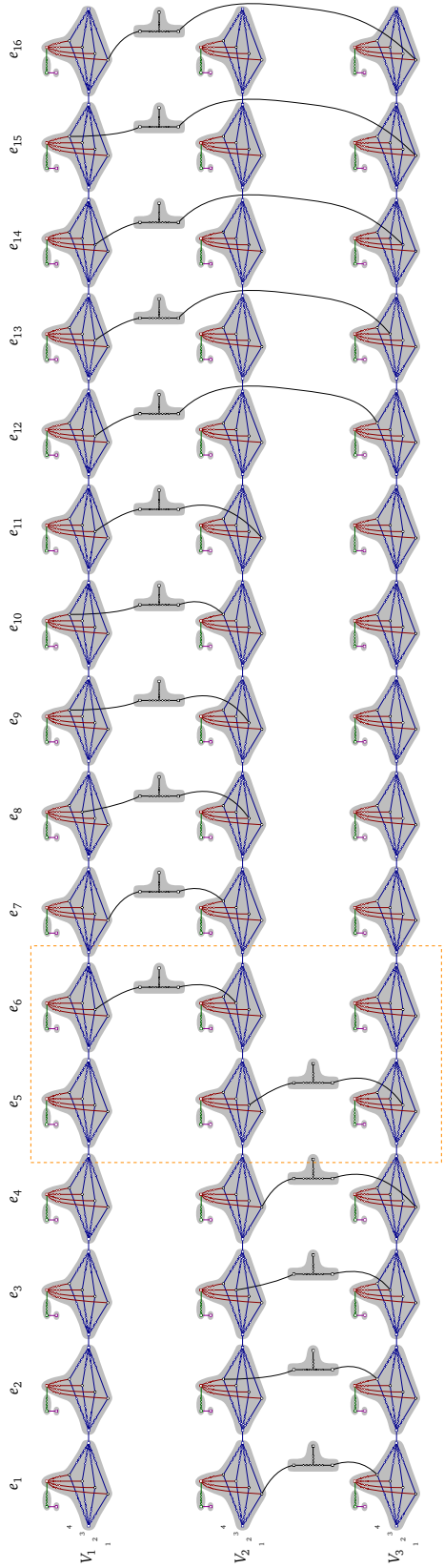}
	\caption{Whole graph showing the reduction used to prove Theorem~\ref{thm:pathwidth} 
	for the graph $H$ in Figure~\ref{fig:graphH}. Each connected gray area corresponds to one
	candidate Voronoi region. Figure~\ref{fig:treewidth2} shows details 
	for a part of the construction.}
	\label{fig:treewidth1}
\end{figure}

\begin{figure}
	\centering
	\includegraphics[page=3,height=.7\textheight]{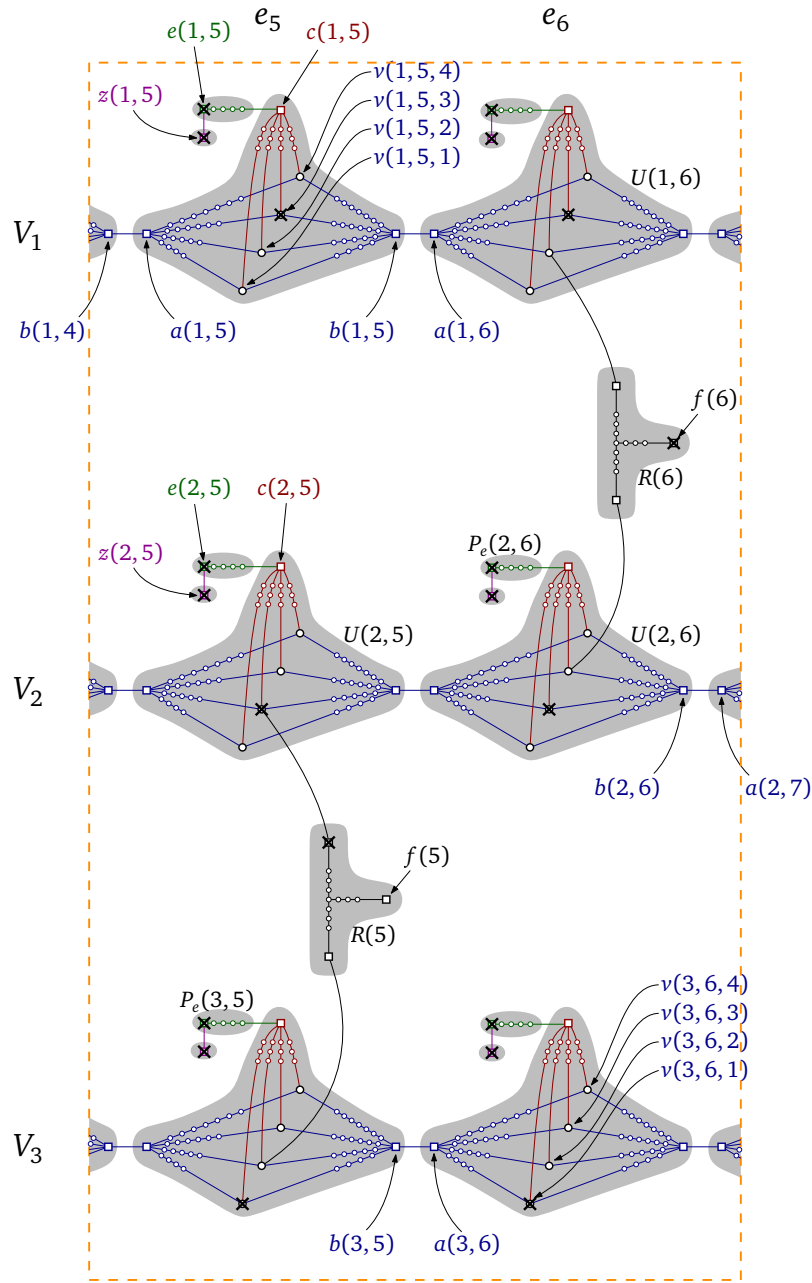}
	\caption{Zoom into a part of the reduction shown in Figure~\ref{fig:treewidth1}
		with some notation. 
		Each connected gray area corresponds to one candidate Voronoi region.
		Some selection of sites marked with crosses that is \emph{locally} correct 
		(but globally would have a problem).
		This selection corresponds to selecting vertex 3 of $V_1$, vertex of 2 of $V_2$ 
		and vertex $1$ of $V_3$.}
	\label{fig:treewidth2}
\end{figure}

\begin{itemize}
	\item
	For each $i\in [\ell]$ and $j\in [m]$, we add to $V(G)$ 
	an independent set $I(i,j)$ of size $|V_i|=t$.
	The vertices of the independent set $I(i,j)$ are denoted by $v(i,j,1)$ to $v(i,j,t)$, 
	the third index being in one-to-one correspondence with the vertices of $V_i$.
	\item 
	For each $i\in [\ell]$ and $j \in [m]$, 
	we add two vertices $a(i,j)$ and $b(i,j)$.
	Furthermore, for each $h \in [t]$, we connect $a(i,j)$ to $v(i,j,h)$ 
	by a private path $P_a(i,j,h)$ of length $t+h$, 
	and we connect $b(i,j)$ to $v(i,j,h)$ by a private path $P_b(i,j,h)$ of length $t+h$.
	For each $i\in [\ell]$ and $j \in [m-1]$, we connect $b(i,j)$ and $a(i,j+1)$ by an edge.
	\item
	For each $i\in [\ell]$ and $j \in [m]$, 
	we add three new vertices $c(i,j)$, $e(i,j)$ and $z(i,j)$.
	For each $h\in [t]$, we add a private path $P_c(i,j,h)$ 
	of length $t$ between $v(i,j,h)$ and $c(i,j)$. 
	Furthermore, we connect $e(i,j)$ and $z(i,j)$ with an edge
	and add a path $P_e(i,j)$ of length $t$ with one extreme on $e(i,j)$ 
	and the other extreme connected through an edge to $c(i,j)$. 
	(Thus $e(i,j)$ and $c(i,j)$ are connected with a path of length $t+1$.)
	\item
	For each $i\in [\ell]$ and $j \in [m]$,
	we denote by $U(i,j)$ the set of vertices comprising $I(i,j)$ and 
	all the paths going from this independent set to $a(i,j)$, $b(i,j)$,
	and $c(i,j)$, including those three vertices. 
	We add $U(i,j)$, $V(P_e(i,j))$, 
	and $Z(i,j)=\{z(i,j)\}$ to the candidate Voronoi cells $\UU$.
	\item
	We call the set $\bigcup_{i=1}^\ell U(i,j)\cup P_e(i,j)\cup Z(i,j)$ the \DEF{$j$-th column}, for a fixed $j \in [m]$.
	We introduce exactly one distinct edge of $E$ per column.
	Let $e_1,\ldots, e_m$ be any ordering of the edges of $E$.
	We put an edge gadget encoding $e_j$ in the $j$-th column, for every $j \in [m]$.
	Assume that $e_j$ is an edge between the $h$-th vertex of $V_i$ and the $h'$-th vertex of $V_{i'}$
	where $i\neq i'$.
	We add a path $P(e_j)$ of length $2t+2$ between $v(i,j,h)$ and $v(i',j,h')$. 
	We add a path $Q(e_j)$ of length $t$ between the middle vertex of $P(e_j)$
	and a new vertex, denoted $f(j)$. (The vertex $f(j)$ has degree 1 and it 
	is at distance $2t+1$ from $v(i,j,h)$ and from $v(i',j,h')$.)
	We add $R(j)=V(Q(e_j)) \cup V(P(e_j)) \setminus \{ v(i,j,h), v(i',j,h')\}$ 
	as a candidate Voronoi region to $\UU$. 
	The subgraph induced by $R(j)$ is the \DEF{edge gadget} of $e_j$.
\end{itemize}

That finishes the construction of $G=G(H)$ and of $\UU = \UU(H)$.
All the edges of $G$ have unit length.
One can observe that $\UU$ is made of pairwise disjoint sets and it contains
$(3\ell+1) m$ candidate Voronoi cells.
We first show that the pathwidth (and thus also the treewidth) of $G$ is at most $2\ell+5$.
For that, we use the cleaning game characterization of pathwidth presented in Section~\ref{sec:basic}.

\begin{lemma}
\label{le:pathwidth1}
    The pathwidth of $G$ is at most $2\ell+5$.
\end{lemma}
\begin{proof}
	We present a winning strategy for cleaning $G$ using $2\ell+6$ cleaners. 
	We make $m$ rounds where in the $j$-th round, $j=1,\dots,m$, we scan completely the $j$-th column
	and the gadget for $e_j$.

	At the start of the $j$-th round we have $2\ell$ cleaners placed at the vertices $a(i,j)$ and $b(i,j)$ for all $i\in [\ell]$.
	Assume that the edge $e_j \in E$ is between the $h$-th vertex of $V_i$ and the $h'$-th vertex of $V_{i'}$.
	We place two cleaners at $v(i,j,h)$ and $v(i',j,h')$.
        Let $X_j$ be the set of $2\ell+2$ vertices where we have cleaners.
        They will stay there for most of the $j$-th round. 

        We then clean the whole $j$-th column plus the edge gadget of $e_j$ using the remaining four cleaners.
        For every $i \in [\ell]$, we define $\mathcal G(i,j) = U(i,j) \cup P_e(i,j) \cup Z(i,j)$, the intersection of the $j$-th column with the \emph{$i$-th row}.
        For every $i$ going from 1 to $\ell$, we place a $2\ell+3$-rd cleaner at $c(i,j)$.
        All the connected components of $\mathcal G(i,j)$ in $G - (X_j \cup \{c(i,j)\})$ are paths or subdivisions of a claw (i.e., a star with three leaves).
        These graphs, and more generally, disjoint unions of subdivided stars can be decontaminated with three cleaners.
        Place a first cleaner at the center of the first subdivided star.
        This disconnects it into a disjoint union of paths, that can be cleaned with two additional cleaners.
        When all the paths are cleaned, move the first cleaner to the center of the second subdivided star and proceed similarly with the remaining connected components.
        When all the paths and subdivided claws of $\mathcal G(i,j)$ have been decontaminated, we remove the $2\ell+3$-rd cleaner (to place it at $c(i+1,j)$ at the next iteration).
        We proceed similary, up to $\mathcal G(\ell,j)$.
        Finally we observe that the edge gadget is also a subdivided claw, so it can be handled with three additional cleaners.
        
	%For this, we note that each connected component of $G - X_j$ contained in the $j$-th column and the connected component induced by $R(j)$ has pathwidth $2$, and thus it can be decontaminated with only three cleaners.
	At this point, the remaining toxic gas is confined \emph{to the right} of the $\ell$ cleaners placed at $b(1,j), \ldots, b(\ell,j)$, that is, on some edge incident to some vertex defined by $j'>j$. 
	If $j=m$, we are done since there can no longer be contaminated edges.
	Otherwise, we move the cleaners from $a(1,j), \ldots, a(\ell,j)$ to $a(1,j+1), \ldots, a(\ell,j+1)$, then the cleaners from $b(1,j), \ldots, b(\ell,j)$ to $b(1,j+1), \ldots, b(\ell,j+1)$, and start the next round.
\end{proof}
  
We now show the correctness of the reduction.
\begin{lemma}
 \label{le:pathwidth2}
	If $H$ has a multicolored independent set of size $\ell$, 
	then there is a set $\Sigma \subseteq V(G)$ such that $\VV_G(\Sigma)=\UU$. 
\end{lemma}
\begin{proof}
	Assume there is a multicolored independent set $X$ of size $\ell$ in $H$.
	We define the set $\Sigma$ of sites as follows.
	\begin{itemize}
	\item
	For each $i \in [\ell]$, we place a site on the $m$ vertices $v(i,j,h_i)$ 
	for all $j\in [m]$,	where $h_i$ is the index of the vertex of $X$ in the partite set $V_i$.
	\item
	For every $i \in [\ell]$ and $j \in [m]$, we place two sites at $e(i,j)$ and $z(i,j)$.
	\item
	For each edge $e_j$ of $H$ with no endpoint in $X$, we place a site at $s_j=f(j)$.
	\item
	For each edge $e_j$ of $H$ with exactly one endpoint in $X$, 
	we place a site on the vertex $s_j$ of $R(j)$ adjacent to $v(i,j,h_i)$.
	\end{itemize}
	Note that, since $X$ is an independent set, there cannot be an edge $e_j$ 
	with two endpoints in $X$. Therefore we have covered all cases.
	This finishes the placement of the sites.

	For every $i \in [\ell]$ and $j \in [m]$, we have
	$\cell(z(i,j),\Sigma) = Z(i,j)$ because $e(i,j)$, the only neighbor of $z(i,j)$, is also a site. 
	It also holds that $\cell(e(i,j),\Sigma) = P_e(i,j)$ because $c(i,j)$ is at distance $t+1$
	from $e(i,j)$ and at distance $t$ from the site $v(i,j,h_i)$.
	For every $i \in [\ell]$ and $j \in [m-1]$, the vertices $v(i,j,h_i)$ and $v(i,j+1,h_i)$ 
	are \emph{compatible} since $d_G(v(i,j,h_i),b(i,j))=t+h_i=d_G(v(i,j+1,h_i),a(i,j))$ and 
	$d_G(v(i,j,h_i),a(i,j+1))=t+h_i+1=d_G(v(i,j+1,h_i),b(i,j))$.
	Here it is relevant to choose the lengths of the paths $P_a(i,j,h)$, $P_b(i,j,h)$ and $P_c(i,j,h)$
	to ensure that the shortest path from vertex $v(i,j,h)$ to $a(i,j)$ is $P_a(i,j,h)$,
	instead of passing through $b(i,j)$ or $c(i,j)$. (Similar statements hold for the shortest paths from
	$v(i,j,h)$ to $b(i,j)$ and to $c(i,j)$.)

	We now only need to check that the site $s_j$ in the edge gadget of $e_j$ 
	--the edge, say, between the $h$-th vertex of $V_i$ and the $h'$-th vertex of $V_{i'}$-- 
	is compatible with the sites chosen in $\Sigma$ for $U(i,j)$ and $U(i',j)$.
	The nice property that makes everything work is that, 
	for every $i \in [\ell], j \in [m], h \neq h' \in [t]$, $d_G(v(i,j,h),v(i,j,h'))$ is always equal to $2t$.
	Indeed the shortest path between $v(i,j,h)$ and $v(i,j,h')$ goes through $c(i,j)$,
	which is at distance $t$ of both vertices.
  
	There are two cases: $s_j$ is on the path $P(e_j)$ or $s_j = f(j)$.
	If $s_j \in P(e_j)$, it means that one of the endpoints of $e_j$ is in the multicolored independent set $X$.
	Without loss of generality, we assume that it is the $h$-th vertex of $V_i$ (hence, $h=h_i$).
	In that case, the sites $s_j$ and $v(i,j,h_i)$ are adjacent vertices, and therefore they are compatible.
	The sites $s_j$ and $v(i',j,h_{i'})$ are also compatible since
	$d_G(s_j,v(i',j,h_{i'}))=2t+1$ and $d_G(v(i',j,h'),v(i',j,h_{i'}))=2t$.
	
	Now, if $s_j = f(j)$, it means that $e_j$ does not touch any vertex of $\Sigma$.
	Hence, $h \neq h_i$ and $h' \neq h_{i'}$.
	Then we have 
	$d_G(s_j,v(i,j,h_i))=d_G(s_j,v(i',j,h_{i'}))=2t+1$, $d_G(v(i,j,h),v(i,j,h_i))=2t$ and
	$d_G(v(i',j,h'),v(i',j,h_{i'}))=2t$. It follows that also in this case 
	the site $s_j$ is compatible with $v(i,j,h_i)$ and $v(i',j,h_{i'})$.

	Therefore, we showed that each site $v(i,j,h_i) \in \Sigma$ is compatible with every other site of $\Sigma$.
	This implies that for every $i \in [\ell]$ and $j \in [m]$, we have $\cell(v(i,j,h_i),T)=U(i,j)$.
	In turn, it implies that $\cell(s_j,T)=R(j)$ for each $j\in [m]$, and therefore $\VV_G(\Sigma)=\UU$.
\end{proof}
  
\begin{lemma}
\label{le:pathwidth3}
	If $H$ has no multicolored independent set of size $\ell$, 
	then there is no set $\Sigma \subseteq V(G)$ such that $\VV_G(\Sigma)=\UU$.
\end{lemma}
\begin{proof}
    A solution for the \textsc{Graphic Inverse Voronoi} has to put sites on every $e(i,j)$ and $z(i,j)$, 
	otherwise the Voronoi cell $Z(i,j)$ would not appear in the set of cells.
    As $e(i,j)$ is at distance $t+1$ of $c(i,j)$, 
	the site chosen for the cell $U(i,j)$ has to be at distance exactly $t$ of $c(i,j)$ (otherwise, 
	this site would not be compatible with $e(i,j)$).
    So, the site chosen for $U(i,j)$ has to be in $I(i,j)$.

    Then we prove that if a site is placed on $v(i,j,h)$, a site should be placed consistently on $v(i,j+1,h)$.
    This is immediate by construction, since the only vertex of $U(i,j+1)$ 
	which has a distance to $a(i,j+1)$ equal to $d_G(v(i,j,h),b(i,j))=t+h$ is $v(i,j+1,h)$.
	Here we are using again that the shortest path from $v(i,j,h)$ to $b(i,j)$ is indeed $P_b(i,j,h)$,
	and does not detour through $a(i,j)$ or $c(i,j)$.
    This implies that, for each $i \in [\ell]$, 
	all the choices of sites for the cells $\{U(i,j)\}_{j \in [m]}$ 
	have to be consistent to the same vertex, say of index $h_i$ in $V_i$.
    This defines a (consistent) set $X$ of $\ell$ vertices of $H$.

    As by assumption $X$ cannot be an independent set, there is an edge $e_j$ with both endpoints in $X$.
    Say those endpoints are the vertices in partite sets $V_i$ and $V_{i'}$.
    Then, the site for $R(j)$ cannot be closer to 
	the two vertices of $R(j)$ that are adjacent to $v(i,j,h_i)$ and $v(i',j,h_{i'})$.
    Hence there is no $\Sigma \subseteq V(G)$ such that $\VV_G(\Sigma)=\UU$.
  \end{proof}

\begin{proof}[Proof of Theorem~\ref{thm:pathwidth}]
	Because of Lemmas~\ref{le:pathwidth2} and~\ref{le:pathwidth3}, 
	solving \textsc{Graphic Inverse Voronoi} for $(G,\UU)$ also solves 
	\textsc{Multicolored Independent Set} for $H$. 
	The graph $G$ has $O(m \ell t^2) = O(|V(H)|^5)$ vertices 
	and pathwidth $p\le 2\ell +5$ because of Lemma~\ref{le:pathwidth1}.
	An algorithm for the \textsc{Graphic Inverse Voronoi}
	with running time $f(p)|V(G)|^{o(p)}$ (for some computable function $f$)
	would imply that we can solve \textsc{Multicolored Independent Set} in time 
	$f(2\ell+5) \left(|V(H)|^5\right)^{o(2\ell +5)} =g(\ell)n^{o(\ell)}$
	for a computable function $g$.
	This would disprove the Exponential Time Hypothesis.
\end{proof}

We show an almost matching upper bound when the potential Voronoi cells form a partition of the vertex set.

\begin{theorem}
  \label{thm:PGIV-XPtw}
  Instances $(G,\UU)$ of \textsc{Graphic Inverse Voronoi} can be solved 
  in time $|V(G)|^{O(w \log k)}$, when the $k$ cells of $\UU$ are pairwise disjoint 
  and $w$ is the treewidth of $G$. 
\end{theorem}
\begin{proof}
  We solve a more general problem where each potential Voronoi cell of $\UU$ 
  comes with a prescribed subset, specifying where one can actually place its site.
  Let $H$ be the graph on $k$ vertices obtained by contracting each cell of $\UU$ into a single vertex.
  Contracting edges can only decrease the treewidth (see for instance \cite{Diestel12}).
  Thus the treewidth of $H$ is at most $w$.
  We exhaustively guess in time $k^{w+1}$ a balanced vertex-separator $S$ of size $w+1$ in the graph $H$.
  Each connected component of $H-S$ has thus less than $2k/3$ vertices.
  We further guess in time at most $|V(G)|^{w+1}$ the $w+1$ corresponding sites, say, $s'_1, \ldots, s'_{w+1}$ (for the cells of $S$) in a fixed solution.
  For each guess, we remove the $w+1$ corresponding cells -- say, their union is $U$ -- from $G$, and update the prescribed subsets of the remaining cells to those placements compatible with the sites that are already fixed.
  More precisely, if $s$ is a placement of the site of $U'$ (not included in $U$) incompatible with a site $s'_i$ (for some $i \in [w+1]$), then we remove $s$ from the prescribed set for cell $U'$.
  We then solve recursively each connected component of $G-U$.
  Thus, we get $k|V(G)|^{w+1} \leq |V(G)|^{2(w+1)}$ independent subproblems, each of them with at most $2k/3$ candidate Voronoi regions (and restricted subset of possible placements).
  Since the depth of the branching tree is $O(\log k)$, the total running time is bounded by $|V(G)|^{O(w \log k)}$.
\end{proof}
  
%%%%%%%%%%%%%%%%%%%%%%%%%%%%%%%%%%%%%%%%%%%%%%%%%%%%%%%%%%%%%%%%%%%%%%%%%%%%%%%%%%%%%%%%%%%%%%%%%%%%%%%%%%%%%%%%%%%%%%%%%%%%%%%%%%%%%%%%%%%%%%%%%%%
\section{Arbitrary trees -- Transforming to nicer instances}
\label{sec:transform}
In this section we provide a transformation to reduce the problem
\textsc{Generalized Graphic Inverse Voronoi in Trees} to instances where the
tree has maximum degree $3$ and the candidate Voronoi regions are disjoint.
First we show how to transform it into disjoint Voronoi regions, 
and then we handle the degree.
In our description, we first discuss the transformation without paying attention to its
efficiency. At the end of the section we discuss how the transformation can be done
in linear time.

\subsection{Transforming to disjoint cells}
\label{sec:disjoint}
In this section we explain how to decrease the overlap between different Voronoi regions.
The procedure is iterative: we consider one edge 
of the tree at a time and transform the instance.
When there are no edges to process, we can conclude that the original instance has no solution or we can find a solution to the original instance.

Consider an instance $I=(T,( (U_1,S_1),\dots, (U_k,S_k) ))$
for the problem \textsc{Generalized Graphic Inverse Voronoi in Trees}.
See Figure~\ref{fig:disjoint} for an example of such an instance.

For each index $i\in [k]$ we define
\begin{align*}
	W_i ~&=~ U_i\setminus \bigcup_{j\neq i} U_j ,\\
	E_i ~&=~ \{ uv\in E(T)\mid u \in W_i,~ v\in U_i\setminus W_i\}.
\end{align*}

The intuition is that each $W_i$ should be the open Voronoi cell 
defined by the (unknown) site $s_i$, that is, the vertices of $T$ with $s_i$ 
as unique closest site; see Lemma~\ref{le:strict}. 
Each $E_i$ is then the set of edges with one vertex in $W_i$ and
another vertex in $U_i\cap U_j$ for some $j\neq i$.
The following result is easy to prove using Lemma~\ref{le:starshaped}.

\begin{lemma}
\label{le:transform0}
	Supposing that there is a solution to 
	\textsc{Generalized Graphic Inverse Voronoi in Trees} with input $I$,
	the following hold.
	\begin{itemize}
	\item[\rm (a)] Each set $U_i$ $(i\in [k])$ and each set $W_i$ $(i\in [k])$ 
		induces a connected subgraph of $T$.
    \item[\rm (b)] If two sets $U_i$ and $U_j$ $(i\ne j)$ intersect, 
		then $E_i\ne\emptyset$ and $E_j\ne \emptyset$.
	\end{itemize}
\end{lemma}
\begin{proof}
	Consider a solution $s_1,\dots,s_k$
	to \textsc{Generalized Graphic Inverse Voronoi in Trees} with input $I$,
	and define $\Sigma=\{ s_1,\dots, s_k\}$.
	This means that, for each $i\in [k]$,
	we have $s_i\in S_i$ and $U_i=\cell_T(s_i,\Sigma)$.
	Note that because of Lemma~\ref{le:strict}, we have
	\begin{equation*}
		\forall i\in [k]:~~~
				W_i ~=~ U_i\setminus \bigcup_{j\neq i} U_j ~=~
				\cell_T(s_i,\Sigma)\setminus \bigcup_{j\neq i} \cell_T(s_j,\Sigma) ~=~
				\cell^<_T(s_i,\Sigma).
	\end{equation*}

	If there are distinct indices $i,j\in [k]$ such that $U_i$ and $U_j$
	intersect, then $W_i\subsetneq U_i$.
	Because of Lemma~\ref{le:strict}, we have $s_i\in W_i$,
	and therefore $W_i$ is nonempty.
	Because of Lemma~\ref{le:starshaped}, the sets $U_i$ and $W_i$
	induce subtrees of $T$. Since $W_i\subsetneq U_i$, it follows that
	$T$ has some edge from $W_i$ to $U_i\setminus W_i$, and therefore $E_i$ is nonempty.
\end{proof}

As a preprocessing step, we replace $S_i$ by $S_i\cap W_i$ for each $i\in [k]$.
Since a site cannot belong to two Voronoi regions, this replacement
does not reduce the set of feasible solutions for $I$. 
To simplify notation, we keep using $I$ for the new instance.
We check that,
for each $i\in [k]$, the set $S_i$ is nonempty and
the sets $U_i$ and $W_i$ induce a connected subgraph of $T$.
If any of those checks fail, we correctly report that there is no solution to $I$.

If the sets $U_1,\dots,U_k$ are pairwise disjoint, we do not need to do anything.
If at least two of them overlap but the sets $E_1,\dots,E_k$ are empty, then
Lemma~\ref{le:transform0} implies that there is no solution.
In the remaining case some $E_i$ is nonempty, and we transform the instance as follows.

In the transformations we will need ``short'' edges.
To quantify this, we introduce the \DEF{resolution} $\res(I)$
of an instance $I$, defined by
\[
	\res(I) ~=~ \min \left( \RR_{>0}\cap \{ d_T(s_i,u) - d_T(s_j,u)\mid
						u\in U_i\cap U_j,~ s_i\in S_i, ~s_j\in S_j,~ i,j\in [k] \}\right).
\]
Here we take the convention that $\min(\emptyset)=+\infty$.
From the definition we have the following property:
\begin{align}
\begin{split}\label{eq:resolution}
	\forall i,j\in [k],~ u\in U_i\cap U_j,~ & s_i\in S_i,~ s_j\in S_j:\\
			& |d_T(s_i,u)-d_T(s_j,u)| < \res(I) ~ \Longrightarrow ~
									d_T(s_i,u)=d_T(s_j,u).
\end{split}
\end{align}

Consider any value $\eps>0$.
Fix any index $i\in [k]$ such that $E_i\neq \emptyset$ and
consider an edge $xy\in E_i$ with $x\in W_i$ and $y\in U_i\setminus W_i$.
By renaming the sets, if needed, we assume henceforth that $i=1$, that is,
$E_1\neq \emptyset$, $x\in W_1$ and $y\in U_1\setminus W_1$.
We build a tree $T'$ with edge-lengths $\lambda'$ and a new set $U'_1$ as follows.
We obtain $T'$ from $T$ by subdividing $xy$ with a new vertex $y'$.
We define $U'_1$ to be the subset of vertices of $U_1$ that belong to the component
of $T-y$ that contains $x$, and then we also add $y'$ into $U'_1$.
Note that $u\in U_1$ belongs to $U'_1$ if and only if
$d_T(u,x)<d_T(u,y)$.
In particular, $y\notin U'_1$.
Finally, we set the edge-lengths $\lambda' (xy')=\lambda(xy)$ and $\lambda'(yy')= \eps$,
and the remaining edges have the same length as in $T$.
This completes the description of the transformation.
Note that $T'$ is just a subdivision of $T$ and, effectively,
the edge $xy$ became a $2$-edge path $xy'y$ that is longer by $\eps$.
All distances in $T'$ are larger or equal than in $T$, and the difference is at most $\eps$.

Let $I'$ be the new instance, where we use $T'$, $\lambda'$ and $U'_1$,
instead of $T$, $\lambda$ and $U_1$, respectively.
(We leave $U_i$ unchanged for each $i\in [k]\setminus \{1 \}$
and $S_i$ unchanged for each $i\in [k]$.)
See Figure~\ref{fig:disjoint2} for two examples of this transformation and
Figure~\ref{fig:transform1} for a schematic view.
We call $I'$ the instance obtained from $I$ by \DEF{expanding the edge $xy$ from $E_1$ by $\eps$}.
Note that $y'$ is not a valid placement for a site in $I'$, since $y'\notin S_1$.

Our definition of $\res(I)$ is carefully chosen so that it does not decrease 
with the expansion of an edge. That is, $\res(I')\ge \res(I)$. 
(This property is exploited in the proof of Lemma~\ref{le:transform_disjoint}.)
This is an important but subtle point needed to achieve efficiency. 
It will permit that all the short edges 
that are introduced during the transformations have the same small length $\eps$,
and we will be able to treat $\eps$ symbolically.

The next two lemmas show the relation between solutions to the instances $I$ and $I'$.

\begin{figure}
\centering
	\includegraphics[page=2,width=\textwidth]{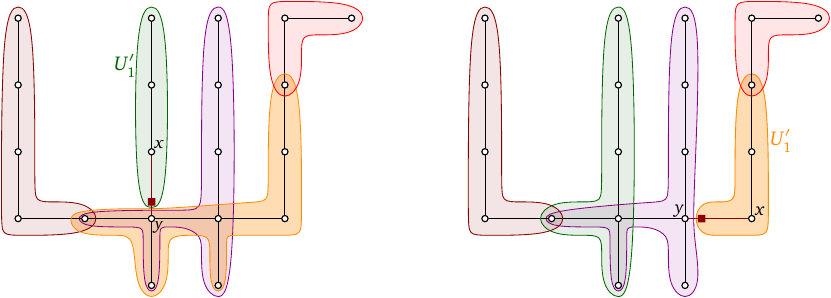}
	\caption{The transformation from the instance $I$ in Figure~\ref{fig:disjoint}
		to $I'$ for two different choices of the set $U_1$ and $xy\in E_1$.
		The new vertex $y'$ appearing because of the subdivision is marked with a square.
		The ``shorter'' edges in the drawing
		have length $\eps$; all other edges have unit length.}
	\label{fig:disjoint2}
\end{figure}

\begin{lemma}
\label{le:transform1}
	Suppose that $\eps>0$.
	If $\Sigma$ is a solution to \textsc{Generalized Graphic Inverse Voronoi in Trees}
	with input $I$, then $\Sigma$ is also a solution to
	\textsc{Generalized Graphic Inverse Voronoi in Trees} with input $I'$. 	
\end{lemma}
\begin{proof}
	We first introduce some notation.
	Let $V_x$ be the vertex set of the component of $T'-y'$ that contains $x$ and
	let $V_y$ be the vertex set of the component of $T'-y'$ that contains $y$.
	See Figure~\ref{fig:transform1}.
	Note that $x\in V_x$ and $y\in V_y$, while $y'$ is neither in $V_x$ nor in $V_y$.
	From the definition of $U'_1$, we have $U'_1=\{ y'\}\cup (V_x\cap U_1)$ and
	$U_1\setminus U'_1= V_y\cap U_1$. 

	\begin{figure}
	\centering
		\includegraphics[page=6,scale=1.2]{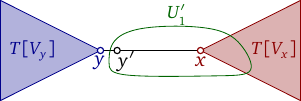}
		\caption{Notation in the proof of Lemma~\ref{le:transform1}.}
		\label{fig:transform1}
	\end{figure}

	We have the following easy relations between distances in $T$ and $T'$;
	we will use them often without explicit reference.
	\begin{align*}
		\forall u,v\in V_x&: ~~~ d_{T'}(u,v)= d_T(u,v)\\
		\forall u,v\in V_y&: ~~~ d_{T'}(u,v)= d_T(u,v)\\
		\forall u\in V_x,~ v\in V_y&: ~~~ d_{T'}(u,v)= d_T(u,v)+\eps\\
		\forall u\in V_x&: ~~~ d_{T'}(u,y')= d_T(u,y)\\
		\forall u\in V_y&: ~~~ d_{T'}(u,y')= d_T(u,y)+\eps.
	\end{align*}
	
	Consider a solution $s_1,\dots,s_k$	to
	\textsc{Generalized Graphic Inverse Voronoi in Trees} with input $I$,
	and define $\Sigma=\{ s_1,\dots, s_k\}$.
	This means that, for all $i\in [k]$,
	we have $s_i\in S_i$ and $U_i=\cell_T(s_i,\Sigma)$.
	Our objective is to show that $U'_1=\cell_{T'}(s_1,\Sigma)$
	and $U_i=\cell_{T'}(s_i,\Sigma)$ for all $i\in [k]\setminus \{ 1 \}$.
	
	Since $U_i=\cell_T(s_i,\Sigma)$ for all $i\in [k]$,
	Lemma~\ref{le:strict} implies that 
	$W_1 = \cell^<_T(s_1,\Sigma)$ and $s_1\in W_1$.
	Since $x\in W_1$, $y\not\in W_1$, 
	and $W_1$ induces a connected subgraph of $T$ 
	because of Lemma~\ref{le:transform0}(a),
	the set $W_1$ is contained in $V_x$.
	Since $W_1\subseteq V_x$ and $W_1\subseteq U_1$,
	we have $W_1\subseteq V_x\cap U_1$ and we 
	conclude that $W_1\subseteq U'_1$. 
	Furthermore, because $s_1\in \cell^<_T(s_1,\Sigma) = W_1$
	and $W_1\subseteq V_x$, 
	we obtain that $s_1\in V_x$. 

	For each $i\in [k]\setminus \{ 1 \}$,
	we have $x\notin U_i$ because $x\in W_1$, 
	and Lemma~\ref{le:transform0}(a) implies that
	the set $U_i=\cell_T(s_i,\Sigma)$ is fully contained either in $V_x$ or in $V_y$.
	
	Consider any index $\ell\in [k]\setminus \{1 \}$ with the property
	that $y\in U_1\cap U_\ell$.
	Since $U_\ell$ contains $y$, it cannot be that $U_\ell\subseteq V_x$,
	and therefore $U_\ell\subseteq V_y$. In particular, $s_\ell\in V_y$.

	We first note that the sets $U'_1, U_2,\dots, U_k$ cover $V(T')$.
	Indeed, since $y\in U_1\cap U_\ell$, 
	the sites $s_1$ and $s_\ell$ are closest sites to $y$ in $T$,
	and using that $s_1\in V_x$ and $s_\ell\in V_y$, we obtain that
	$U_1\setminus U'_1$ is contained in $U_\ell$.
	Further, since $U_1,\dots,U_k$ cover $V(T)$, $y'\in U'_1$ by construction,
	and $V(T')=V(T)\cup \{ y'\}$, we conclude that indeed
	$U'_1, U_2, \dots, U_k$ cover $V(T')$.

	Next, we make the following two claims.
	
	\begin{claim}\label{claim1}
	$y'\in \cell_{T'}(s_1,\Sigma)$ and
	$y'\notin \cell_{T'}(s_i,\Sigma)$ for any $i\in [k]\setminus\{ 1\}$.
	\end{claim}
	\begin{proof}
	Fix any index $i\in [k]\setminus \{1\}$.
	Consider first the case when $s_i\in V_x$.
	In this case the path from $s_i$ to $y'$ passes through $x$,
	which is a vertex in $\cell^<_T(s_1,\Sigma)$.
	It follows that $d_T(s_1,x)<d_T(s_i,x)$, which
	implies
	\[
		d_{T'}(s_1,y') = d_T(s_1,y) < d_T(s_i,y) = d_{T'}(s_i,y').
	\]
	
	Consider now the case when $s_i\in V_y$.
	Because $y\in U_1=\cell(s_1,\Sigma)$, we have $d_T(s_1,y)\le d_T(s_i,y)$
	and we conclude that
	\[
		d_{T'}(s_i,y')= d_T(s_i,y)+\eps \ge d_T(s_1,y) +\eps = d_{T'}(s_1,y') +\eps > d_{T'}(s_1,y').
	\]
	
	In each case we get $d_{T'}(s_1,y')<d_{T'}(s_i,y')$, and the claim follows.
	\end{proof}

	\begin{claim}\label{claim2}
	$y\notin \cell_{T'}(s_1,\Sigma)$.
	\end{claim}
	\begin{proof}
	Since $y$ belongs to $U_1\cap U_\ell$,
	we have $d_T(s_1,y) = d_T(s_\ell,y)$.
	Using that $U_\ell$ is contained in $V_y$, and thus $s_\ell\in V_y$, we have
	\[
		d_{T'}(s_\ell,y) = d_T(s_\ell,y) = d_T(s_1,y) = d_{T'}(s_1,y)-\eps < d_{T'}(s_1,y).
	\]
	We conclude that $y$ is not an element of $\cell_{T'}(s_1,\Sigma)$.
	\end{proof}
	
	Claims~\ref{claim1} and~\ref{claim2} imply that $y'$ belongs \emph{only} to the
	Voronoi region $\cell_{T'}(s_1,\Sigma)$ and $y$ does not belong to $\cell_{T'}(s_1,\Sigma)$.
	This means that each vertex of $V_x$ belongs only to
	some regions $\cell_{T'}(s_i,\Sigma)$ with $s_i\in V_x$ 
	and each vertex of $V_y$ belongs
	to some regions $\cell_{T'}(s_i,\Sigma)$ with $s_i\in V_y$.
	That is, it cannot be that some vertex $u\in V_x$ belongs to $\cell_{T'}(s_i,\Sigma)$
	with $s_i\in V_y$ and it cannot be that some vertex $u\in V_y$ belongs to
	$\cell_{T'}(s_i,\Sigma)$ with $s_i\in V_x$. 
	Effectively, this means that $y'$
	splits the Voronoi diagram $\VV_{T'}(\Sigma)$ into the part within $T'[V_x]$
	and the part within $T'[V_y]$, with the gluing
	property that $y'\in \cell_{T'}(s_1,\Sigma)$.
	Since $U'_1\setminus \{y'\}=U_1\cap V_x$ and
	the distances within $T'[V_x]$ and within $T'[V_y]$ are the same as in $T$,
	the result follows.
\end{proof}

The converse property is more complicated. We need $\eps$ to be small 
enough and we also have to assume that $I$ has a solution. 
It is this tiny technicality that makes the reduction nontrivial.

\begin{lemma}
\label{le:transform2}
	Suppose that $0< \eps < \res(I)$ and the answer to
	\textsc{Generalized Graphic Inverse Voronoi in Trees} with input $I$ is ``yes''.
	If $\Sigma'$ is a solution to \textsc{Generalized Graphic Inverse Voronoi in Trees}
	with input $I'$, then $\Sigma'$ is also a solution to
	\textsc{Generalized Graphic Inverse Voronoi in Trees} with input $I$. 	
\end{lemma} 	
\begin{proof}		
	When the instance $I$ has \emph{some} solution,
	then the properties discussed in Lemmas~\ref{le:transform0} 
	and~\ref{le:transform1} hold.
	We keep using the notation and the properties established earlier.
	In particular, each set $U_i$ ($i\in [k]\setminus \{1\}$)
	is contained either in $V_x$ or in $V_y$, 
	and we have $W_1\subseteq U'_1 \subseteq V_x\cup \{ y'\}$.
	
	Consider a solution $s_1,\dots, s_k$
	to \textsc{Generalized Graphic Inverse Voronoi in Trees} with input $I'$,
	and set $\Sigma=\{ s_1,\dots,s_k\}$.
	This means that $U'_1=\cell_{T'}(s_1,\Sigma)$ and,
	for all $i\in [k]\setminus \{1\}$, we have $U_i=\cell_{T'}(s_i,\Sigma)$.
	We have to show that, for all $i\in [k]$, we have $U_i=\cell_T(s_i,\Sigma)$,
	which implies that $\Sigma$ is a solution to input $I$.
	
	Like before, we split the proof into claims that show that
	$\Sigma$ is a solution to \textsc{Generalized Graphic Inverse Voronoi in Trees}
	with input $I$. We start with an auxiliary property that plays a key role.

	\begin{claim}
	\label{cl:A}
	For each $i\in [k]$, we have $y\in U_i$ if and only if $y\in \cell_T(s_i,\Sigma)$.
	\end{claim}
	\begin{proof}
	Suppose first that $y\in U_i$ and $i\ne1$. Then $U_i\subseteq V_y$.
	Since $y\in U_i=\cell_{T'}(s_i,\Sigma)$ and $y\notin U'_1=\cell_{T'}(s_1,\Sigma)$, we have
	\begin{align}
	\label{eq1}
		d_T(s_i,y) = d_{T'}(s_i,y)< d_{T'}(s_1,y) = d_T(s_1,y)+\eps .
	\end{align}
	Since $y'\notin U_i=\cell_{T'}(s_i,\Sigma)$ and $y'\in U'_1=\cell_{T'}(s_1,\Sigma)$, we have
	\begin{align}
	\label{eq2}
		d_T(s_1,y) = d_{T'}(s_1,y') < d_{T'}(s_i,y') = d_T(s_i,y)+\eps .
	\end{align}
	Combining~\eqref{eq1} and~\eqref{eq2} we get
	\[
		|d_T(s_i,y) - d_T(s_1,y)| < \eps<  \res(I).
	\]
	From property~\eqref{eq:resolution} and since $y\in U_1\cap U_i$, we conclude that $d_T(s_1,y) = d_T(s_i,y)$.
	For each $s_j\in V_y$ we use that $y\in U_i=\cell_{T'}(s_i,\Sigma)$ to obtain
	\[
		d_T(s_j,y) = d_{T'}(s_j,y) \ge d_{T'}(s_i,y)=d_T(s_i,y).
	\]
	For each $s_j\in V_x$ we use that the path from $s_j$ to $y$
	goes through $x\in U'_1=\cell_{T'}(s_1,\Sigma)$ to obtain
	\[
		d_T(s_j,y) \ge d_T(s_1,y) = d_T(s_i,y).
	\]	
	We conclude that for each $j\in [k]$ we have $d_T(s_j,y) \ge d_T(s_i,y)$,
	and therefore $y\in \cell_T(s_i,\Sigma)$.
	
	We know there exists some $\ell \in [k]\setminus \{1\}$ such that $y \in U_\ell$. 
	According to the analysis above, $y \in \cell_T(s_\ell, \Sigma)$ and $d_T(s_1,y)=d_T(s_\ell,y)$. 
	It follows that $y \in \cell_T(s_1, \Sigma)$.
	With this we have shown one direction of the implication.
	
	To show the other implication, consider some index $i\in [k]$
	such that $y\in \cell_T(s_i,\Sigma)$. If $i=1$, then $y\in U_1$ by construction,
	and the implication holds.
	So we consider the case when $i\neq 1$.	
	First we show that it cannot be that $s_i\in V_x$.
	Assume, for the sake of reaching a contradiction, that $s_i \in V_x$.
	Because of the implication left-to-right
	that we showed, we have $y\in \cell_T(s_1,\Sigma)$.
	Since we have $y\in \cell_T(s_i,\Sigma)$ and $y\in \cell_T(s_1,\Sigma)$,
	we obtain $d_T(s_i,y)=d_T(s_1,y)$. 
	Because $s_1,s_i\in V_x$, we obtain $d_T(s_i,x)=d_T(s_1,x)$ and
	therefore $d_{T'}(s_i,x)=d_{T'}(s_1,x)$. 
	Further, since $x\in U'_1=\cell_{T'}(s_1,\Sigma)$, 
	we get $x\in \cell_{T'}(s_i,\Sigma)=U_i$, 
	which implies $x\notin W_1$.
	We conclude that it must be $s_i\notin V_x$, and thus $s_i \in V_y$.

	Take an index $\ell\in [k]\setminus \{ 1\}$ such that $y\in U_\ell$.
	Such an index exists because $y\notin W_1$. We have $U_\ell\subseteq V_y$
	and thus $s_\ell\in V_y$.
	Because of the implication left-to-right
	that we showed, we have $y\in \cell_T(s_\ell,\Sigma)$.
	Since we have $y\in \cell_T(s_i,T)$ and $y\in \cell_T(s_\ell,\Sigma)$,
	we obtain $d_T(s_i,y)=d_T(s_\ell,y)$.
	Because $s_i,s_\ell \in V_y$ we then have
	\[
		d_{T'}(s_i,y)= d_T(s_i,y)=d_T(s_\ell,y)=d_{T'}(s_\ell,y).
	\]
	Since $d_{T'}(s_i,y)=d_{T'}(s_\ell,y)$ and $y\in U_{\ell}=\cell_{T'}(s_\ell,\Sigma)$,
	we conclude that $y\in \cell_{T'}(s_i,\Sigma)=U_i$.
	\end{proof}
	
	\begin{claim}
	\label{cl:B}
	$x\in\cell_T(s_1,\Sigma)$ and $x\notin \cell_T(s_i,\Sigma)$ for any $i\in [k]\setminus\{ 1\}$.
	\end{claim}	
	\begin{proof}
	Since $x\in U'_1=\cell_{T'}(s_1,\Sigma)$ and $x\notin U_i=\cell_{T'}(s_i,\Sigma)$
	for any $i\in [k]\setminus\{ 1\}$, we have
	\[
		\forall i\in [k]\setminus\{ 1\}:~~~ d_{T'}(s_1,x)< d_{T'}(s_i,x).
	\]
	We then have
	\begin{align}
	\label{eq:B1}
		\forall s_i\in V_x,~ s_i\neq s_1:~~~
			d_T(s_1,x)= d_{T'}(s_1,x)< d_{T'}(s_i,x)=d_T(s_i,x).
	\end{align}
	
	For each $s_i\in V_y$, note that the path from $s_i$ to $x$ passes through $y$,
	and $y\in \cell_T(s_1,\Sigma)$ because of Claim~\ref{cl:A}. Using that $s_1\in V_x$, 
	we have
	\begin{align}
	\label{eq:B2}
		\forall s_i\in V_y:~~~
			d_T(s_1,x)<d_T(s_i,x).
	\end{align}
	Combining~\eqref{eq:B1} and~\eqref{eq:B2}, the claim follows.
	\end{proof}
	
	\begin{claim}
	\label{cl:C}
	For each $u\in V_y$, we have $u\in U_1$ if and only if $u\in \cell_T(s_1,\Sigma)$.
	\end{claim}	
	\begin{proof}
	Consider \emph{some} solution $s^*_1,\dots, s^*_k$
	to \textsc{Generalized Graphic Inverse Voronoi in Trees} with input $I$,
	and set $\Sigma^*=\{ s^*_1,\dots,s^*_k\}$.
	This means that $U_i=\cell_T(s^*_i,\Sigma^*)$ for each $i\in [k]$.
	We also fix an index $\ell\in [k]\setminus \{ 1 \}$ such that $y\in U_\ell\cap U_1$.
	Recall that $U_\ell\subseteq V_y$ because $x\notin U_\ell$,
	and $W_1\subseteq V_x$ because $x\in W_1$ and $y\notin W_1$.
	Using Claim~\ref{cl:A} and using that $\Sigma^*$ is a solution to $I$ we have
	\begin{align}
	\label{eq:C1}
		d_T(s_1,y)=d_T(s_\ell,y)~~\text{ and }~~ d_T(s^*_1,y)=d_T(s^*_\ell,y).
	\end{align}
	
	Consider some $u\in U_1\cap V_y$. We will show that $u\in \cell_T(s_1,\Sigma)$.
	Consider the subtree $\tilde T$ defined by the paths connecting the vertices
	$s_1,s^*_1,s_\ell,s^*_\ell,u$. See Figure~\ref{fig:cl:C}.
	The path from $u$ to $s^*_1$ attaches to the path from $s_\ell^*$ to $y$
	at the vertex $y$. Indeed, if it attaches at
	another vertex $a\neq y$, then we would have $d_T(s^*_\ell,a)< d_T(s^*_1,a)$
	because of~\eqref{eq:C1}, which would imply $d_T(s^*_\ell,u)< d_T(s^*_1,u)$,
	contradicting the assumption that $u\in \cell_T(s^*_1,\Sigma^*)=U_1$.
	Because $W_\ell$ does not contain $y$ and $W_\ell$ is a
	connected subgraph of $T$ (applying Lemma~\ref{le:starshaped}), 
	$W_\ell$ is contained in a connected component of $T-y$. 
	Further since $W_\ell$ contains $s_\ell$ and $s^*_\ell$,
	and we replaced $S_\ell$ with $S_\ell\cap W_\ell$ in the preprocessing step	
	\footnote{
		Without the replacement $S_\ell$ with $S_\ell\cap W_\ell$, 
		the lemma is actually not true because it can happen
		that $s_\ell\in U_1\cap U_\ell$. 
		Indeed, we could have $s_\ell\in S_\ell\cap U_\ell\cap U_1$, which is not
		a valid placement in $I$ but would be a valid placement in $I'$.
	\label{note1}}, $s_\ell$ and $s^*_\ell$ are in the same component of $T-y$.
	Therefore, the $(u,s_1)$-path attaches to the $(s_\ell,y)$-path
	at the vertex $y$.
	
	\begin{figure}
	\centering
		\includegraphics[page=3,scale=1.2]{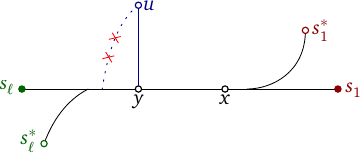}
		\caption{Situation in the proof of Claim~\ref{cl:C}.}
		\label{fig:cl:C}
	\end{figure}

	Since each path from $s_1$, $s^*_1$, $s_\ell$ and $s^*_\ell$ to $u$ passes
	through $y$, from~\eqref{eq:C1} we get
	\begin{align}
	\label{eq:C2}
		d_T(s_1,u)=d_T(s_\ell,u)~~\text{ and }~~ d_T(s^*_1,u)=d_T(s^*_\ell,u).
	\end{align}
	Together with $u\in U_1=\cell_T(s^*_1,\Sigma^*)$ we
	conclude that $u\in \cell_T(s^*_\ell,\Sigma^*)=U_\ell$.
	Since $u\in U_\ell=\cell_{T'}(s_\ell,\Sigma)$ we have
	\[
		\forall s_j\in V_y:~~~ d_T(s_1,u)=d_T(s_\ell,u)\le d_T(s_j,u).
	\]
	Together with the fact that each $s_j\in V_x$ is no closer to $u$ than $s_1$
	because $x\in \cell_T(s_1,\Sigma)$, we conclude that $u\in \cell_T(s_1,\Sigma)$.
	This finishes the left-to-right direction of the implication.	
		
	Consider now a vertex $u\in V_y\cap \cell_T(s_1,\Sigma)$. Since $y$ is on the path
	from $s_1$ to $u$, we obtain from~\eqref{eq:C1} that $d_T(s_\ell,u)\le d_T(s_1,u)$,
	and therefore $u\in \cell_T(s_\ell,\Sigma)$. Because $u\in V_y$, 
	$d_{T'}(s_\ell,u)=d_T(s_\ell,u)$, and distances in $T'$ can only be larger
	than in $T$,
	we have $u\in \cell_{T'}(s_\ell,\Sigma)=U_\ell = \cell_{T}(s^*_\ell,\Sigma^*)$.
	This means that
	\begin{align}
	\label{eq:C3}
		\forall i\in [k]:~~~ d_T(s^*_\ell,u)\le d_T(s^*_i,u).
	\end{align}
	Since $u\in \cell_T(s_1,\Sigma)$, $u\in \cell_T(s_\ell,\Sigma)$
	and $d_T(s_1,y)=d_T(s_\ell,y)$, the vertex $y$ is on the path from 
	$s_\ell$ to $u$.
	Note that the vertices $s_\ell$ and $s^*_\ell$ must be contained in the 
	same component of $T-y$ because the $(s_\ell,s^*_\ell)$-path 
	must be contained $W_\ell$  
	(Lemma~\ref{le:starshaped} and footnote~\ref{note1}), but $y\notin W_\ell$.
	This implies that $y$ is also on the path from $s^*_\ell$ to $u$.
	Since $y$ is also on the path from $s^*_1$ to $u$,
	we get from~\eqref{eq:C1} and~\eqref{eq:C3} that
	\[
		\forall i\in [k]:~~~ d_T(s^*_1,u)=d_T(s^*_\ell,u)\le d_T(s^*_i,u).
	\]	
	It follows that	$u\in \cell_T(s^*_1,\Sigma^*)=U_1$.
	This finishes the proof of the claim.
	\end{proof}

	We are now ready to prove Lemma~\ref{le:transform2}:
	for all $i\in [k]$ we have $U_i=\cell_T(s_i,\Sigma)$.
	Consider first the case $i=1$.
	Because of Claim~\ref{cl:C} we have
	$V_y\cap U_1 = V_y\cap \cell_T(s_1,\Sigma)$.
	We will show that $V_x\cap U_1 = V_x\cap \cell_T(s_1,\Sigma)$,
	which implies that $U_1=\cell_T(s_1,\Sigma)$ because $V_x\cup V_y=V(T)$.
	Recall that $V_x\cap U_1=U'_1\setminus \{y' \}$.
	Consider any vertex $u\in V_x\cap U_1$.
	Because of Claim~\ref{cl:B} we have $x\in \cell^<_T(s_1,\Sigma)$,
	and therefore $u\in V_x$ implies
	\[
		\forall s_j\in V_y:~~~ d_T(s_1,u)< d_T(s_j,u).
	\]	
	On the other hand, since $u\in U'_1=\cell_{T'}(s_1,\Sigma)$ we have
	\[
		\forall s_j\in V_x:~~~ 
		d_T(s_1,u)= d_{T'}(s_1,u) \le d_{T'}(s_j,u) = d_T(s_j,u).
	\]	
	We conclude that $d_T(s_1,u)\le d_T(s_j,u)$ for all $s_j\in \Sigma$,
	and therefore $u\in \cell_T(s_1,\Sigma)$.
	This shows that $V_x\cap U_1 \subseteq V_x\cap \cell_T(s_1,\Sigma)$.
	To show the inclusion in the other direction, 
	consider any $u\in V_x\cap \cell_T(s_1,\Sigma)$.
	We then have
	\[
		\forall s_j\in \Sigma:~~~
			d_{T'}(s_1,u) = d_T(s_1,u) \le d_T(s_j,u) \le d_{T'}(s_j,u),
	\]
	which implies $u\in \cell_{T'}(s_1,\Sigma) = U'_1$.
	It follows that $u\in V_x\cap U'_1= V_x\cap U_1$
	This finishes the proof of $U_1=\cell_T(s_1,\Sigma)$, that is, the case $i=1$.
	
	Consider now the indices $i\in [k]\setminus \{ 1\}$ with $s_i\in V_y$.
	Recall that we have $U_i=\cell_{T'}(s_i,\Sigma)$ and $U_i\subseteq V_y$.
	Fix an index $\ell\in [k]\setminus \{ 1 \}$ such that $y\in U_\ell\cap U_1$.
	Such an index exists because $y\notin W_1$.
	We must have $U_\ell\subseteq V_y$ because $x\notin U_\ell$,
	and thus $s_\ell \in V_y$.
	Because of Claim~\ref{cl:A} we have $d_T(s_1,y)=d_T(s_\ell,y)$,
	and using that $x\in \cell^<_T(s_1,\Sigma)$, implied by Claim~\ref{cl:B}, we get
	\begin{equation}
	\label{eq:another}
		\forall u\in V_y,~ s_j\in V_x:~~~
			d_T(s_\ell,u) \le d_T(s_1,u) \le d_T(s_j,u).
	\end{equation}
	This implies that in $T$ each vertex of $V_y$ has at least one closest site (from $\Sigma$)
	that belongs to $V_y$.
	Moreover, because $x\in \cell^<_T(s_1,\Sigma)$, no vertex of $V_x$ is closest to
	a site in $V_y$.
	Therefore, for each site $s_i\in V_y$, we have 
	\begin{equation}
	\label{eq:T}
		\cell_T(s_i,\Sigma) = \cell_{T[V_y]}(s_i,\Sigma\cap V_y).
	\end{equation}
	A similar argument can be used for $T'$, as follows.
	From \eqref{eq:another} we get
	\[
		\forall u\in V_y,~ s_j\in V_x:~~~
			d_{T'}(s_\ell,u) = d_T(s_\ell,u) \le d_T(s_1,u) \le d_T(s_j,u)< d_{T'}(s_j,u).
	\]
	This implies that in $T'$ each vertex of $V_y$ has all closest sites (from $\Sigma$)
	in $V_y$.
	Moreover, because $y'\in \cell^<_{T'}(s_1,\Sigma)$, no vertex of $V_x\cup\{ y'\}$ 
	is closest to a site in $V_y$.
	Therefore, for each site $s_i\in V_y$, we have
	\begin{equation}
	\label{eq:T'}
		\cell_{T'}(s_i,\Sigma) = \cell_{T'[V_y]}(s_i,\Sigma\cap V_y).
	\end{equation}
	Noting that $T[V_y]=T'[V_y]$, we use \eqref{eq:T} and \eqref{eq:T'} to obtain, 
	for each site $s_i\in V_y$,
	\[
		U_i = \cell_{T'}(s_i,\Sigma) = \cell_{T'[V_y]}(s_i,\Sigma\cap V_y) = 
		\cell_{T[V_y]}(s_i,\Sigma\cap V_y) 
		= \cell_T(s_i,\Sigma).
	\]
	
	It remains to consider the indices $i\in [k]\setminus \{ 1\}$ with $s_i\in V_x$.
	The approach is similar, and actually simpler because $x\in \cell^<_T(s_1,T)$
	implies that there is no influences from the sites $V_y$.
	(No care is needed for $y'$ because it belongs to $\cell^<_{T'}(s_1,\Sigma)$.
	Therefore, for each $s_i\in V_x\setminus \{ s_1 \}$,
	\[
		U_i = \cell_{T'}(s_i,\Sigma) = \cell_{T'[V_x]}(s_i,\Sigma\cap V_x) = 
		\cell_{T[V_x]}(s_i,\Sigma\cap V_x) 
		= \cell_T(s_i,\Sigma).
	\]
	
	We have covered all the cases: $s_i=s_1$, $s_i\in V_y$, 
	and $s_i\in V_x\setminus \{ s_1 \}$.
	This finishes the proof of the Lemma.
\end{proof}

It is important to note that the transformation described above only works for trees.
A similar transformation for arbitrary graphs may have feasible solutions that do not
correspond to solutions in the original problem. See Figure~\ref{fig:disjoint4} for
a simple example.

\begin{figure}
\centering
	\includegraphics[page=4]{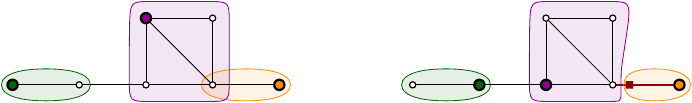}
	\caption{A similar transformation for arbitrary graphs does not work.
		On the right side we have the transformed instance with a feasible solution
		that does not correspond to a solution in the original setting.}
	\label{fig:disjoint4}
\end{figure}

Another important point is that we need the assumption that $I$ had a solution. This means
that, any solution $\Sigma'$ we obtain after making a sequence of expansions, has
to be tested in the original instance. However, if $\Sigma'$ is not a valid solution in $I$,
then $I$ has no solution.

\medskip

We are going to make a sequence of edge expansions.
The replacement of $S_i$ with $S_i\cap W_i$ (for $i\in [k]$) 
needs to be made only at the preprocessing step and it is important 
for correctness (see footnote~\ref{note1}).
It is not needed later on because with each edge expansion 
the sets $W_i$ (for $i\in [k]$) can only increase.

Consider an instance $I=(T,((U_1,S_1),\dots,(U_k,S_k)))$.
Set $I_0=I$ and define, for $t\ge 1$, the instance $I_t$
by transforming $I_{t-1}$ using an expansion of some edge.
For all expansions we use the same parameter $\eps$.
We finish the sequence when we obtain the first instance
$\tilde I=(\tilde T,((\tilde U_1,\tilde S_1),\dots,(\tilde U_k,\tilde S_k)))$
such that the sets $\tilde U_1,\dots,\tilde U_\ell$ are pairwise disjoint.
Note that this procedure stops because the number of pairs $(i,j)$ with $U_i\cap U_j \ne\emptyset$
decreases with each expansion. This implies that the number of steps is at most $\binom{k}{2}$.
In fact, the number of steps is even smaller.

\begin{lemma}\label{le:number of steps}
	$\tilde I$ is reached after at most $k-1$ edge expansions.
\end{lemma}
\begin{proof}
	We prove this by induction on $k$. There is nothing to show if $k=1$. 
	Otherwise, note that the sets $U_i$ in $V_x$ and those in $V_y$ (respectively) 
	give rise to two independent subproblems with $k_x$ and $k_y$
	sites (respectively), where $k_x+k_y=k$. 
	By induction, the number of edge expansions is at most $1+(k_x-1)+(k_y-1)=k-1$.
\end{proof}

The next lemma shows that using the same parameter $\eps$ for all edge expansions
is a correct choice.
This is due to our careful definition of resolution $\res(\cdot)$.

\begin{lemma}
\label{le:transform_disjoint}
	Assume that $0<\eps<\res(I)$ and the answer to
    \textsc{Generalized Graphic Inverse Voronoi in Trees} with input $I$ is ``yes''.
	Then $\Sigma$ is a solution to \textsc{Generalized Graphic Inverse Voronoi in Trees}
	with input $I$ if and only if $\Sigma$ is also a solution to
	\textsc{Generalized Graphic Inverse Voronoi in Trees} with input $\tilde I$.
\end{lemma}
\begin{proof}
	Note that, by construction, $\res(I_{t-1})\le \res(I_t)$ for all $t\ge 1$.
	Indeed, when we expand the edge $xy$ inserting $y'$, then there is no
	set $U_i$ that is on both sides of $T'-y'$. This means that for all
	the parameters $s_i,s_j,u$ considered in the definition of $\res(I_t)$
	we have $d_{T'}(s_i,u) - d_{T'}(s_j,u)= d_T(s_i,u) - d_T(s_j,u)$.	
    Therefore, $\eps<\res(I_t)$ for all $t$.
    The claim now follows easily from
    Lemmas~\ref{le:transform1} and \ref{le:transform2}
    by induction on $t$.
\end{proof}

%%%%%%%%%%%%%%%%%%%%%%%%%%%%%%%%%%%%%%%%%%%%%%%%%%%%%%%%%%%%%%%%%
\subsection{Transforming to maximum degree 3}
\label{sec:degree3}

Consider an instance $I=(T,((U_1,S_1),\dots, (U_k,S_k)))$
for the problem \textsc{Generalized Graphic Inverse Voronoi in Trees}, where $T$ is a tree and the sets
$U_1,\dots, U_k$ are pairwise disjoint.
We assume that each $U_i$ induces a connected subgraph in $T$.
See Figure~\ref{fig:maxdegree3} for an example of such an instance viewed around a vertex of degree $>3$.
We want to transform it into another instance $I'=(T',((U'_1,S'_1),\dots, (U'_k,S'_k)))$
where the maximum degree of $T'$ is $3$,
the sets $U'_1,\dots, U'_k$ are pairwise disjoint,
and a solution to $I'$ corresponds to a solution of $I$.

\begin{figure}
\centering
	\includegraphics[page=5]{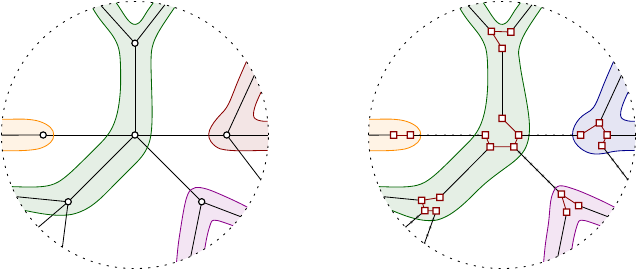}
	\caption{The behavior of the reduction to obtain maximum degree $3$.
			Left: part of an instance with a tree of arbitrary degrees.
			Right: result after the reduction for the left instance.
				The edges between different candidate Voronoi cells are shortened
				by $\delta$.}
	\label{fig:maxdegree3}
\end{figure}

In the transformations we will need ``short'' edges again and
we will shorten some edges.
We need \emph{another} version of the resolution:
\[
	\res'(I) ~=~ \min \left( \RR_{>0}\cap \{ d_T(v,u) - d_T(v',u)\mid
						v,v', u\in V(T) \}\right).
\]
In particular, $\res'(I)\le \lambda(uv)$ for all edges $uv$ of $T$.
From the definition we have the following property:
\begin{align}
\label{ineq:res0}
	\forall v,v',u\in V(T):~~~~
			& d_T(v,u)< d_T(v',u) ~ \Longrightarrow ~
									d_T(v,u) + \frac{\res'(I)}{2} < d_T(v',u).
\end{align}

We explain how to transform the instance into one where all
vertices have maximum degree $3$. We will use $T'$ and $\lambda'$
for the new graph and its edge-lengths. The construction uses
two values $\delta$ and $\delta'$, where
\[
	0 < \delta< \frac{\res'(I)}{6n} 
	~~\text{ and }~~
	\delta' = \frac{\delta}{4n} \, .
\]
The intuition is that edges connecting different candidate Voronoi cells
are shorten by $\delta$, and then we split the vertices of degree larger
than three using short edges of length $\delta'$, 
where $0<\delta'\ll \delta\ll \res'(I)$.

For each edge $uv$ of $T$ we place two vertices $a_{u,v}$ and $a_{v,u}$
in $T'$, and connect them with an edge. If $u$ and $v$ belong to
the same set $U_i$, then the length $\lambda'$
of such an edge $a_{u,v}a_{v,u}$ is set to $\lambda(uv)$.
If $u\in U_i$ and $v\in U_j$ with $i\neq j$, 
then the length $\lambda'$
of such an edge $a_{u,v}a_{v,u}$ is set to $\lambda(uv)-\delta>0$.
For each vertex $u$ of $T$, we connect the vertices
$\{ a_{u,v}\mid uv\in E(T)\}$ with a path.
The length $\lambda'$ of the edges on these $|V(T)|$ paths
is set to $\delta'$.
Finally, for each $i\in [k]$ we define the sets
\begin{align*}
	U'_i ~&=~ \{ a_{u,v}\mid u\in U_i, uv\in E(T)\},\\
	S'_i ~&=~ \{ a_{u,v}\mid u\in S_i, uv\in E(T)\}.	
\end{align*}
Note that the sets $U'_1,\dots, U'_k$ are pairwise disjoint.
For an example of the whole process see Figure~\ref{fig:maxdegree3}.

To recover the solutions, we define the projection map $\pi(a_{u,v})=u$.
Thus, $\pi$ sends each vertex of $T'$ to the corresponding vertex of $T$ that
was used to create it.
Note that for each $i\in[k]$ we have $\pi(S'_i)=S_i$ and $\pi(U'_i)=U_i$.

The distances in $T'$ and $T$ are closely related. Using
that the tree $T'$ has fewer than $2n$ new short edges of length $\delta'$
and the path connecting any two vertices of $U'_i$ is contained in $T'[U'_i]$
we get
\begin{align}
\begin{split}  
\label{ineq:res1}
	\forall i\in [k],~u,v\in U'_i:~~~ d_T(\pi(u),\pi(v)) \le d_{T'}(u,v) &<
							d_T(\pi(u),\pi(v))+ 2n\delta' 
							\\ &< d_T(\pi(u),\pi(v))+ \delta.
\end{split}
\end{align}
Using that the path between two vertices in different sets $U_i$ and $U_j$, 
$i\neq j$, uses at least one edge and at most $n-1$ edges
that have been shortened by $\delta$, we get
\begin{align}
\begin{split}
\label{ineq:res2}
	\forall i \neq j \in [k], ~u\in U'_i, ~v\in U'_j:~~ 
		d_T(\pi(u),\pi(v))-n\delta < d_{T'}(u,v) &<
							d_T(\pi(u),\pi(v)) - \delta + 2n\delta' \\ 
						 &< d_T(\pi(u),\pi(v)).
\end{split}
\end{align}

\begin{lemma}
\label{le:transform_3}
	Suppose that $0< \delta < \res'(I) / 6n$
	and the sets $U_1,\dots, U_k$ are pairwise disjoint subsets of $V(T)$
	that induce connected subtrees of $T$.
	The answer to $(T, ((U_1,S_1),\dots, (U_k,S_k)))$ is ``yes''
	if and only if the answer to $(T', ((U'_1,S'_1),\dots, (U'_k,S'_k) ))$ is ``yes''. 	
\end{lemma}
\begin{proof}
	We first show the ``if'' part. Suppose that the answer to $I'$ is ``yes''. Then, 
	there exist $s'_1, \ldots, s'_k$ with $s'_i \in S'_i$ 
	and $U'_i=\cell_{T'}(s'_i,\{ s'_1,\dots,s'_k\})$ for each $i\in [k]$.
	Set $s_i=\pi(s'_i)$ for all $i\in [k]$, $\Sigma'=\{s'_1,\dots, s'_k\}$
	and $\Sigma=\{s_1,\dots,s_k\}$.
	
	Consider any fixed $i\in [k]$ and any vertex $u\in U_i$. 
	There exists some vertex $u' \in U'_i$ such that $u=\pi(u')$.
	Since $u'\in U'_i=\cell_{T'}(s'_i,\Sigma')$ and 
	$u'\notin U'_j=\cell_{T'}(s'_j,\Sigma')$ for all $j\neq i$, we have
	\[
		\forall j\in [k]\setminus \{i \}:~~~ 
				d_{T'}(s'_i,u') < d_{T'}(s'_j,u').
	\]
	For $j\neq i$, since $u,s_i\in U_i$ and $s_j\notin U_i$ 
	we use the relations~\eqref{ineq:res1} and~\eqref{ineq:res2} to get 
	\[
		\forall j\in [k]\setminus \{i \}:~~~ 
				d_T(s_i,u) \le d_{T'}(s'_i,u') < d_{T'}(s'_j,u') < d_T(s_j,u).
	\]
	We conclude that $u\in \cell_T(s_i,\Sigma)$ and $u\notin \cell_T(s_j,\Sigma)$
	for all $j\in [k]\setminus\{ i\}$.
	It follows that $U_i=\cell_T(s_i,\Sigma)$ for all $i\in [k]$,
	and the answer to the instance $I$ ``yes''.
		
	Now we turn to the ``only if'' part. 
	Then, there exist $s_1, \ldots, s_k$ with $s_i \in S_i$ 
	and $U_i=\cell_T(s_i,\{ s_1,\dots,s_k\})$ for each $i\in [k]$. 
	Take a vertex $s'_i \in \pi^{-1}(s_i)$ for each $i\in [k]$,
	$\Sigma=\{s_1,\dots,s_k\}$ and $\Sigma'=\{s'_1,\dots, s'_k\}$.	

	Consider any fixed index $i\in [k]$ and any vertex $u'\in U'_i$.
	Set $u=\pi(u') \in U_i$.
	Since $u\in U_i=\cell_T(s_i,\Sigma)$ and 
	$u\notin U_j=\cell_T(s_j,\Sigma)$ for all $j\neq i$, we have
	\[
		\forall j\in [k]\setminus \{i \}:~~~ 
				d_T(s_i,u) < d_T(s_j,u).
	\]
	Because of property~\eqref{ineq:res0} we have 
	\[
		\forall j\in [k]\setminus \{i \}:~~~ 
				d_T(s_i,u) + \frac{\res'(I)}{2} < d_T(s_j,u),
	\]
	and thus
	\[
		\forall j\in [k]\setminus \{i \}:~~~ 
				d_T(s_i,u) + 3n\cdot\delta < d_T(s_j,u).
	\]
	Using the relations~\eqref{ineq:res1} and~\eqref{ineq:res2} we get 
	\[
		\forall j\in [k]\setminus \{i \}:~~~ 
				d_{T'}(s'_i,u') < d_T(s_i,u) + 2n\delta 
								< d_T(s_j,u) - n\delta 
								< d_{T'}(s'_j,u').
	\]
	This implies that $u'\in \cell_{T'}(s'_i,\Sigma')$
	and $u'\notin \cell_{T'}(s'_j,\Sigma')$ for all $j \in [k]\setminus \{ i \}$.
	It follows that $U'_i= \cell_{T'}(s'_i,\Sigma')$.
	Since this holds for all $i\in [k]$, it follows that 
	the answer to the instance $I'$ ``yes''. 
\end{proof}

%%%%%%%%%%%%%%%%%%%%%%%%%%%%%%%%%%%%%%%%%%%%%%%%%%%%%%%%%%%%%%%%%
\subsection{Algorithm to transform}
\label{sec:algorithm}
We are now ready to explain algorithmic details of the whole transformation and
explain its efficient implementation. 

Suppose that we have an instance $I=(T,(U_1,\dots,U_k))$
for the problem \textsc{Graphic Inverse Voronoi in Trees}.
Let us use $n$ for the number of vertices in $T$ and
$N=N(I)=|V(T)|+\sum_{i} |U_i|$ for the description size
of $I$. As mentioned earlier, we can convert in $O(N)$ time this
to an equivalent instance $(T, ((U_1,S_1),\dots, (U_k,S_k)))$ for the problem
\textsc{Generalized Graphic Inverse Voronoi in Trees}.
Let $I'$ be this new instance and note that its description size is $O(N)$.

First, we root the tree $T$ at an arbitrary vertex $r$ and store
for each vertex $v$ of $T$ its parent node $\parent(v)$. (The parent of $r$ is set to \textsc{null}.)
We add to each vertex a flag to indicate whether it belongs to
a subset of vertices under consideration. Initially all flags are set to false.
This takes $O(|V(T)|)=O(N)$ time.

With this representation of $T$ we can check whether any given subset $U$ of vertices of $T$
induces a connected subgraph in $O(|U|)$ time. 
The key observation is that the subgraph $T[U]$ induced by $U$ is connected if and only if
there is exactly one vertex in $U$ whose parent does not belong to $U$.
(Here we use the convention that for the root $\parent(r)= \textsc{null} \notin U$.)
To check this condition, we set the flag of the vertices of $U$ to true,
count how many vertices $v\in U$ have the property that $\parent(v)\notin U$,
decide the connectivity of $T[U]$ depending on the counter, and at the end set
the flags of vertices of $U$ back to false.

For each vertex $v\in V(T)$ we make a list $L(v)$ that contains the
indices $i\in [k]$ with $v\in U_i$.
The lists $L(v)$, for all $v\in V(T)$, can be computed in $O(N)$ time
by scanning the sets $U_1,\dots, U_k$: for each $v\in U_i$ we add $i$ to $L(v)$.
Note that a vertex $v\in V(T)$ belongs to $W_i$ if and only if $i$
is the only index in the list $L(v)$. Thus, for any given $v\in U_i$, we can
decide in $O(1)$ time whether $v\in W_i$. With this we can compute
the sets $W_1,\dots,W_k$ in $O(\sum_i |U_i|)=O(N)$ time. 
Scanning the sets $S_1,\dots,S_k$, we can
replace each set $S_i$ with the set $S_i\cap W_i$.
Together we have spent $O(N)$ time and we have made the preprocessing
step described after Lemma~\ref{le:transform0}.

During the algorithm, as we make the edge expansions,
we maintain the lists $L(v)$ for each vertex $v$ and the rooted representation
of the tree.

Now we explain how to make the expansions of the edges in \emph{batches}: 
we iterate over the indices $i\in [k]$ and, for each fixed $i$, 
we identify $E_i$ and make all the edge expansions for $E_i$ in $O(|U_i|)$ time.
Assume for the time being that $\eps$ is already known.
We will discuss its choice below.

Consider any fixed index $i\in [k]$.
We compute $W_i$ in $O(|U_i|)$ time using the lists $L(v)$ for $v\in U_i$. 
(The set $W_i$ may have changed because of expansions for $E_j$, $j\neq i$, and thus
has to be computed again.)
We also check in $O(|U_i|)$ time that $U_i$ and $W_i$ induce connected subgraphs of $T$
using the representation of $T$.
(If any of them fails the test, then we correctly report that there is no solution.)
We construct the induced tree $T[U_i]$ explicitly and store it using adjacency lists:
for each vertex $v\in U_i$ we can find its neighbors in $T[U_i]$ in time proportional
to the number of neighbors. From this point, we will use the representation of $T[U_i]$. 

Next, we compute $E_i$, for the fixed index $i\in [k]$, in the obvious way.
For each edge $xy$ of $T[U_i]$, we check whether $x\in W_i$ and $y\notin W_i$
or whether $y\in W_i$ and $x\notin W_i$
to decide whether $xy\in E_i$.
This procedure to compute $E_i$ takes $O(|U_i|)$ time.

We keep considering the fixed index $i\in [k]$.
Now we make the expansion for each edge of $E_i$. Here it is important that
the expansion of different edges of $E_i$ are independent: each expansion 
affects to $U_i$ in a different connected component of $T-W_i$.
We make the expansion of an edge $xy\in E_i$ with $x\in W_i$ and 
$y\in U_i\setminus W_i$ as follows:
edit $T$ by inserting $y'$, set the new edge-lengths for the edges $yy'$
and $xy'$, remove from $U_i$ the subset $R_{xy}$ of elements of $U_i$ that are closer to $y$
than to $x$, and insert $y'$ in $U_i$.
The set $R_{xy}$ of elements to be removed from $U_i$ 
is obtained using the representation of $T[U_i]$ in $O(|R_{xy}|)$ time.
We correct the lists $L(v)$ by removing $i$ from $L(v)$ 
for each for each $v\in R_{xy}$.
(We do not need to update $T[U_i]$ because the sets $R_{xy}$ are pairwise disjoint
for all $xy\in E_i$.)
We conclude that expanding an edge $xy\in E_i$ takes $O(|R_{xy}|)$.
Since each element of $U_i$ can be deleted at most once from $U_i$,
and the elements $y'$ we insert cannot be deleted because they belong
only to (the new) $U_i$,
the expansions for the edges in $E_i$ takes $O(|U_i|)$ time all together.
This finishes the description of the work carried out for a fixed $i\in [k]$.

We iterate over all $i\in [k]$ making the expansions for (the current) edges in $E_i$.
Since for each $i\in [k]$ we spend $O(|U_i|)$ time,
all the expansions required for Lemma~\ref{le:transform_disjoint}
are carried out in $O(N)$ time. 
All this was assuming that the value $\eps$ is available, which remains to be discussed.
Let $\tilde I$ be the resulting instance with the disjoint sets.

Now we can make the transformation from $\tilde I$ to an instance with maximum
degree $3$. Assume for the time being that we have the parameter $\delta$ available.
Then the transformation described in Section~\ref{sec:degree3} can be easily
carried out in linear time. Thus, in $O(N)$ time we obtain
the final instance with pairwise disjoint sets $U_1,\dots, U_k$
and the tree $T$ of maximum degree $3$.

It remains to discuss how to choose the values of $\eps$ and $\delta$ for the transformations.
It is unclear whether $\eps$ or $\delta$ can be computed in $O(N)$ time when the
edges have arbitrary lengths. (If, for example, all edges have integral lengths,
then we could take $\eps=1/4$, $\delta=1/10n$ and $\delta'=1/40n^2$.)
We will handle this using composite lengths.
The length of each edge $e$ is going to be described by a triple $(a,b,c)$ that represents
the number $a+b\eps+c\delta'$ for infinitesimals $\delta'\ll\eps$. (Recall that
$4n\delta'=\delta$.)
Thus the length encoded by $(a,b,c)$ is smaller than the length encoded by $(a',b',c')$
if and only if $(a,b,c)$ is lexicographically smaller than $(a',b',c')$.
In the original graph we replace the length of each edge $e$ by $(\lambda(e),0,0)$.
In the expansion, the new edges $yy'$ get length $(0,1,0)$, and in converting
the tree to maximum degree $3$ we introduce new edges of length $(0,0,1)\equiv \delta'$ 
and we replace some edges of length $(a,b,0)$ by $(a,b,-4n)$.
The length of a path becomes a triple $(a,b,c)$ that is obtained as the vector
sum of the triples over its edges. Each comparison and addition of edge-lengths
costs $O(1)$ time.
We summarize.

\begin{theorem}
\label{thm:transformation}
	Suppose that we are given an instance $I$ for the problem \textsc{Graphic Inverse Voronoi in Trees}
	or for the problem \textsc{Generalized Graphic Inverse Voronoi in Trees} 
	of description size $N=N(I)$ over a tree $T$ with $n$ vertices.
	In $O(N)$ time we can either detect that $I$ has no solutions,
	or construct another instance $I'$ for the problem
	\textsc{Generalized Graphic Inverse Voronoi in Trees} over a tree $T'$ with the following properties:
	\begin{itemize}
	\item the tree $T'$ in the instance $I'$ has maximum degree $3$,
	\item the sets in the instance $I'$ are pairwise disjoint,
	\item the description size of $I'$ and the number of vertices in $T'$ is $O(n)$,
	\item if the answer to $I$ is ``yes'', then any solution to $I'$ can be transformed into
		a solution to $I$ using the projection map $\pi$.
	\end{itemize}
\end{theorem}
\begin{proof}
	It remains only to bound the size of $T'$ and the description size of $I'$.
	If $k>n$, then the instance $I$ has no solution and we report it.
	Otherwise, Lemma~\ref{le:number of steps} implies that we are making 
	$k-1$ expansions, which means that the resulting tree $T'$ has $n+k-1=O(n)$
	vertices. The size of the instance $I'$ is $O(n)$ because
	the sets in the instance are pairwise disjoint and there are $O(n)$ vertices
	in total.
\end{proof}

%%%%%%%%%%%%%%%%%%%%%%%%%%%%%%%%%%%%%%%%%%%%%%%%%%%%%%%%%%%%%%%%%%%%%%%%%%%%%%%%%%%%%%%%%%%%%%%%%%%%%%%%%%%%%%%
\section{Algorithm for subcubic trees with disjoint Voronoi cells}
\label{sec:algorithmDP}
In this section we consider the problem \textsc{Generalized Graphic Inverse Voronoi in Trees}
for an input $(T,\UU)$, with the following properties:
\begin{itemize}
	\item $T$ is a tree of maximum degree $3$
	\item $\UU$ is a sequence of pairs $(U_1,S_1),\dots, (U_k,S_k)$
		where the sets $U_1,\dots, U_k$ are pairwise disjoint.
\end{itemize}
Our task is to find sites $s_1,\dots,s_k$ such that, for each $i\in [k]$,
we have $U_i=\cell_T(s_i,\{ s_1,\dots, s_k\})$ and $s_i\in S_i$.
We may assume that $V(T)=\bigcup_{i\in [k]} U_i$, that $T[U_i]$ is connected for each $i\in [k]$,
and that $S_i\subseteq U_i$ for each $i\in [k]$, as otherwise it is clear that there is no solution.
These conditions can easily be checked in linear time.

First, we describe an approach to decide whether there is a solution without paying much attention to the running time.
Then, we describe its efficient implementation taking time $O(n\log^2 n)$, where $n$
is the number of vertices of $T$.

\subsection{Characterization}
\label{sec:characterization}
For each vertex $v$, let $i(v)$ be the unique index such that $v\in U_{i(v)}$.
We choose a leaf $r$ of $T$ as a root and henceforth consider
the tree $T$ rooted at $r$. We do this so that each vertex of $T$ has at most two children.
For each vertex $v$ of $T$, let $T(v)$ be the subtree of $T$ rooted at $v$,
and define also
\begin{align*}
	J(v) ~=~ \{ j\in [k]\mid U_j\cap T(v)\neq \emptyset \}.
\end{align*}
Note that $i(v)\in J(v)$. Since each $U_j$ defines a connected subset of $T(v)$,
for each $j\in J(v)$, $j\neq i(v)$, we have $U_j\subseteq T(v)$ and 
therefore it must be that $s_j\in T(v)$.

Consider a fixed vertex $v$ of $T$ and the corresponding subtree $T(v)$.
We want to parameterize possible distances from $v$ to the site $s_{i(v)}$,
that is, the site whose cell contains the vertex $v$,
that provide the desired Voronoi diagram restricted to $T(v)$.
A more careful description is below.
We distinguish possible placements of $s_{i(v)}$ within $T(v)$,
which we refer as ``below'' (or on) $v$ and for which we use the notation $B(v)$,
and possible placements outside $T(v)$,
which we refer as ``above'' and for which we use the notation $A(v)$.

First we deal with the placements where $s_{i(v)}$ is ``below'' $v$.
In this case we start defining
$X(v)$ as the set of tuples $(s_j)_{j\in J(v)}$ that satisfy the
following two conditions:
\begin{align*}
	\forall j\in J(v): &~~ s_j\in S_j\cap T(v),\\
	\forall j\in J(v):&~~ \cell_{T(v)}(s_j,\{s_t \mid t\in J(v) \})\cap T(v) =U_j\cap T(v).
\end{align*}
Note that $X(v)\subseteq~ \prod_{j\in J(v)} S_j$.
Finally, we define
\begin{align*}
	B(v) ~&=~ \left\{ d_T(s_{i(v)},v) \mid (s_j)_{j\in J(v)} \in X(v) \right\}.
\end{align*}
The set $B(v)$ represents the valid distances at which we can place $s_{i(v)}$ inside $T(v)$ such
that $s_{i(v)}$ is the closest site to $v$, and still complete the rest of the placements of the sites
to get the correct portion of $\UU$ inside $T(v)$.

Now we deal with the placements ``above'' $v$.
For $\alpha>0$, let $T^+_\alpha(v)$ be the tree obtained from $T(v)$ by adding an edge
$vv_{\rm new}$, where $v_{\rm new}$ is a new vertex, and setting
the length of $vv_{\rm new}$ to $\alpha$. The role of $v_{\rm new}$ is the placement
of the site closest to $v$, when it is outside $T(v)$.
See Figure~\ref{fig:algorithm1} for an illustration.
In the following discussion we also use Voronoi diagrams with respect to $T^+_\alpha(v)$.
Let $Y_\alpha(v)$ be the set of tuples $(s_j)_{j\in J(v)}$ that satisfy all of the
following conditions:
\begin{align*}
										&~~  s_{i(v)}=  v_{\rm new},\\
	\forall j\in J(v)\setminus \{i(v) \}: &~~ s_j\in S_j,\\		
	\forall j\in J(v):&~~ \cell_{T^+_\alpha(v)}(s_j, \{s_t \mid t\in J(v) \})\cap T(v)=U_j \cap T(v).
\end{align*}
Finally we define
\begin{align*}
	A(v) ~&=~ \left\{ \alpha \in \RR_{>0}\mid Y_\alpha(v) \neq \emptyset \right\}.
\end{align*}

\begin{figure}
\centering
	\includegraphics[page=1,scale=.88]{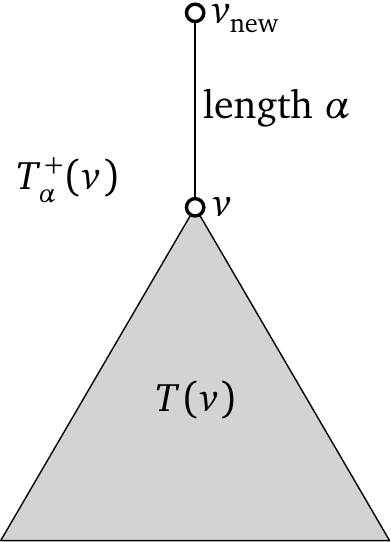}
	\caption{The tree $T^+_\alpha(v)$ used to define $A(v)$.}
	\label{fig:algorithm1}
\end{figure}

We are interested in deciding whether $B(r)$ is nonempty.
Indeed, for the root $r$ we have $J(r)=[k]$ and $T(r)=T$ by construction.
The definition of $X(v)$ implies that $B(r)$ is nonempty if and only if
there is some tuple $(s_1,\dots,s_k)\in S_1\times\dots\times S_k$ such that
\[
	\forall i\in J(r)=[k]:~~
		\cell_T(s_i,\{ s_1,\dots,s_k\}) = \cell_{T(r)}(s_i,\{ s_1,\dots,s_k\})
		=U_i\cap T(r)= U_i.
\]
This is precisely the condition we have to check to solve
\textsc{Generalized Graphic Inverse Voronoi in Trees}.

\medskip

We are going to compute $A(v)$ and $B(v)$ in a bottom-up fashion along the tree $T$.
If $v$ is leaf of $T$, then $J(v)=\{ i(v)\}$ and clearly we have
\[
	A(v)=\RR_{>0}~~~\text{ and }~~~
	B(v)=\begin{cases}
		\{ 0\} &\text{if $v\in S_{i(v)}$,}\\
		\emptyset &\text{if $v\notin S_{i(v)}$.}
		\end{cases}
\]

Consider now a vertex $v$ of $T$ that has two children $v_1$ and $v_2$.
Assume that we already have $A(v_j)$ and $B(v_j)$ for $j=1,2$.
For $j=1,2$ define the sets
\begin{align*}
	A'(v_j)~&=~ \{ x-\lambda(vv_j)\mid x\in A(v_j)\},\\
	B'(v_j)~&=~ \{ x+\lambda(vv_j)\mid x\in B(v_j)\},\\
	C'(v_j)~&=~ \{ \alpha \mid \exists x\in B(v_j) \text{ such that }
					x-\lambda(vv_j)< \alpha<x+\lambda(vv_j)\}\\
			~&=~ \bigcup_{x\in B(v_j)}\Bigl( x-\lambda(vv_j),x+\lambda(vv_j)\Bigr).
\end{align*}
This is the offset we obtain when we take into account the length of the edge $vv_j$.
The set $C'(v_j)$ will be relevant for the case when $i(v)\neq i(v_j)$.
The following lemmas show how to compute $A(v)$ and $B(v)$ from its children.
Figure~\ref{fig:algorithm2} is useful to understand the different cases.

\begin{figure}
\centering
	\includegraphics[page=2,width=\textwidth]{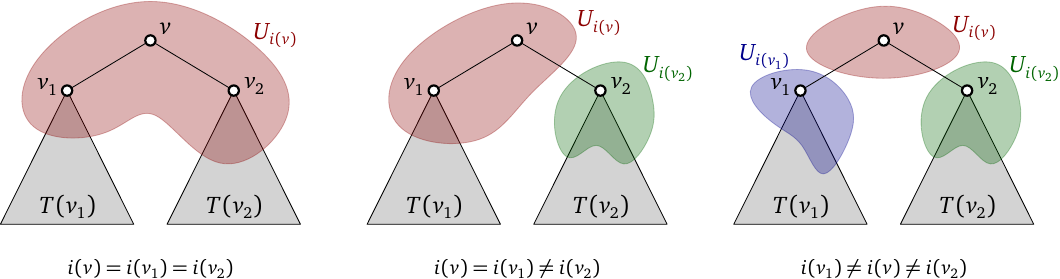}
	\caption{Different cases in the computation of $A(v)$ and $B(v)$ when $v$ has
		children $v_1$ and $v_2$. (The case $i(v)=i(v_2)\neq i(v_1)$ is symmetric
		to the case $i(v)=i(v_1)\neq i(v_2)$.)}
	\label{fig:algorithm2}
\end{figure}

\begin{lemma}
\label{le:A(v)}
	If the vertex $v$ has two children $v_1$ and $v_2$, then
	\[
		A(v)~=~ \RR_{>0} \cap
			\begin{cases}
			A'(v_1)\cap A'(v_2)&\text{if $i(v)=i(v_1)=i(v_2)$,}\\
			A'(v_1)\cap C'(v_2)&\text{if $i(v)=i(v_1)\neq i(v_2)$,}\\
			A'(v_2)\cap C'(v_1)&\text{if $i(v)=i(v_2)\neq i(v_1)$,}\\
			C'(v_1)\cap C'(v_2)&\text{if $i(v)\neq i(v_1)$ and $i(v)\neq i(v_2)$.}
			\end{cases}
	\]
\end{lemma}
\begin{proof}
	This is a standard proof in dynamic programming.
	We only point out the main insight showing the role of $A'(v_j)$ and $C'(v_j)$ for $j\in \{ 1, 2 \}$.

	When $i(v)=i(v_j)$, placing $s_{i(v)}$ at $v_{\rm new}$ of the
	tree $T^+_{\alpha}(v)$ is the same as placing it
	at $v_{\rm new}$ of $T^+_{\alpha+\lambda(v v_j)}(v_j)$.
	The valid values $\alpha$ for $T^+_{\alpha+\lambda(v v_j)}(v_j)$ are described by $A'(v_j)$,
	a shifted version of $A(v_j)$.

	When $i(v)\neq i(v_j)$,
	there has to be a compatible placement of $s_{i(v_j)}$ inside $T(v_j)$ such
	that $v$ is closer to $s_{i(v)}=v_{\rm new}$ than to $s_{i(v_j)}$, while $v_j$
	is closer to $s_{i(v_j)}$ than to $s_{i(v)}$.
	That is, we must have
	\[
		d_T(v_{\rm new},v) < d_T(s_{i(v_j)},v) ~~~\text{and}~~~
		d_T(s_{i(v_j)},v_j) < d_T(v_{\rm new},v_j),
	\]
	or equivalently, $\alpha$ must satisfy
	\[
		\alpha < d_T(s_{i(v_j)},v_j)+\lambda(vv_j) ~~~\text{and}~~~
		d_T(s_{i(v_j)},v_j) < \alpha + \lambda(vv_j).
	\]
	Thus, each possible value $x$ of $d_T(s_{i(v_j)},v_j)$,
	that is, each $x\in B(v_j)$, gives the interval
	$\bigl( x-\lambda(vv_j), x+\lambda(vv_j) \bigr)$
	of possible values for $\alpha$. The union of these intervals
	over $x\in B(v_j)$ is precisely $C'(v_j)$.
\end{proof}

To construct $B(v)$ it is useful to have a function that tells whether $v$
is a valid placement for $s_{i(v)}$.
For this matter we define the following function:
\[
	\chi(v)~=~ \begin{cases}
			\{ 0 \} &\text{if $i(v)=i(v_1)=i(v_2)$, $v\in S_{i(v)}$, $0\in A'(v_1)$
							and $0\in A'(v_2)$,}\\
			\{ 0 \} &\text{if $i(v)=i(v_1)\neq i(v_2)$, $v\in S_{i(v)}$, $0\in A'(v_1)$
							and $0\in C'(v_2)$,}\\
			\{ 0 \} &\text{if $i(v)=i(v_2)\neq i(v_1)$, $v\in S_{i(v)}$, $0\in A'(v_2)$
							and $0\in C'(v_1)$,}\\
			\{ 0 \} &\text{if $i(v)\neq i(v_1)$, $i(v)\neq i(v_2)$, $v\in S_{i(v)}$, $0\in C'(v_1)$
							and $0\in C'(v_2)$,}\\
			\emptyset &\text{otherwise.}
			\end{cases}
\]
\begin{lemma}
\label{le:B(v)}
	If the vertex $v$ has two children $v_1$ and $v_2$, then
	\[
		B(v)~=~ \chi(v) \cup\begin{cases}
			(B'(v_1)\cap A'(v_2))\cup (B'(v_2)\cap A'(v_1)) &\text{if $i(v)=i(v_1)=i(v_2)$,}\\
			B'(v_1)\cap C'(v_2)&\text{if $i(v)=i(v_1)\neq i(v_2)$,}\\
			B'(v_2)\cap C'(v_1)&\text{if $i(v)=i(v_2)\neq i(v_1)$,}\\
			\emptyset &\text{if $i(v)\neq i(v_1)$ and $i(v)\neq i(v_2)$.}
			\end{cases}
	\]
\end{lemma}
\begin{proof}
	First we note that $\chi(v)=\{ 0\}$ if and only if $v$ is a valid placement for $s_{i(v)}$.
	Indeed, the formula is the same that was used for $A(v)$, but for the value $\alpha=0$,
	and it takes into account whether $v\in S_{i(v)}$.
	
	The proof for the correctness of $B(v)$ is again based on standard dynamic programming.
	The case for $s_{i(v)}$ being placed at $v$ is covered by $\chi(v)$.
	The main insight for the case when $s_{i(v)}$ is placed in $T(v_1)$ is that,
	from the perspective of the other child, $v_2$, the vertex is placed ``above'' $v_2$.
	That is, only the distance from $s_{i(v)}$ to $v_2$ is relevant. Thus, we have
	to combine $B(v_1)$ and $A(v_2)$, with the appropriate shifts. More precisely,
	for $v_2$ we have to use $A'(v_2)$ or $C'(v_2)$ depending on whether $i(v_2)=i(v)$
	or $i(v_2)\neq i(v)$.
\end{proof}

When $v$ has a unique child $v_1$, then the formulas are simpler and the argumentation
is similar.
We state them for the sake of completeness without discussing their proof.
\begin{align*}
	A(v)~&=~ \RR_{>0} \cap
		\begin{cases}
		A'(v_1)		&\text{if $i(v)=i(v_1)$,}\\
		C'(v_1)  	&\text{if $i(v)\neq i(v_1)$.}
		\end{cases}\\ \\
	B(v)~&=~ \begin{cases}
			B'(v_1)\cup \{0 \}	&\text{if $i(v)=i(v_1)$, $v\in S_{i(v)}$, and $\lambda(vv_1)\in A(v_1)$,}\\
			B'(v_1)	&\text{if $i(v)=i(v_1)$ and
							$\bigl(v\notin S_{i(v)} \text{ or }\lambda(vv_1)\notin A(v_1)\bigr)$,}\\
			\{ 0 \} &\text{if $i(v)\neq i(v_1)$, $v\in S_{i(v)}$ and $0\in C'(v_1)$,}\\
			\emptyset &\text{if $i(v)\neq i(v_1)$ and ($v\notin S_{i(v)}$ or $0\notin C'(v_1)$).}
			\end{cases}
\end{align*}

\subsection{Efficient manipulation of monotonic intervals}
The efficient algorithm that we will present is based on an efficient representation
of the sets $A(v)$ and $B(v)$ using binary search trees.
Here we discuss the representation that we will be using.

We first consider how to store a set $X$ of real values under the following
operations.
\begin{itemize}
	\item \TREEcopy makes a copy of the data structure storing $X$;
	\item \TREEreport returns the elements of $X$ sorted;
	\item \TREEinsert$(y)$ adds a new element $y$ in $X$;
	\item \TREEdelete$(y)$ removes the element $y\in X$ from $X$;
	\item \TREEsucc$(y)$ returns the successor of $y$ in $X$, 
		defined as the smallest number in $X$ that is at least 
		as large as $y$;
	\item \TREEpred$(y)$ returns the predecessor of $y$ in $X$,
		defined as the largest number in $X$ that is smaller or equal than $y$;
	\item \TREEsplit$(y)$ returns the representation for $X_\le=\{ x\in X\mid x\le y\}$
		and the representation for $X_>=\{ x\in X\mid x> y\}$;
		the representation of $X$ is destroyed in the process;
	\item \TREEjoin$(X_1,X_2)$ returns the representation of $X=X_1\cup X_2$ if
		$\max(X_1)< \min (X_2)$, and otherwise it returns an error;
		the representations of $X_1$ and $X_2$ are destroyed in the process;
	\item \TREEshift$(\alpha)$ adds the given value $\alpha$ to all the elements of $X$.
\end{itemize}

These operations can be done efficiently using dynamic balanced binary search tree with
so-called augmentation, that is, with some extra information attached to the nodes.
Strictly speaking the following result is not needed, but understanding
it will be useful to understand the more involved data structure we eventually employ.

\begin{theorem}
\label{thm:tree}
	There is an augmented dynamic binary search tree to 
	store sets of $m$ real values with the following time guarantees:
	\begin{itemize}
		\item the operations \TREEcopy and \TREEreport take $O(m)$ time;
		\item the operations \TREEinsert, \TREEdelete, \TREEsucc, \TREEpred, \TREEsplit, 
			\TREEjoin\ and \TREEshift\ take $O(\log m)$ time.
			(For \TREEjoin\ the value $m$ is the size of the resulting set.)
	\end{itemize}
\end{theorem}
\begin{proof}
	Let $X$ be the set of values to store.
	We use a dynamic balanced binary search tree $\mathcal{T}$ 
	where each node represents one element of $X$. 
	For each node $\mu$ of $\mathcal{T}$, let $x(\mu)$ be the value represented by $\mu$.
	The tree $\mathcal{T}$ is a binary search tree with respect to the values $x(\mu)$.
	However, we \emph{do not} store $x(\mu)$ explicitly at $\mu$, but we store it
	in so-called difference form.
	At each non-root node $\mu$ with parent $\mu'$, we store
	$\TREEdiff(\mu):=x(\mu)-x(\mu')$.
	At the root $r$ we store $\TREEdiff(r)=x(r)$.
	(This choice is consistent with using $x(\textsc{null})=0$.)
	This is a standard technique already used by Tarjan~\cite{Tarjan97}.
	Whenever we want to obtain $x(\mu)$ for a node $\mu$, we have
	to add $\TREEdiff(\mu')$ for the nodes $\mu'$ along the
	root-to-$\mu$ path.
	Since operations on a tree are performed always locally, that is,
	accessing a node from a neighbour, we spend $O(\log m)$ time to compute
	the first value $x(\mu)$, and from there on each value $x(\cdot)$
	is computed in $O(1)$ additional time from the value of its neighbor.
	Of course, the values $\TREEdiff(\mu)$ have to be updated
	through the changes in the tree, 
	including rotations or other balancing operations.

	With this representation it is trivial to perform the operation \TREEshift$(\alpha)$
	in constant time: at the root $r$ of $\mathcal{T}$, 
	we just add $\alpha$ to $\TREEdiff(r)$.
	
	For the rest of operations, the time needed to execute them is the same
	as for the dynamic balanced search trees we employ.
	Brass~\cite[Chapter 3]{Brass08} explains dynamic trees with the requested
	properties;
	see Section 3.11 of the book for the more complex operations of split and join.
	(The same time bounds with amortized time bounds, which are sufficient
	for our application, can be obtained using the classical 
	splay trees~\cite{SleatorT85}.)
\end{proof}

Consider now a family $\II$ of intervals on the real line.
The family $\II$ is \DEF{monotonic} if no interval
contains another interval. 
One can define a linear order on a family of monotonic
intervals, as follows.
First note that for any two intervals $I$ and $J$
in the monotonic family, the sets $I\setminus J$ and $J\setminus I$ are both nonempty.
We say that $I$ is to the left of $J$ if $I\setminus J$ is to the left of $J\setminus I$.
If the intervals are closed, the order is the same as that obtained by sorting the 
left endpoints or their right endpoints.

In our forthcoming discussion, we will restrict our attention to intervals that are 
closed. However, the results hold when we have any type of intervals (closed, open, or semi-open).
One approach for this is to simulate an open interval $(a,b)$ with the closed interval $[a+\eps,b-\eps]$,
an interval $(a,b]$ with the closed interval $[a+\eps,b]$, and an interval $[a,b)$ with
the closed interval $[a,b-\eps]$, where $\eps>0$ is symbolic infinitessimal. This means that we store
with each interval information about the containment of its endpoints, and comparisons
are made taking into account the symbolic $\eps$. In practice, one would make comparisons
depending on the inclusion of endpoints in the intervals.
		
We want to maintain a set $\II$ of monotonic intervals under the following
operations. 
\begin{itemize}
	\item \INTcopy makes a copy of the data structure storing $\II$.
	\item \INTreport returns the elements of $\II$ sorted by their left endpoint.
	\item \INTinsert$(J)$ adds a new interval $J$ in $\II$; it assumes that
		the resulting family keeps being monotonic.
	\item \INTdelete$(J)$ deletes the interval $J\in \II$.
	\item \INThitby$(J)$, for an interval $J$, returns whether $J$ 
		intersects some interval of $\II$.
	\item \INTcontaining$(y)$, for a real value $y$, returns
		the representation for $\II'=\{ I\in \II\mid y\in I\}$ 
		and the representation for $\II''= \II\setminus \II'$.
		The representation of $\II$ is destroyed in the process.\footnote{In our application, 
		the operation \INTcontaining$(y)$ is used only internally, but it seems useful in general
		and we keep it in this description of external operations.}
	\item \INTclip$(J)$, for an interval $J=[x,y]$, 
		returns the representation for the intervals 
		$\II'=\{ I\cap J\mid I \in \II\}$
		and for the intervals 
		$\II''=\{ I\cap (-\infty,x] \mid I \in \II\}\cup
		\{ I\cap [y,+\infty) \mid I \in \II\}$.
		In both cases we remove empty intervals,
		and remove intervals contained in another one, so that we keep having
		monotonic families.
		The representation of $\II$ is destroyed in the process.
	\item \INTjoin$(\II_1,\II_2)$ 
		returns the representation of $\II=\II_1\cup \II_2$ if
		$\II$ is a monotonic family and 
		all the intervals of $\II_1$ are to the left of all the intervals
		of $\II_2$. Otherwise it returns an error.
		The representations of $\II_1$ and $\II_2$ are destroyed in the process.
	\item \INTshift$(\alpha)$, for a given real value $\alpha$,
		shifts all the intervals by $\alpha$; 
		this is, each interval $[a,b]$ in $\II$ is replaced by
		$[a+\alpha, b+\alpha]$.
	\item \INTextend$(\lambda)$, for a given real value $\lambda>0$,
		extends all the intervals by $\lambda$ in both directions; 
		this is, each interval $[a,b]$ in $\II$ is replaced by
		$[a-\lambda, b+\lambda]$.
\end{itemize}

\begin{theorem}
\label{thm:monotonic}
	There is a data structure to store monotonic families
	of $m$ intervals with the following time guarantees:
	\begin{itemize}
		\item the operations \INTcopy and \INTreport take $O(m)$ time;
		\item the operations \INTinsert, \INTdelete, \INThitby, \INTcontaining,
			\INTclip, \INTjoin, \INTshift, \INTextend\ take $O(\log m)$ time.
			(For \INTjoin\ the value $m$ is the size of the resulting set $\II$.)
	\end{itemize}
\end{theorem}
\begin{proof}
	Let $\II$ be the family of monotonic intervals to store.
	We use a dynamic balanced binary search tree $\mathcal{T}$ 
	where each node represents one element of $\II$. 
	For the node $\mu$ of $\mathcal{T}$ that represents the interval $I$, 
	let $a(\mu)$, $b(\mu)$ and $\ell(\mu)=b(\mu)-a(\mu)$ be the left endpoint,
	the right endpoint, and the length of $I$, respectively.
	Thus, if $\mu$ represents $[a_i,b_i]$, 
	we have $a_i=a(\mu)$ and $b_i= a(\mu)+\ell(\mu)$.

	The tree $\mathcal{T}$ is a binary search tree with respect to the values $a(\mu)$.
	Because the family of intervals is monotonic, $\mathcal{T}$ is \emph{also} 
	a binary search tree with respect to the values $b(\mu)$.
	However, the values $a(\mu)$, $b(\mu)$ or $\ell(\mu)$ are \emph{not stored explicitly}.
	Instead, the values are stored in difference form and implicitly.
	More precisely, at each node $\mu$ of $\mathcal{T}$ we store two values, 
	$\TREEdiff(\mu)$ and $\INTlength(\mu)$, defined as follows.
	If $\mu$ is the root of the tree and represents the interval $[a,b]$, 
	then $\TREEdiff(\mu)=a$ and $\INTlength(\mu)=b-a$.
	If $\mu$ is a non-root node of the tree and $\mu'$ is its parent, then
	$\TREEdiff(\mu)=a(\mu)-a(\mu')$ and 
	$\INTlength(\mu)=\ell(\mu)-\ell(\mu')=(b(\mu)-a(\mu))-(b(\mu')-a(\mu'))$.
	
	This is an extension of the technique employed
	in the proof of Theorem~\ref{thm:tree}. In fact, $\mathcal{T}$ is just
	the tree in the proof of Theorem~\ref{thm:tree} for the left endpoints
	of the intervals, where additionally each node stores information about
	the length of the interval, albeit this additional information is stored 
	also in difference form.
 	
	Whenever we want to obtain $a(\mu)$ or $\ell(\mu)$ for a node $\mu$, 
	we have to add $\TREEdiff(\mu')$ or $\INTlength(\mu')$
	for the nodes $\mu'$ along the root-to-$\mu$ path, respectively.
	The right endpoint $b(\mu)$ is obtained from $b(\mu)=a(\mu)+\ell(\mu)$.
	Since operations in a tree always go from a node
	to a neighbor, we can assume that the values $a(\mu)$, $b(\mu)$ and $\ell(\mu)$
	are available at a cost of $O(1)$ time per node, after an initial
	cost of $O(\log m)$ time to compute the values at the first node.
	Of course, the values $\TREEdiff(\mu)$ and $\INTlength(\mu)$
	have to be updated through the changes in the tree, 
	including rotations or other balancing operations.

	Since $\mathcal{T}$ is a binary search tree with respect to the values $a(\cdot)$
	and also with respect to the values $b(\cdot)$, 
	we can make the usual operations that can be performed in a binary
	search tree, such as
	predecessor or successor, with respect to any of those two keys. 
	For example, we can get in $O(\log m)$ time 
	the rightmost interval that contains a given value $y$, which amounts
	to a predecessor query with $y$ for the values $a(\cdot)$, or 
	we can get the leftmost interval that contains a given value $y$,
	which amounts to a successor query with $y$ for the values $b(\cdot)$.
	
	With this representation, it is trivial to perform the operations 
	\INTshift$(\alpha)$ or \INTextend$(\lambda)$ in $O(1)$ time.
	We just update \TREEdiff\ or \INTlength\ at the root.

	The operations
	\INTcopy, \INTreport, \INTinsert\ and \INTdelete\ can be carried
	out as normal operations in a dynamic binary search tree.
	The operation \INTjoin\ is also just the join operation for trees.
		
	For the operation \INThitby$(J)$\ with $J=[x,y]$
	we make a predecessor and a successor query with $x$ for the values $a(\cdot)$. 
	This gives the two intervals $I_1, I_2 \in \II$ such that
	$x$ is between the left endpoint of $I_1$ and $I_2$.
	We then check whether $I_1\cup I_2$ intersect $J$, which requires
	constant time.

	For the operation \INTcontaining$(y)$ we proceed as follows.
	We find the rightmost interval $[a_\ell,b_\ell]\in \II$ 
	with $b_\ell < y$. If there is no interval in $\II$ completely to the left of $y$, 
	we set $a_\ell=-\infty$.
	Because $\II$ is a monotonic family of intervals,
	the intervals that contain $y$ are precisely those with the 
	left endpoint in the interval $(a_\ell,y]$. We use
	the operations \TREEsplit$(a_\ell)$ and \TREEsplit$(y)$ 
	with respect to the values $a(\cdot)$ to obtain
	the representations of  
	\begin{align*}
		\II_1~&=~\{ [a,b]\in \II\mid a\le a_\ell \},\\
		\II_2~&=~\{ [a,b]\in \II\mid a_\ell < a \le y \} ~=~
			\{ I\in \II\mid y\in I\},\\
		\II_3~&=~\{ [a,b]\in \II\mid y <a \}.
	\end{align*}
	We then use the \TREEjoin\ operation to merge the representations of 
	$\II_1$ and $\II_3$.
	
	For the operation \INTclip$(J)$\ with the interval $J=[x,y]$ we proceed as follows.
	We use $\INTcontaining(x)$, $\INTcontaining(y)$, and $\INTjoin$
	to separate $\II$ into the group $\II'$ of intervals pierced by $x$ or $y$,
	and the rest, $\II''$.
	Then we use $\TREEsplit(x)$ (with respect to $a(\cdot)$) 
	and $\TREEsplit(y)$ (with respect to $a(\cdot)$)
	to split $\II''$ into three groups:
	$\II_1$ containing intervals of $\II$ contained in $(-\infty, x]$, 
	$\II_2$ containing intervals of $\II$ contained in $[x,y]$, 
	and $\II_3$ containing intervals of $\II$ contained in $[y,+\infty)$.
	In $\II'$ we find the leftmost interval that contains $x$,
	clip it with $(-\infty, x]$, and add it to $\II_1$.
	Again in the same group, $\II'$, we find the rightmost interval that
	contains $x$, clip it with $[x,y]$, and add it to $\II_2$.
	We do a similar procedure for $y$: add to $\II_2$ 
	the leftmost interval of $\II'$ that contains $y$, clipped with $J$,
	and add to $\II_3$
	the rightmost interval of $\II'$ that contains $y$, clipped with $[y,+\infty)$.
	If the two intervals we added to $\II_2$ are the same, which means that
	they both are $[x,y]$, we only add one of them.
	The procedure takes $O(\log m)$ time.
\end{proof}

Consider a set $A\subseteq \RR$.
A \DEF{representation} of $A$ is a family $\II$ of monotonic intervals such that
$A=\bigcup_{I\in \II} I$. 
The intervals in $\II$ \emph{may intersect} and the representation
is not uniquely defined.
See Figure~\ref{fig:algorithm3} for an example. 
The \emph{size} of the representation $\II$ is the number of (possibly non-disjoint) intervals in $\II$.
This is potentially larger than the minimum number of intervals that is needed because
the intervals in $\II$ can intersect. 

\begin{figure}
\centering
	\includegraphics[page=3]{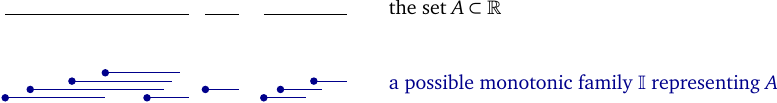}
	\caption{The set $A$ at the top and one possible representation $\II$ of $A$.
		The size of this representation is $9$.}
	\label{fig:algorithm3}
\end{figure}

Consider some set $A$ and its representation $\II$.
If we use the data structure of Theorem~\ref{thm:monotonic} to store $\II$,
the operations reflect operations we do with $A$.
For example, \INThitby$(J)$ tells whether $J$ intersects $A$,
while $\INTclip([x,y])$
returns a representation of $A\cap [x,y]$ and a representation of 
$A\cap (-\infty,x]\cup A\cap [y,+\infty)$.
When $\II_1$ and $\II_2$ are representations of $A_1$ and $A_2$,
then \INTjoin$(\II_1,\II_2)$ returns the representation of $A_1\cup A_2$,
assuming that $\max(A_1) < \min(A_2)$.

\subsection{Algorithm}
In this section we present an efficient algorithm based on the characterization of the previous section.
We keep using the same notation. In particular, $T$ keeps being a rooted tree
and each vertex has at most two children.
We use $n$ for the number of vertices of $T$.

There are two main ideas used in our approach.
The first one is that, for each vertex of the tree with two children,
we want to spend time (roughly) proportional to the size of the smaller subtree of its children.
The second idea is to use representations of $A(v)$ and $B(v)$ and manipulate them 
using the data structure of Theorem~\ref{thm:monotonic}.

The following lemma, which is folklore, shows the advantage of the first idea.
For each node $v$ with two children, let $v_1$ and $v_2$ be its two children.
If $v$ has only one child, we denote it by $v_1$.
For each node $v$, let $n(v)$ be the number of vertices in the subtree $T(v)$.
(Thus $n(r)=n$.)

\begin{lemma}
\label{le:sumnodes}
	If $V_2$ denotes the vertices of $T$ with two children, then
	\[
		\sum_{v\in V_2} \min \{ n(v_1), n(v_2)\} = O(n\log n).
	\]
\end{lemma}
\begin{proof}
	For each vertex $u$ of $T$ define
	\[
		\sigma(u)~=~\sum_{v\in V_2 \cap V(T(u))} \min \{ n(v_1), n(v_2)\}.
	\]
	Thus, we want to bound $\sigma(r)$.
	We show by induction on $n(u)$ that
	\[
		\sigma(u)~\le~ n(u) \log_2 n(u).
	\]
	For the base case note that, when $n(u)= 1$, the vertex $u$ is a leaf and $\sigma(u)=0$,
	so the statement holds.
	
	If $u$ has one child $u_1$, then we have $V_2\cap T(u)= V_2\cap T(u_1)$,
	\[
		\sigma(u)~=~\sigma(u_1)~\le ~ n(u_1) \log_2 n(u_1) ~\le~ n(u) \log_2 n(u),
	\]
	and the bound holds.
	If $u$ has two children $u_1$ and $u_2$, then we can assume without loss of generality
	that $n(u_1)\le n(u_2)$, which implies that $n(u_1)< n(u)/2$.
	Using the induction hypothesis for $n(u_1)$ and $n(u_2)$,
	we obtain
	\begin{align*}
		\sigma(u) ~&=~ \sum_{v\in V_2 \cap V(T(u))} \min \{ n(v_1), n(v_2)\} \\
					&=~	\sigma(u_1)+ \sigma(u_2) + n(u_1) \\
					&\le~ n(u_1)\log_2 n(u_1)+ n(u_2)\log_2 n(u_2) + n(u_1) \\
					&<~ n(u_1)\log_2 \left(n(u)/2 \right)+ n(u_2)\log_2 n(u) + n(u_1) \\
					&=~ n(u_1)\left(\log_2 n(u) -1\right)+ n(u_2)\log_2 n(u) + n(u_1)\\
					&=~ \bigl( n(u_1)+ n(u_2)\bigr) \log_2 n(u) \\
					&<~ n(u) \log_2 n(u).   \qedhere
	\end{align*}
\end{proof}

We manipulate the sets $A(v)$ and $B(v)$ using representations
$\II(A(v))$ and $\II(B(v))$, respectively.
Note that $B(v)$ is just a finite set of values,
but we store it as a family of zero-length monotonic intervals. 
The reason for this artificial approach to treat $B(v)$, as opposed to using
a set of values, is that in our algorithm sometimes 
we set the lengths of intervals defined by $B(v)$.
Thus, there is no real difference between how we treat the representations of $A(\cdot)$
and $B(\cdot)$. 

The families of intervals $\II(A(v))$ and $\II(B(v))$
are stored and manipulated using the data structure of Theorem~\ref{thm:monotonic}.
Thus, we are using the data structure described in Theorem~\ref{thm:monotonic}
to represent $A(v)$ and $B(v)$ implicitly, as the union of monotonic intervals.
The reason for this choice is technical and reflected in the proof of the 
next lemma. 

For each vertex $v$ of $T$,
we use $m_A(v)$ and $m_B(v)$ to denote the sizes of $\II(A(v))$ and $\II(B(v))$,
respectively. Although the value $m_A(v)$ actually depends on the family $\II(A(v))$ of 
intervals that is used, this relaxation of the notation will not lead to confusion.

It is clear that $B(v)$ has at most $n(v)$ values because each value corresponds to a vertex of $T(v)$. Thus, $m_B(v)\le n(v)$.
A similar bound will hold for $m_A(v)$ by induction.

\begin{lemma}
\label{le:degree2}
	Consider a vertex $v$ of $T$ with two children $v_1$ and $v_2$,
	and assume that we have representations 
	$\II(A(v_1))$, $\II(B(v_1))$, $\II(A(v_2))$
	and $\II(B(v_2))$ of 
	$A(v_1)$, $B(v_1)$, $A(v_2)$ and $B(v_2)$, respectively,
	each of them stored in the data structure of Theorem~\ref{thm:monotonic}.
	Set $m_1=m_A(v_1)+ m_B(v_1)$ and $m_2=m_A(v_2)+ m_B(v_2)$,
	and assume that $m_1\le m_2$.
	We can compute in $O(m_1 \log m_2)$ time families $\II(A(v))$ and
	$\II(B(v))$ that represent $A(v)$ and $B(v)$, respectively,
	each of them stored in the data structure of Theorem~\ref{thm:monotonic}.\footnote{In the
	process we destroy the data structures for $\II(A(v_2))$ and $\II(B(v_2))$.}
	Moreover, the representation $\II(A(v))$ has size at most
	\[
		\max\{ m_A(v_1)+m_A(v_2), \, m_A(v_1)+m_B(v_2), \, m_B(v_1)+m_A(v_2),\, m_B(v_1)+m_B(v_2)\}.
	\]
\end{lemma}

\begin{proof}
	First we compute $\chi(v)$.
	To check whether $0\in A'(v_j)$, where $j\in \{1,2 \}$, we perform the operation 
	$\INThitby([\lambda(vv_j),\lambda(vv_j)])$ in the representation $\II(A(v_j))$.
	To check whether $0\in C'(v_j)$, where $j\in \{1,2 \}$,
	we observe that $0\in C'(v_j)$ if and only if $[-\lambda(vv_j),+\lambda(vv_j)]$
	contains some element of $B(v_j)$. This latter question
	is answered making the query 
	$\INThitby([-\lambda(vv_j),+\lambda(vv_j)])$ in the representation $\II(B(v_j))$. 
	We conclude, that $\chi(v)$ can be computed in $O(\log m_2)$ time
	without changing any of the representations.
	
	Next, for each $j\in \{1,2 \}$,
	we compute the representation 
	$\II(A'(v_j))$ of $A'(v_j)$ applying the operation $\INTshift(-\lambda(vv_j))$ to $\II(A(v_j))$.
	Similarly, we can compute the representation $\II(B'(v_j))$ of $B'(v_j)$.
	This takes $O(\log m_1)+ O(\log m_2)= O(\log m_2)$ time.
	More importantly, with an additional cost of $O(\log m_j)$ time 
	we can use indistinctly the representation of $B(v_j)$ or $B'(v_j)$, whatever is more convenient.

	Note that we cannot afford to make copies of the representations 
	$\II(A'(v_2))$ or $\II(B'(v_2))$ because this would take $\Theta(m_2)$ time, which may be too much.
	On the other hand, we can manipulate and make explicit copies of
	$\II(A'(v_1))$ and $\II(B'(v_1))$ because it takes $O(m_1)$ time.
	Define the \DEF{minimal representation} of a set $A\subset \RR$ to be the maximal intervals 
	(with respect to inclusion) in $A$. From $\II(A'(v_1))$ we
	can compute the minimal representation of $A'(v_1)$ in linear time, that is, $O(m_1)$ time.
	For this we use the operation \INTreport in $\II(A'(v_1))$, which returns
	the intervals in $\II(A'(v_1))$ sorted by their left endpoints,
	and sequentially merge adjacent intervals that intersect.
	Similarly, we can find a minimal representation of $B'(v_1)$, which is a list 
	of the values in $B'(v_1)$.
	Thus, after $O(m_1)$ time we have the minimal representation of $A'(v_1)$ as
	a list of (sorted) at intervals $J_1,\dots, J_s$
	and $B'(v_1)$ as a sorted list of values $y_1,\dots,y_t$, where $s+t\le m_1$.
	
	Now we distinguish cases depending on the relations between $i(v)$, $i(v_1)$ and $i(v_2)$.
	
	Consider the case when $i(v)=i(v_1)=i(v_2)$.
	We have two parts.
	\begin{enumerate}
	\item First we compute the representation $\II(B(v))$ of $B(v)$. 
		Because of Lemma~\ref{le:B(v)}, we have 
		\[
			B(v)=\chi(v)\cup (B'(v_1)\cap A'(v_2))\cup (B'(v_2)\cap A'(v_1)).
		\]
		Recall that we have an explicit representation of $B'(v_1)$.
		For each element $y$ in $B'(v_1)$, we query $\II(A'(v_2))$ 
		using $\INThitby([y,y])$ to decide whether $y\in A'(v_2)$.
		Thus, we can compute an explicit representation of 
		$B'(v_1)\cap A'(v_2)$ in $O(m_1\log m_2)$ time.
	
		Recall that we also have an explicit minimal representation $J_1,\dots, J_s$ of $A'(v_1)$.
		For each interval $J$ in that representation,
		we query $\II(B(v_2))$ with $\INTclip(J)$
		to obtain the representation of $J\cap B'(v_2)$.
		Since the sets $J_1,\dots, J_s$ are pairwise disjoint, we indeed obtain representations
		of the sets	$J_1\cap B'(v_2),\dots, J_s\cap B'(v_2)$. 
		We then merge them using $\INTjoin$. Since the intervals $J_1,\dots, J_s$ are pairwise disjoint,
		the operation $\INTjoin$ can be indeed performed.
		In total we have used $s\le m_1$ times the operations $\INTclip$ and $\INTjoin$,
		and thus we spent $O(m_1\log m_2)$ time in total.
		Inserting in this representation the values (as zero-length intervals) of
		$B'(v_1)\cap A'(v_2)$, we finally obtain a representation of 
		$(B'(v_1)\cap A'(v_2))\cup (B'(v_2)\cap A'(v_1))$.
		If $\chi(v)$ is nonempty, we also insert the interval $[0,0]$ in the representation. 
		The final result is a representation $\II(B(v))$ of $B(v)$.
		Note that in this computation we have destroyed the representation of $\II(B'(v_2))$
		because of the operations $\INTclip$.
	\item Next we compute the representation $\II(A(v))$ of $A(v)$.
		Because of Lemma~\ref{le:A(v)} we have that $A(v)=\RR_{>0}\cap A'(v_1)\cap A'(v_2)$.
		Recall that we have an explicit minimal representation $J_1,\dots, J_s$ of $A'(v_1)$.
		For each interval $J_i$ in the minimal representation of $A'(v_1)$,
		we extract from $\II(A'(v_2))$ a representation of $J_i\cap A'(v_2)$
		using $\INTclip(J_i)$. 
		Then we compute a representation of
		$\bigcup_{i\in [s]} J_i\cap A'(v_2) = A'(v_1)\cap A'(v_2)$
		using $s-1$ times the operation $\INTjoin$. In both cases it is important
		that the intervals $J_1,\dots,J_s$ are pairwise disjoint.		
		This takes $O(s \log m_2)=O(m_1 \log m_2)$ time.
		To obtain $\II(A(v))$ we apply $\INTclip(\RR_{>0})$. 
		Note that in this computation we have destroyed the representation $\II(A'(v_2))$ of $A'(v_2)$,
		because of the $\INTclip$ operations,
		and therefore this step has to be made after the computation of $B(v)$,
		which is also using the representation $\II(A'(v_2))$, but not changing it.
	\end{enumerate}

	\noindent Consider now the case when $i(v)=i(v_1)\neq i(v_2)$.
	We proceed as follows.
	\begin{enumerate}
	\item First we compute the representation $\II(B(v))$ of $B(v)$.
		Because of Lemma~\ref{le:B(v)} we have $B(v)=\chi(v)\cup (B'(v_1)\cap C'(v_2))$.
		Note that, for each $y\in \RR$, we have $y\in C'(v_2)$ if and only if the interval
		$[y-\lambda(vv_2), y+\lambda(vv_2)]$ contains some element of $B(v_2)$.
		Recall that we have an explicit description $y_1,\dots,y_t$ of $B'(v_1)$.
		Therefore, for each element $y\in B'(v_1)$, we use the operation
		$\INThitby([y-\lambda(vv_2), y+\lambda(vv_2)])$ in $\II(B(v_2))$ to detect whether
		$y\in C'(v_2)$. With this we computed $B'(v_1)\cap C'(v_2)$ explicitly
		in $O(m_1\log m_2)$ time and we did not change the representation $\II(B(v_2))$.
		Finally, we build the data structure for the representation $\II(B(v))$ of $B(v)$
		by inserting the intervals $[y,y]$ with $y\in B'(v_1)\cap C'(v_2)$ and,
		if $\chi(v)$ is nonempty, we also insert $[0,0]$ in the data structure.
	\item Next we compute the representation $\II(A(v))$ of $A(v)$.
		Because of Lemma~\ref{le:A(v)} we have $A(v)=\RR_{>0}\cap A'(v_1)\cap C'(v_2)$.
		We obtain a representation $\II(C'(v_2))$ of $C'(v_2)$ 
		applying the operation $\INTextend(\lambda(vv_2))$ to the representation $\II(B(v_2))$ of $B(v_2)$.
		Recall that we have an explicit minimal representation $J_1,\dots, J_s$ of $A'(v_1)$.
		For $i=1,\dots,s$, we iteratively use 
		the operation $\INTclip(J_i)$ in the representation $\II(C'(v_2))$.
		Using \INTjoin\ over the representations reported we obtain 
		the representation of 
		\[
			\bigcup_{i\in [s]} \left(C'(v_2)\cap J_i \right)
			~=~ C'(v_2)\cap \left(\bigcup_{i\in [s]}J_i\right) 
			~=~ C'(v_2)\cap A'(v_1).
		\]
		To obtain $\II(A(v))$ we apply $\INTclip(\RR_{>0})$.
		Since we are making $O(m_1)$ operations, we spend $O(m_1 \log m_2)$ time.
		Note that in this computation we have destroyed the representation of $\II(B'(v_2))$
		(or the equivalent representation $\II(B(v_2))$),
		and thus this step has to be made after the computation of $\II(B(v))$, 
		which is also using $\II(B'(v_2))$.
	\end{enumerate}

	\noindent Consider now the case when $i(v)=i(v_2)\neq i(v_1)$.
	We proceed as follows.
	\begin{enumerate}
	\item First we compute the representation $\II(B(v))$ of $B(v)$.
		Because of Lemma~\ref{le:B(v)} we have $B(v)=\chi(v)\cup (B'(v_2)\cap C'(v_1))$.
		We compute explicitly the minimal representation of $C'(v_1)$.
		Then, for each interval $I$ in that representation we query for
		the elements $I\cap B'(v_2)$ using $\INTclip(I)$ in $\II(B'(v_2))$ and 
		join the answers using $\INTjoin$ over all intervals $I$. 
		This takes $O(m_1\log m_2)$ time and changes the data structure of the
		representation $\II(B'(v_2))$.
		Finally, if $\chi(v)$ is nonempty, we also insert $[0,0]$ in the result.
		The total time is $O(m_1\log m_2)$.
	\item Next we compute the representation $\II(A(v))$ of $A(v)$.
		Because of Lemma~\ref{le:A(v)} we have $A(v)=\RR_{>0}\cap A'(v_2)\cap C'(v_1)$.
		Again, we compute explicitly the minimal representation of $C'(v_1)$.
		For each interval $I$ in the minimal representation of $C'(v_1)$
		we use the operation $\INTclip(I)$ in $\II(A'(v_2))$ to obtain
		a representation of $I\cap A'(v_2)$,
		and then use $\INTjoin$ to join all the answers.
		With this we obtain a representation of $A'(v_2)\cap C'(v_1)$,
		to which we apply $\INTclip(\RR_{>0})$.
		This procedure takes $O(m_1\log m_2)$ time and changes the representation of $\II(A'(v_2))$.
	\end{enumerate}
	
	\noindent Consider now the remaining case, when $i(v)\neq i(v_1)$ and $i(v)\neq i(v_2)$.
	We proceed as follows.
	\begin{enumerate}
	\item The computation of $B(v)$ is trivial, since $B(v)=\chi(v)$ by Lemma~\ref{le:B(v)}.
	\item The computation of the representation of 
		$A(v)=C'(v_1)\cap C'(v_2)$ is similar to the case when $i(v)=i(v_1)\neq i(v_2)$.
		We compute explicitly the minimal representation of $C'(v_1)$,
		and use it as it was done there (for $\II(A'(v_1))$).
		This takes $O(m_1\log m_2)$ time.
	\end{enumerate}	

	\noindent In each case we spent $O(m_1\log m_2)$ time, and the time bound follows.
	For the upper bound on the representation $\II(A(v))$ of $A(v)$, we note that each left endpoint of
	each interval in $\II(A(v))$ 
	gives rise to at most one interval in the representation of $A(v)$. The four terms
	correspond to the four possible cases we considered for the indices $i(v)$, $i(v_1)$ and $i(v_2)$.
\end{proof}

\begin{lemma}
\label{le:disjoint}
	The problem \textsc{Generalized Graphic Inverse Voronoi in Trees} for an input $(T,\UU)$
	where $T$ is an $n$-vertex tree of maximum degree $3$ and the candidate Voronoi cells
	are pairwise disjoint, can be solved in $O(n\log^2 n)$ time.
\end{lemma}
\begin{proof}
	We root $T$ at a leaf so that each node has at most two descendants.
	For each vertex $v$ of $T$, we compute a representation $\II(A(v))$ and $\II(B(v))$
	of the sets $A(v)$ and $B(v)$, respectively.
	The computation is bottom-up: we compute $\II(A(v))$ and $\II(B(v))$ when this has
	been computed for all the children of $v$.
	If $v$ has two children, we use Lemma~\ref{le:degree2}.
	If $v$ has one child, then the computation can be done in $O(\log m_A(v)+\log m_B(v))$ time in 
	a straightforward manner.
	When we arrive to the root $r$, we just have to check whether $B(r)$ is nonempty.
	
	We can see by induction that, for each vertex $v$ of $T$, $m_A(v)\le n(v)$.
	(We already mentioned earlier that $B(v)$ has at most $n(v)$ values, one per vertex of $T(v)$.)
	This is clear for the leaves because $A(\cdot)$ has only one interval.
	For the internal nodes $v$ that have one child $u$ it follows because the 
	representation $\II(A(v))$ of $A(v)$
	is obtained from the representation of $\II(A(u))$ by a shift.
	For the internal nodes $v$ with two children $v_1$ and $v_2$,
	the bound on $m_A(v)$ follows by induction from the bound in Lemma~\ref{le:degree2}.
	In particular, we have $O(\log m_A(v)+\log m_B(v))=O(\log n)$ at each node $v$ of $T$.
	
	For each vertex with one child we spend $O(\log n)$ time.
	For each vertex $v$ with two children $v_1$ and $v_2$
	we spend $O(\min\{n(v_1),n(v_2)\}\log n)$ time.
	Thus, if $V_1$ and $V_2$ denote the vertices with one and two children, respectively,
	we spend
	\begin{align*}
		O(n) \,+\, \sum_{v\in V_1} O(\log n) \,+\,& \sum_{v\in V_2} O(\min\{n(v_1),n(v_2)\}\log n) \\ =~
		&O(n\log n) \,+\, O(\log n) \sum_{v\in V_2} O(\min\{n(v_1),n(v_2)\})
	\end{align*}
	time. Using Lemma~\ref{le:sumnodes}, this time is $O(n\log^2 n)$.
	
	Standard (but non-trivial) adaptations can be used to recover an actual solution.
	One option is to use persistent data structures for the search trees that store 
	families $\II$ of monotonic intervals.
	A persistent data structure allows to make queries to any version of the tree in the past. 
	Thus, it stores implicitly copies of the trees that existed at any time.
	Sarnak and Tarjan~\cite{SarnakT86} explain how to make red-black tree persistent
	(and how the \TREEjoin\ and \TREEsplit\ operations can also be done).
	Since we have access to the past versions of the tree, we can recover how the solution
	was obtained. Each operation in the past takes $O(\log m)$ time, where $m$ is the sum
	of operations that were performed. In our case this is $O(\log (n\log^2 n))=O(\log n)$ time
	per operation/query in the tree, and the running time is not modified.
	Another, conceptually simpler option is to store through the algorithm information on how to undo
	each operation. Then, at the end of the algorithm, we can run the whole algorithm backwards
	and recover the solutions.	
\end{proof}

\begin{theorem}
	The problem \textsc{Generalized Graphic Inverse Voronoi in Trees}
	for instances $I=(T,((U_1,S_1),\dots, \allowbreak (U_k,S_k)))$
	can be solved in time $O(N+ n \log^2 n)$, where $T$ is a tree with $n$ vertices
	and $N=|V(T)|+\sum_i (|U_i|+|S_i|)$.
\end{theorem}

\begin{proof}
	Because of Theorem~\ref{thm:transformation}, we can transform in $O(N)$ time
	the instance $I$ to another instance $I'=(T',((U'_1,S'_1),\dots, (U'_k,S'_k)))$,
	where $T'$ has maximum degree $3$, the sets $U'_1,\dots,U'_k$ are pairwise disjoint,
	and $T'$ has $O(n)$ vertices.
	We can compute a solution to instance $I'$ in $O(n\log^2 n)$ time using Lemma~\ref{le:disjoint}.
	Then, we have to check whether this solution is actually a solution for $I$.
	For this we use Lemma~\ref{le:check}.
\end{proof}

\begin{corollary}
	The problem \textsc{Graphic Inverse Voronoi in Trees}
	for instances $I=(T, (U_1,\dots, U_k))$,
	can be solved in time $O(N+ n\log^2 n)$, where $T$ is a tree with $n$ vertices
	and $N=|V(T)|+\sum_i |U_i|$.
\end{corollary}

%%%%%%%%%%%%%%%%%%%%%%%%%%%%%%%%%%%%%%%%%%%%%%%%%%%%%%%%%%%%%%%%%%%%%%%%%%%%%%%%%%%%%%%%%%%%%%%%%%%%%%%%%%%%%%%%%%%%%%%%%%%%%%%%%%%%%%%%%%%%%%%%%%%%%%%%%%%%%%%%%%%%%%%%%%%%%%%%%%%%%%%%%%%%%%%%%%%%%%%%%%%%%%

\section{Lower bound for trees}
\label{sec:lowerbound}

We can show the following lower bound on any algorithm based on
algebraic operations on the lengths of the edges.

\begin{theorem}
\label{thm:lowerbound}
	In the algebraic computation tree model,
	solving \textsc{Graphic Inverse Voronoi in Trees} 
	with $n$ vertices takes $\Omega(n\log n)$ operations,
	even when the lengths are integers.
\end{theorem}
\begin{proof}
	Consider an instance $X,Y$ for the decision problem \textsc{Set Intersection}, where
	$X=\{x_1,\dots,x_n\}$ and $Y=\{y_1,\dots,y_n\}$ are sets of integers.
	The task is to decide whether $X\cap Y$ is nonempty.
	This problem has a lower bound of $\Omega(n \log n)$ in the algebraic computation
	tree model~\cite{Yao91}. (In particular, this implies the same lower bound for the
	bounded-degree algebraic decision tree model.)
	Adding a common value to all the numbers, we may assume that $X$ and $Y$ contain
	only positive integers.
	
	We construct an instance to the \textsc{Graphic Inverse Voronoi in Trees} problem, as follows. 
	See Figure~\ref{fig:lowerbound}.
	We construct a star $S_X$ with $n+1$
	leaves. The edges of $S_X$ have lengths $x_1,\dots,x_n,2$.
	We construct also a star $S_Y$ with $n+1$ leaves whose
	edges have lengths $y_1+1,\dots,y_n+1,1$. Finally, we identify
	the leaf of $S_X$ incident to the edge of length $2$ and
	the leaf of $S_Y$ incident to the edge of length $1$.
	Let $T$ be the resulting tree. We take the sets $U_1$ and $U_2$ to
	be the vertex sets of $S_X$ and $S_Y$, respectively.
	Note that $T$ has $2n+3$ vertices.
	The reduction makes $O(n)$ operations.
	
	\begin{figure}
	\centering
		\includegraphics[]{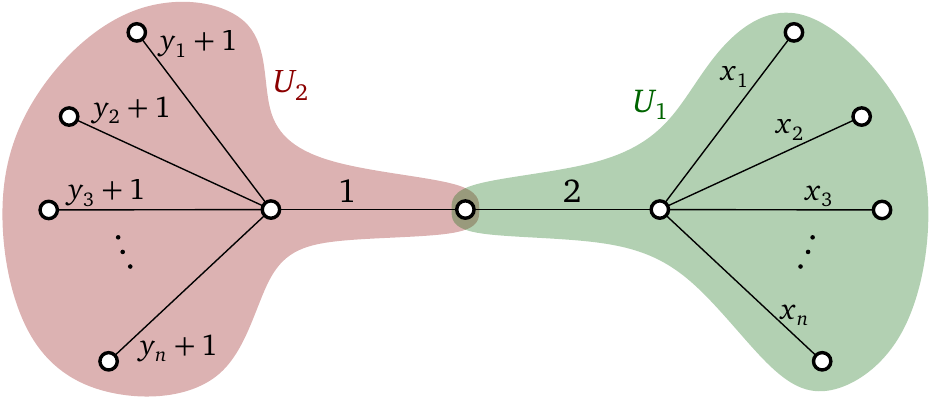}
		\caption{Construction to show the lower bound in Theorem~\ref{thm:lowerbound}.}
		\label{fig:lowerbound}
	\end{figure}	

	Since placing the sites on the centers of the stars does not produce a solution,
	it is straightforward to see that the answers to \textsc{Set Intersection}$(X,Y)$
	and to \textsc{Graphic Inverse Voronoi in Trees}$(T,(U_1,U_2))$ are the same.
	Thus, solving \textsc{Graphic Inverse Voronoi in Trees}$(T,(U_1,U_2))$ in $o(n\log n)$ time
	would provide a solution to \textsc{Set Intersection}$(X,Y)$ in $o(n\log n)$ time,
	and contradict the lower bound.	
\end{proof}

The lower bound also extends to the problem \textsc{Generalized Graphic Inverse Voronoi in Trees}
with disjoint regions because we can apply the transformation to make the cells
disjoint.
%%%%%%%%%%%%%%%%%%%%%%%%%%%%%%%%%%%%%%%%%%%%%%%%%%%%%%%%%%%%%%%%%%%%%%%%%%%%%%%%%%%%%%%%%%%%%%%%%%%%%%%%%%%%%%%%%%%%%%%%%%%%%%%%%%%%%%%%%%%%%%%%%%%%%%%%%%%%%%%%%%%%%%%%%%%%%%%%%%%%%%%%%%%%%%%%%%%%%%%%%%%%%%

\section{Conclusions}
We have introduced the inverse Voronoi problem for graphs and
we have shown several different hardness results, 
also within the framework of parameterized complexity.
We have provided an algorithm for the inverse Voronoi problem in trees
and a lower bound in a standard computation model. 

Here we list some possible directions for further research.
\begin{itemize}
	\item Is there an algorithm to solve the problem in $n^{O(w)}$ time for graphs with
		$n$ vertices and treewidth $w$ when the candidate Voronoi cells intersect?
		Perhaps one can also use some treewidth associated to the candidate Voronoi regions.
		In particular, for planar graphs a running time of $n^{O(\sqrt{k})}$ 
		seems plausible but challenging when the Voronoi cells overlap. 
	\item Considering cells defined by additively weighted sites.
	\item Following the analogy to problems considered in the Euclidean case~\cite{AloupisPPTT13,BanerjeeBDKMR13},
		find the smallest set $\Sigma$ such that each $U_i$ is the
		union of some Voronoi cells in $\VV(\Sigma)$. Taking $\Sigma=V(G)$ gives a feasible solution,
		and our hardness implies that the problem is NP-hard.
		Can one get approximation algorithms?
	\item Since for trees the upper bound of our algorithm and the lower bound differ,
		a main open question is closing this gap.
		Considering trees with unit edge lengths may also be interesting.
		Our lower bound does not apply for such instances.
\end{itemize}

%%%%%%%%%%%%%%%%%%%%%%%%%%%%%%%%%%%%%%%%%%%%%%%%%%%%%%%%%%%%%%%%%%%%%%%%%%%%%%%%%%%%%%%%%%%%%%%%%%%%%%%%%%%%%%%%%%%%%%%%%%%%%%%%%%%%%%%%%%%%%%%%%%%%%%%%%%%%%%%%%%%%%%%%%%%%%%%%%%%%%%%%%%%%%%%%%%%%%%%%%%%%%%
\section*{Acknowledgments}
We are very grateful to anonymous reviewers that pointed out an error in the previous version of Section~\ref{sec:degree3} and several other useful corrections.

Part of this work was done at the 21st Korean Workshop on Computational Geometry, held
in Rogla, Slovenia, in June 2018. We thank all workshop participants for their helpful
comments.


\begin{thebibliography}{10}

\bibitem{AloupisPPTT13}
G.~Aloupis, H.~P{\'{e}}rez{-}Ros{\'{e}}s, G.~Pineda{-}Villavicencio,
  P.~Taslakian, and D.~Trinchet{-}Almaguer.
\newblock Fitting {V}oronoi diagrams to planar tesselations.
\newblock {\em Combinatorial Algorithms - 24th International Workshop,
  {IWOCA}}, pp.~349--361. Springer, Lecture Notes in Computer Science 8288,
  2013, \url{https://doi.org/10.1007/978-3-642-45278-9_30}.

\bibitem{Ash85}
P.~F. Ash and E.~D. Bolker.
\newblock Recognizing {D}irichlet tessellations.
\newblock {\em Geometriae Dedicata} 19(2):175--206, Nov 1985,
  \url{http://dx.doi.org/10.1007/BF00181470}.

\bibitem{BandyapadhyayBDS15}
S.~Bandyapadhyay, A.~Banik, S.~Das, and H.~Sarkar.
\newblock {V}oronoi game on graphs.
\newblock {\em Theor. Comput. Sci.} 562:270--282, 2015,
  \url{https://doi.org/10.1016/j.tcs.2014.10.003}.

\bibitem{BanerjeeBDKMR13}
S.~Banerjee, B.~B. Bhattacharya, S.~Das, A.~Karmakar, A.~Maheshwari, and
  S.~Roy.
\newblock On the construction of generalized {V}oronoi inverse of a rectangular
  tessellation.
\newblock {\em Trans. Computational Science}, pp.~22--38. Springer, Lecture
  Notes in Computer Science 8110, 2013,
  \url{https://doi.org/10.1007/978-3-642-41905-8_3}.

\bibitem{abs-1810-09232}
A.~Biniaz, S.~Cabello, P.~Carmi, J.~{De Carufel}, A.~Maheshwari, S.~Mehrabi,
  and M.~Smid.
\newblock On the minimum consistent subset problem.
\newblock {\em CoRR} abs/1810.09232, 2018,
  \url{http://arxiv.org/abs/1810.09232}.

\bibitem{Brass08}
P.~Brass.
\newblock {\em Advanced Data Structures}.
\newblock Cambridge University Press, 2008,
  \url{https://doi.org/10.1017/CBO9780511800191}.

\bibitem{Cabello19}
S.~Cabello.
\newblock Subquadratic algorithms for the diameter and the sum of pairwise
  distances in planar graphs.
\newblock {\em {ACM} Trans. Algorithms} 15(2):21:1--21:38, 2019,
  \url{https://doi.org/10.1145/3218821}.

\bibitem{ChaudhuriZ00}
S.~Chaudhuri and C.~D. Zaroliagis.
\newblock Shortest paths in digraphs of small treewidth. {P}art {I:} sequential
  algorithms.
\newblock {\em Algorithmica} 27(3):212--226, 2000,
  \url{https://doi.org/10.1007/s004530010016}.

\bibitem{CdV10}
{\'{E}}.~{Colin de Verdi{\`{e}}re}.
\newblock Shortest cut graph of a surface with prescribed vertex set.
\newblock {\em Algorithms - {ESA} 2010, 18th Annual European Symposium, Part
  {II}}, pp.~100--111. Springer, Lecture Notes in Computer Science 6347, 2010,
  \url{http://dx.doi.org/10.1007/978-3-642-15781-3_9}.

\bibitem{CyganFKLMPPS15}
M.~Cygan, F.~V. Fomin, L.~Kowalik, D.~Lokshtanov, D.~Marx, M.~Pilipczuk,
  M.~Pilipczuk, and S.~Saurabh.
\newblock {\em Parameterized Algorithms}.
\newblock Springer, 2015, \url{http://dx.doi.org/10.1007/978-3-319-21275-3}.

\bibitem{DeyFW15}
T.~K. Dey, F.~Fan, and Y.~Wang.
\newblock Graph induced complex on point data.
\newblock {\em Comput. Geom.} 48(8):575--588, 2015,
  \url{https://doi.org/10.1016/j.comgeo.2015.04.003}.

\bibitem{Diestel12}
R.~Diestel.
\newblock {\em Graph Theory, 4th Edition}.
\newblock Graduate texts in mathematics 173. Springer, 2012.

\bibitem{Erwig00}
M.~Erwig.
\newblock The graph {V}oronoi diagram with applications.
\newblock {\em Networks} 36(3):156--163, 2000,
  \url{http://dx.doi.org/10.1002/1097-0037(200010)36:3<156::AID-NET2>3.0.CO;2-L}.

\bibitem{GawrychowskiKMS18}
P.~Gawrychowski, H.~Kaplan, S.~Mozes, M.~Sharir, and O.~Weimann.
\newblock {V}oronoi diagrams on planar graphs, and computing the diameter in
  deterministic $\tilde {O}(n^{5/3})$ time.
\newblock {\em Proc. 29th {ACM-SIAM} Symposium on Discrete Algorithms, {SODA}
  2018}, pp.~495--514. {SIAM}, 2018,
  \url{http://dx.doi.org/10.1137/1.9781611975031.33}.
\newblock Full version available at \url{http://arxiv.org/abs/1704.02793}.

\bibitem{GawrychowskiMWW18}
P.~Gawrychowski, S.~Mozes, O.~Weimann, and C.~{Wulff{-}Nilsen}.
\newblock Better tradeoffs for exact distance oracles in planar graphs.
\newblock {\em Proc. 29th {ACM-SIAM} Symposium on Discrete Algorithms, {SODA}
  2018}, pp.~515--529, 2018,
  \url{http://dx.doi.org/10.1137/1.9781611975031.34}.

\bibitem{GerbnerMPPR14}
D.~Gerbner, V.~M{\'{e}}sz{\'{a}}ros, D.~P{\'{a}}lv{\"{o}}lgyi, A.~Pokrovskiy,
  and G.~Rote.
\newblock Advantage in the discrete {V}oronoi game.
\newblock {\em J. Graph Algorithms Appl.} 18(3):439--457, 2014,
  \url{https://doi.org/10.7155/jgaa.00331}.

\bibitem{GottliebKN18}
L.~Gottlieb, A.~Kontorovich, and P.~Nisnevitch.
\newblock Near-optimal sample compression for nearest neighbors.
\newblock {\em {IEEE} Trans. Information Theory} 64(6):4120--4128, 2018,
  \href{http://dx.doi.org/10.1109/TIT.2018.2822267}%
{doi:10.1109/TIT.2018.2822267}, \url{https://doi.org/10.1109/TIT.2018.2822267}.

\bibitem{Hart68}
P.~E. Hart.
\newblock The condensed nearest neighbor rule (corresp.).
\newblock {\em {IEEE} Trans. Information Theory} 14(3):515--516, 1968,
  \url{https://doi.org/10.1109/TIT.1968.1054155}.

\bibitem{ImpagliazzoETH}
R.~Impagliazzo and R.~Paturi.
\newblock On the complexity of k-{SAT}.
\newblock {\em J. Comput. Syst. Sci.} 62(2):367--375, 2001,
  \url{https://doi.org/10.1006/jcss.2000.1727}.

\bibitem{Impagliazzo01}
R.~Impagliazzo, R.~Paturi, and F.~Zane.
\newblock Which problems have strongly exponential complexity?
\newblock {\em J. Comput. Syst. Sci.} 63(4):512--530, 2001,
  \url{https://doi.org/10.1006/jcss.2001.1774}.

\bibitem{KannanMZ18}
S.~Kannan, C.~Mathieu, and H.~Zhou.
\newblock Graph reconstruction and verification.
\newblock {\em ACM Trans. Algorithms} 14(4):40:1--40:30, 2018,
  \url{http://doi.acm.org/10.1145/3199606}.
\newblock Preliminary version in ICALP 2015.

\bibitem{KirousisP85}
L.~M. Kirousis and C.~H. Papadimitriou.
\newblock Interval graphs and searching.
\newblock {\em Discrete Mathematics} 55(2):181--184, 1985,
  \url{https://doi.org/10.1016/0012-365X(85)90046-9}.

\bibitem{Lichtenstein82}
D.~Lichtenstein.
\newblock Planar formulae and their uses.
\newblock {\em {SIAM} J. Comput.} 11(2):329--343, 1982,
  \url{https://doi.org/10.1137/0211025}.

\bibitem{Marx10}
D.~Marx.
\newblock Can you beat treewidth?
\newblock {\em Theory of Computing} 6(1):85--112, 2010,
  \url{https://doi.org/10.4086/toc.2010.v006a005}.

\bibitem{MarxP15}
D.~Marx and M.~Pilipczuk.
\newblock Optimal parameterized algorithms for planar facility location
  problems using {V}oronoi diagrams.
\newblock {\em CoRR} abs/1504.05476, 2015,
  \url{http://arxiv.org/abs/1504.05476}.
\newblock Preliminary version at ESA 2015.

\bibitem{Mulzer08}
W.~Mulzer and G.~Rote.
\newblock Minimum-weight triangulation is {NP}-hard.
\newblock {\em J. {ACM}} 55(2):11:1--11:29, 2008,
  \url{http://doi.acm.org/10.1145/1346330.1346336}.

\bibitem{OkabeS12}
A.~Okabe and K.~Sugihara.
\newblock {\em Spatial Analysis along Networks. Statistical and Computational
  Methods}.
\newblock John Wiley and Sons, 2012.

\bibitem{RitterWLI75}
G.~L. Ritter, H.~B. Woodruff, S.~R. Lowry, and T.~L. Isenhour.
\newblock An algorithm for a selective nearest neighbor decision rule
  (corresp.).
\newblock {\em {IEEE} Trans. Information Theory} 21(6):665--669, 1975,
  \url{https://doi.org/10.1109/TIT.1975.1055464}.

\bibitem{SarnakT86}
N.~Sarnak and R.~E. Tarjan.
\newblock Planar point location using persistent search trees.
\newblock {\em Commun. {ACM}} 29(7):669--679, 1986,
  \url{https://doi.org/10.1145/6138.6151}.

\bibitem{SleatorT85}
D.~D. Sleator and R.~E. Tarjan.
\newblock Self-adjusting binary search trees.
\newblock {\em J. ACM} 32(3):652--686, 1985,
  \url{http://doi.acm.org/10.1145/3828.3835}.

\bibitem{Tarjan97}
R.~E. Tarjan.
\newblock Dynamic trees as search trees via euler tours, applied to the network
  simplex algorithm.
\newblock {\em Math. Program.} 77:169--177, 1997,
  \url{https://doi.org/10.1007/BF02614369}.

\bibitem{Yao91}
A.~C. Yao.
\newblock Lower bounds for algebraic computation trees with integer inputs.
\newblock {\em {SIAM} J. Comput.} 20(4):655--668, 1991,
  \url{https://doi.org/10.1137/0220041}.

\end{thebibliography}
\end{document}